\documentclass[sigconf]{acmart}
\usepackage{amsfonts,amsmath,amsthm,amscd,dsfont,mathrsfs,mathtools}

\usepackage{graphicx}
\usepackage{appendix}
\usepackage{float}	
\usepackage{xcolor}
\usepackage{amsthm}
\usepackage{geometry,enumitem}
\usepackage[subtle]{savetrees}

\usepackage{comment}
\usepackage{graphicx,subcaption}
\usepackage{varwidth}
\usepackage{url}
\usepackage{xcolor}

\newtheorem{thm}{Theorem}[section]
\newtheorem{defn}[thm]{Definition}

\newtheorem{prop}{Proposition}
\newtheorem{lem}[thm]{Lemma}

\newtheorem{nota}[thm]{Notation}

\newcommand{\bpf}{\begin{proof}}
\newcommand{\epf}{\end{proof}}

\newcommand{\red}[1]{\textcolor{red}{#1}}

\newcommand{\beqa}{\begin{eqnarray}}

\newcommand{\eeqa}{\end{eqnarray}}

\newcommand{\beqan}{\begin{eqnarray*}}

\newcommand{\eeqan}{\end{eqnarray*}}

\newcommand{\beq}{\begin{equation}}

\newcommand{\eeq}{\end{equation}}

\newcommand{\prob}{{\mathbb {P}}}

\newcommand{\F}{{\mathcal F}}

\newcommand{\singlespacing}{\let\CS=\@currsize\renewcommand{\baselinestretch}{0.95}\tiny\CS}

\newcommand{\oneandahalfspacing}{\let\CS=\@currsize\renewcommand{\baselinestretch}{1.25}\tiny\CS}

\newcommand{\doublespacing}{\let\CS=\@currsize\renewcommand{\baselinestretch}{1.39}\tiny\CS}

\newcommand{\be}{\begin{equation}}

\newcommand{\ee}{\end{equation}}

\newcommand{\bc}{\begin{center}}

\newcommand{\ec}{\end{center}}

\newcommand{\bfl}{\begin{flushleft}}

\newcommand{\efl}{\end{flushleft}}

\newcommand{\scheme}{{\sf Prism} }
\newcommand{\schemenosp}{{\sf Prism}\ignorespaces}

\newcommand{\fruitchains}{{\sf FruitChains } }
\newcommand{\fruitchainsnosp}{{\sf FruitChains}\ignorespaces}

\newcommand{\bitcoin}{{\sf Bitcoin} }
\newcommand{\bitcoinnosp}{{\sf Bitcoin}\ignorespaces}

\newcommand{\ohie}{{\sf OHIE} }
\newcommand{\ohienosp}{{\sf OHIE}\ignorespaces}

\makeatletter
\renewcommand\subsubsection{\@startsection{subsubsection}{3}{\z@}%
                       {-18\p@ \@plus -4\p@ \@minus -4\p@}%
                       {4\p@ \@plus 2\p@ \@minus 2\p@}%
                       {\normalfont\normalsize\bfseries\boldmath
                        \rightskip=\z@ \@plus 8em\pretolerance=10000 }}

\newcommand{\Z}{{\mathcal{Z}}}

\renewcommand{\H}{{\mathcal{H}}}
\renewcommand{\C}{{\mathcal{C}}}
\newcommand{\A}{{\mathcal{A}}}

\newcommand{\rmax}{r_{\rm max}}

\makeatother

\renewcommand{\epsilon}{\varepsilon}

\usepackage{algorithm}
\usepackage[noend]{algpseudocode}
\usepackage{xcolor}
\definecolor{azure}{rgb}{0.54, 0.17, 0.89}
\definecolor{newcode}{RGB}{201,28,28}
\newcommand{\colorcomment}[1]{\Comment{ {\color{azure} #1}} }
\newcommand{\av}[1]{$#1$}
\newcommand{\maincolorcomment}[1]{{\color{azure}// #1 } }

\usepackage{enumitem,kantlipsum}

\begin{document}

\title{Securing  Parallel-chain Protocols under Variable Mining Power}

%


\author{Xuechao Wang}
\email{xuechao2@illinois.edu}
\affiliation{%
  \institution{University of Illinois Urbana-Champaign}
}
\author{Viswa Virinchi Muppirala}
\email{virinchi@uw.edu}
\affiliation{%
  \institution{University of Washington at Seattle}
}
\author{Lei Yang}
\email{leiy@csail.mit.edu}
\affiliation{%
  \institution{MIT CSAIL}
}
\author{Sreeram Kannan}
\email{ksreeram@uw.edu}
\affiliation{%
  \institution{University of Washington at Seattle}
}
\author{Pramod Viswanath}
\email{pramodv@illinois.edu}
\affiliation{%
  \institution{University of Illinois Urbana-Champaign}
}
\thanks{Correspondence can be sent to ksreeram@uw.edu.}


\begin{abstract}
Several emerging proof-of-work (PoW) blockchain protocols rely on a ``parallel-chain'' architecture for scaling, where instead of a single chain, multiple chains are run in parallel and aggregated. A key requirement of practical PoW blockchains is to adapt to mining power variations over time (\bitcoinnosp's total mining power has increased by a $10^{14}$ factor over the decade). In this paper, we consider the design of provably secure parallel-chain protocols which can adapt to such mining power variations. 

The \bitcoin difficulty adjustment rule adjusts the difficulty target of block mining periodically to get a constant mean inter-block time. While superficially simple, the rule has proved itself to be  sophisticated and successfully secure, both in practice and in theory \cite{backbone,full2020}.   We show that natural adaptations of the \bitcoin adjustment rule to the parallel-chain case  open the door  to subtle, but catastrophic  safety and liveness breaches. We uncover a   meta-design principle that allow us to design variable mining difficulty protocols for three popular PoW blockchain proposals  (\scheme \cite{bagaria2019prism}, \ohie \cite{yu2020ohie}, Fruitchains \cite{pass2017fruitchains}) inside a common rubric. 

The principle has three components: (M1) a pivot chain, based on which blocks in all chains choose difficulty, (M2) a monotonicity condition for referencing pivot chain blocks and (M3) translating additional protocol aspects from using levels (depth) to using ``difficulty levels''. We show that protocols employing a subset of these principles may have catastrophic failures. The security of the designs is also proved using a common rubric -- the key technical challenge involves analyzing the interaction between the pivot chain and the other chains, as well as bounding the sudden changes in difficulty target experienced in non-pivot chains. We empirically investigate the responsivity of the new mining difficulty rule via simulations based on historical \bitcoin  data, and find that the protocol very effectively controls the forking rate across all the chains. 

\end{abstract}
{\def\addcontentsline#1#2#3{}\maketitle}

\section{Introduction}
\label{sec:intro}

{\bf Scaling problem.} Built on the pioneering work of Nakamoto, \bitcoin~\cite{nakamoto2008bitcoin} is a permissionless blockchain operating on proof-of-work based on the Nakamoto protocol. The Nakamoto longest-chain protocol was proven to be secure  as long as the adversary controlled less than $50\%$ of the mining power in the breakthrough work~\cite{backbone}. Recent works~\cite{eyal2016bitcoin,sompolinsky2015secure,li2018scaling} have tried to improve the scalability of \bitcoin \cite{croman2016scaling,bonneau2015sok}, in particular the throughput and latency, by redesigning the core consensus protocol. A variety of approaches have been proposed, for example hybrid consensus algorithms~\cite{schwartz2014ripple, gilad2017algorand,mazieres2015stellar,pass2017hybrid} try to convert the permissionless problem into a permissioned consensus problem by subselecting a set of miners from a previous epoch. While such approaches achieve scalability, they are not natively proof-of-work (PoW) and hence do not retain the dynamic availability, unpredictability and security against adaptive adversaries that the Nakamoto longest chain protocol enjoys. 

\noindent{\bf Parallel-chain protocols.} An emerging set of proof-of-work protocols maintain the native PoW property of \bitcoin and achieve provable scaling by using {\em many parallel} chains. The chains run in parallel and use an appropriate aggregation rule to construct an ordered ledger of transactions out of the various parallel chains.  We will highlight three examples of parallel-chain protocols (PCP): (1) \scheme~\cite{bagaria2019prism}, which achieves high-throughput and low-latency using a proposer chain and many voter chains, (2) \ohie~\cite{yu2020ohie}, which achieves high-throughput using parallel chains and (3) \fruitchains~\cite{pass2017fruitchains}, which achieves fairness using two distinct types of blocks (blocks and fruits) mined in parallel. There are other approaches such as ledger combiners~\cite{fitzi2020ledger}, which achieve some of the same goals using different architectures.

\noindent{\bf Common structure of PCP.} In all of these parallel-chain protocols (PCP), there are multiple types of blocks (for example, in \ohienosp, each type may correspond to a different chain) and we determine the final type only after mining the block - we will term this process as \textit{hash sortition}. The idea of sortition was first formalized in~\cite{backbone} called 2-for-1 PoW. All three PCPs adopt this technique to achieve parallel mining. A miner creates a single commitment (for example, a Merkle root) to the potential version of the different block types and performs a mining operation. Depending on the region the hash falls, the block is considered mined of a certain type. Different protocols utilize different types of aggregation rules and semantics in order to consider the final ledger out of these parallel chains.

\noindent{\bf Variable mining power problem.} A key requirement of deployed PoW blockchains is to adapt to the immense variation in mining power. For example, the mining power of \bitcoin increased exponentially by an astonishing factor of $10^{14}$ during its decade of deployment. If \bitcoin had continued to use the same difficulty for the hash puzzle, then the inter-block time would have fallen from the original $10$ minutes to {\bf $6$ picoseconds}. Such a drop would have caused an intolerable forking rate and seriously undermined the security of \bitcoinnosp, lowering the tolerable adversarial mining power from nearly $50\%$ to $10^{-11}$. However, this is prevented by adjusting the difficulty threshold of \bitcoin using a difficulty adjustment algorithm.

\noindent{\bf Bitcoin difficulty adjustment algorithm.} There are three core ideas to the \bitcoin difficulty adjustment algorithm: (a) vary the difficulty target of block mining based on the median inter-block time from the previous epoch (of $2016$ blocks), (b) use the {\em heaviest} chain (calculated by the sum of the block difficulties) instead of the longest chain to determine the ledger, and (c) allow the difficulty to be adjusted only mildly every epoch (by an upper bound of a factor of $4$). While this appears to be a simple and intuitive algorithm, minor seemingly-innocuous variants turn out to be dangerously insecure.

\noindent{\bf Difficulty adjustment terminology.} Throughout the paper, we call the hash puzzle threshold in PoW mining the {\em target} of a block. The {\em block difficulty} of each block is measured in terms of how many times the block is harder to obtain than using the initial target of the system that is embedded in the genesis block. However, for simplicity, we will adapt the notation of block difficulty to be the inverse of the target of the block. The {\em chain difficulty} of a chain is the sum of block difficulties of all blocks that comprise the chain, then each block in the chain {\em covers} an interval of chain difficulty. The chain with the largest chain difficulty is said to be the {\em heaviest} chain. We also refer the chain difficulty of a block as the chain difficulty of the chain ending at this block. This notation is summarised in the following table.

\begin{center}
    \begin{tabular}{ | p{2cm} | p{5.7cm}|  } 
        \hline
        Target &  Threshold of the hash puzzle in PoW mining\\ 
        \hline
        Block difficulty & Inverse of the target of a block \\ 
        \hline
        Chain difficulty & Sum of block difficulties of all blocks in the chain \\
\hline
    \end{tabular}
\end{center}

\noindent{\bf Difficulty adjustment requires nuanced design.} Consider a simpler algorithm using only (b), i.e., simply let the nodes choose their own difficulty and then use (b) the heaviest chain rule. At a first glance, this rule appears kosher - the heaviest chain rule seems to afford no advantage to any node to manipulate their difficulty. However, this lack of advantage only holds in expectation, and the variance created by extremely difficult adversarial blocks can thwart a confirmation rule that confirms deeply-embedded blocks, no matter how deep, with non-negligible probability proportional to the attacker's mining power (refer to Appendix~\ref{app:attack} for a detailed discussion). Now consider a more detailed rule involving only (a) and (b). It turns out that there is a difficulty raising attack (refer to Appendix~\ref{app:attack} for a detailed discussion), where the adversary creates an epoch filled with timestamps extremely close-together, so that the difficulty adjustment rule from (a) will set the difficulty extremely high for the next epoch, at which point, the adversary can utilize the high variance of the mining similar to the aforementioned attack. This more complex attack is only thwarted using the full protocol that employs (a), (b) and (c) together. The full proof of the Nakamoto heaviest chain protocol was obtained in a breakthrough work \cite{garay2017bitcoin}.

\noindent{\bf Difficulty adjustment in PCP.}
When there are multiple parallel-chains, one  natural idea is to apply \bitcoinnosp's difficulty adjustment algorithm to each of the chains independently. However, this idea does not integrate well with hash sortition since the range of a particular chain will depend on the state of other chains. Instead, since the mining power variation is the same across all chains, a natural approach is to use {\em the same difficulty threshold} across all chains, which is then modulated based on past evidence. How should this common difficulty threshold be chosen? One approach is to utilize inter-block arrival times across all the chains to get better statistical averaging and respond faster to mining power variation. However, it requires some sort of synchronization across the chains and breaks the independence assumption.  



\noindent{\bf General methodology.} We propose a general methodology by which to adapt parallel-chain architectures to the variable mining rate problem. Our general methodology is comprised of three parts, as detailed below. 
\begin{itemize}
    \item{\bf M1: Pivot-chain.} Use a single chain as the pivot chain for difficulty adjustment. Blocks mined in any other chain need to refer to a block in the pivot chain and use the target inferred therefrom. 
    \item{\bf M2: Monotonicity.} In a non-pivot chain, blocks can only refer to pivot-chain blocks of non-decreasing chain difficulty. 
    \item{\bf M3: Translation.}  Wherever the protocol uses the concept of a block's level, it is updated to refer to the block's chain difficulty instead.
\end{itemize}
Using \textbf{M1} pivot-chain for difficulty adjustment ensures that we can continue to use the hash-sortition method. The \textbf{M2} monotonicity rule ensures that blocks in non-pivot chain do not refer to stale/old pivot blocks with target which is very different from expected in the present round. Finally, the \textbf{M3} translation rule ensures that other aspects of the protocol, such as the confirmation rule are adapted correctly to deal with the variable difficulty regime correctly. We show in Section~\ref{sec:prism} why each of the three aspects of our methodology is critical in designing variable difficulty for \scheme by showing attacks for subsets of \textbf{M1,M2,} and \textbf{M3}. 

On the positive side, we show a concrete adaptation of our general methodology to various schemes, in particular to \scheme in Section~\ref{sec:prism}, to \ohie in Section~\ref{sec:ohie} and to \fruitchains in Section~\ref{sec:fruit}.

\noindent{\bf Security proofs.} The problem of analyzing the difficulty adjustment mechanism in \bitcoin was first addressed in \cite{garay2017bitcoin} in
the lock-step synchronous communication model. It introduces a setting where the number of participating parties' rate of change in a sequence of rounds is bounded but follows a predetermined schedule. Later two concurrent works~\cite{full2020,chan2020Varying} analyzed the problem in a bounded-delay network with an adaptive (as opposed to predetermined) dynamic participation, with different proof techniques. 
Following the two later papers, we adopts the more general network and adversary models: we assume a $\Delta$-synchronous communication model, where every message that is received by a honest node is received by all other honest nodes within $\Delta$ rounds; we allow the adversary to control the mining rate even based on the stochastic realization of the blockchain, as long as the mining rate does not change too much in a certain period of time. We assume that the adversarial nodes are Byzantine and they do not act rationally. Under this general model, we establish that our proposed modification to \schemenosp, \ohie and \fruitchains satisfy the dual security properties of safety and liveness. The proofs require a new understanding of how difficulty evolution in a non-pivot chain progresses based on the difficulty in the pivot chain - this statistical coupling presents a significant barrier to surmount in our analysis, and differs from previous work in this area. We show these results in Section~\ref{sec:proof}.

\noindent{\bf Systems implementation.} Our variable difficulty scheme does not add significant computation and communication overhead on existing parallel-chain protocols, making our protocol an easy upgrade. We conduct extensive simulation studies to examine how our systems respond to varying mining power. Results show that our scheme is able to closely match the system mining power and the mining difficulty for each individual chain, thus keeping the chain forking rate stable. We examine adversarial behavior and how it can influence the difficulties of various chains, and confirm that our scheme is secure against significant adversarial presence. The simulations are based on historical Bitcoin mining power data and parameters collected from real-world experiments of the \scheme \cite{prism-system} parallel-chain protocol, making the insights meaningful for real-world systems. 


\noindent{\bf Other related works.} A recently proposed blockchain protocol {\sf Taiji}~\cite{li2020taiji} combines  \scheme with a BFT protocol to construct a dynamically available PoW protocol which has almost deterministic confirmation with low latency. Since {\sf Taiji} inherits the parallel-chain structure from \schemenosp, our meta-principles will also apply. The vulnerability of selfish mining has recently been discussed on several existing blockchain projects with variable difficulty in \cite{negy2020selfish}. Our proposed variable difficulty \fruitchains protocol guarantees fairness of mining, thus disincentivizes selfish mining.  

\section{Model}
\label{sec:model}

\noindent{\bf Synchronous network.} We describe our protocols in the now-standard  $\Delta$-synchronous network model considered in \cite{pass2017analysis,full2020,badertscher2018ouroboros} for the analysis of proposed variable difficulty protocols, where there is an upper bound $\Delta$ in the delay (measured in number of rounds) that the adversary may inflict to the delivery of any message. Observe that notion of ``rounds'' still exist in the model (since we consider discretized time), but now these are not synchronization rounds within which all messages are supposed to be delivered to honest parties.

Similar to  \cite{full2020,pass2017analysis}, the protocol execution proceeds in ``round'' with inputs provided by an environment program denoted by $\Z(1^\kappa)$ to parties that execute the protocol $\Pi$, where $\kappa$ is a security parameter. The adversary $\A$ is adaptive, and allowed to take control of parties on the fly, as well as ``rushing'', meaning that in any given round the adversary gets to observe honest parties’ actions before deciding how to react. The network is modeled as a diffusion functionality similar to those in \cite{full2020,pass2017analysis}: it allows order of messages to be controlled by $\A$, i.e., $\A$ can inject messages for selective delivery but cannot change the contents of the honest parties’ messages nor prevent them from being delivered beyond $\Delta$ rounds of delay — a functionality parameter.

\noindent{\bf Random oracle.} We abstract the hash function as a random oracle functionality. It accepts queries of the form $(\texttt{compute}, x)$ and $(\texttt{verify}, x, y)$. For the first type of query, assuming $x$ was never queried before, a value $y$ is sampled from $\{0, 1\}^\kappa$ and it is entered to a table $T_H$. If $x$ was queried before, the pair $(x, y)$ is recovered from $T_H$. In both cases, the value $y$ is provided as an answer to the query. For the second type of query, a lookup operation is performed on the table. Honest parties are allowed to ask one query per round of the type \texttt{compute} and unlimited queries of the type \texttt{verify}. The adversary $\A$ is given a bounded number of \texttt{compute} queries per round and also unlimited number of \texttt{verify} queries. The bound for the adversary is determined as follows. Whenever a corrupted party is activated the bound is increased by 1; whenever a query is asked the bound is decreased by 1 (it does not matter which specific corrupted party makes the query).

\noindent{\bf Adversarial control of variable mining power.} We assume no rational node in the adversarial model. The adversary can decide on the spot how many honest parties are activated adaptively. In a round $r$, the number of honest parties that are active in the protocol is denoted by $n_r$ and the number of corrupted parties controlled by $\A$ in round $r$ is denoted by $t_r$. Note that $n_r$ can only be determined by examining the view of all honest parties and is not a quantity that is accessible to any of the honest parties individually. We make the ``honest majority'' assumption, i.e., $t_r < (1-\delta) n_r$ for all $r$, where the positive constant $\delta < 1$ is the advantage of honest parties. Further, we will restrict the environment to fluctuate the number of parties in a certain limited fashion. Suppose $\Z, \A$ with fixed coins produces a sequence of parties $n_r$, where $r$ ranges over all rounds of the entire execution, we define the following notation.

\begin{defn}
\label{def:respecting}
Let $\rmax \in \mathbb{N}$ is the total number of rounds in the execution. For $\gamma \in \mathbb{R}^+$, we call $(n_r), r \in [0,\rmax]$, as $(\gamma, s)$-respecting if for any set $S \subseteq [0, \rmax]$ of at
most $s$ consecutive rounds, 
$$\max_{r \in S}\  n_r \leq \gamma \min_{r\in S} \  n_r.$$ 
We say that $\Z$ is $(\gamma, s)$-respecting if for all $\A$ and coins for $\Z$ and $\A$ the sequence of honest parties $n_r$ is $(\gamma, s)$-respecting.
\end{defn}

\section{Prism}
\label{sec:prism}

\subsection{Fixed Difficulty Algorithm}
Each block in the longest chain of the \bitcoin protocol performs dual roles: Proposing and Voting. A proposed block gets confirmed with high reliability only after the block, and several more blocks extending it make it in the longest chain. The latency of the protocol is the number of blocks for which one needs to wait (this number depends on the reliability). To guarantee security, the mining rate remains low~\cite{backbone}, which leads to low throughput and high latency. 
\scheme \cite{bagaria2019prism} is a Proof-of-Work protocol that decouples block proposals and voting to scale throughput and latency. We will briefly explain the \scheme as it was originally described in the fixed mining rate, i.e., fixed difficulty regime. As show in Figure \ref{fig:prism_fig}, \scheme runs multiple $m+1$ separate parallel ``blocktrees'' where one of the trees, called the proposer tree, consists of blocks from which the final transaction ledger is constructed. We define a block's level in the proposer tree as the block's depth from the genesis. The final transaction ledger will comprise of one proposer block at each level chosen by the longest chain in each of the $m$ voter blocktrees; they are referred to as voter chains. A voter block in any voter blocktree can vote for one or more proposer blocks at different levels by including a pointer to the corresponding proposer blocks in its payload. A voter block can also consist of a null vote if the chain it entered already voted on the latest level. At a given level of the proposer tree, a voter chain can vote for exactly one proposer block. The net vote for a proposer block can be counted by aggregating which of the $m$ voter chains voted for that block. The block with the most votes at any particular level is termed as the leader block for that level, and the ledger is constructed by concatenating the leader blocks at various levels.  

\noindent{\bf Mining and sortition.} In order to ensure that the adversary cannot focus the mining power onto a single chain, a ``sortition'' mechanism is used. A miner creates a ``super-block'' containing information about its parent block in each of the $m+1$ trees. Each tree has a target of $T$, and for $1\leq i \leq m$ if a node creates a block of hash value in between $iT$ and $(i+1)T$, it will be able to mine this block in the voter block-tree $i$. If it creates a block of hash value less than $T$, it will be able to mine this block in the proposer block-tree as show in Figure~\ref{fig:prism_fig}. 

\begin{figure}
    \centering
    \includegraphics[width=8cm]{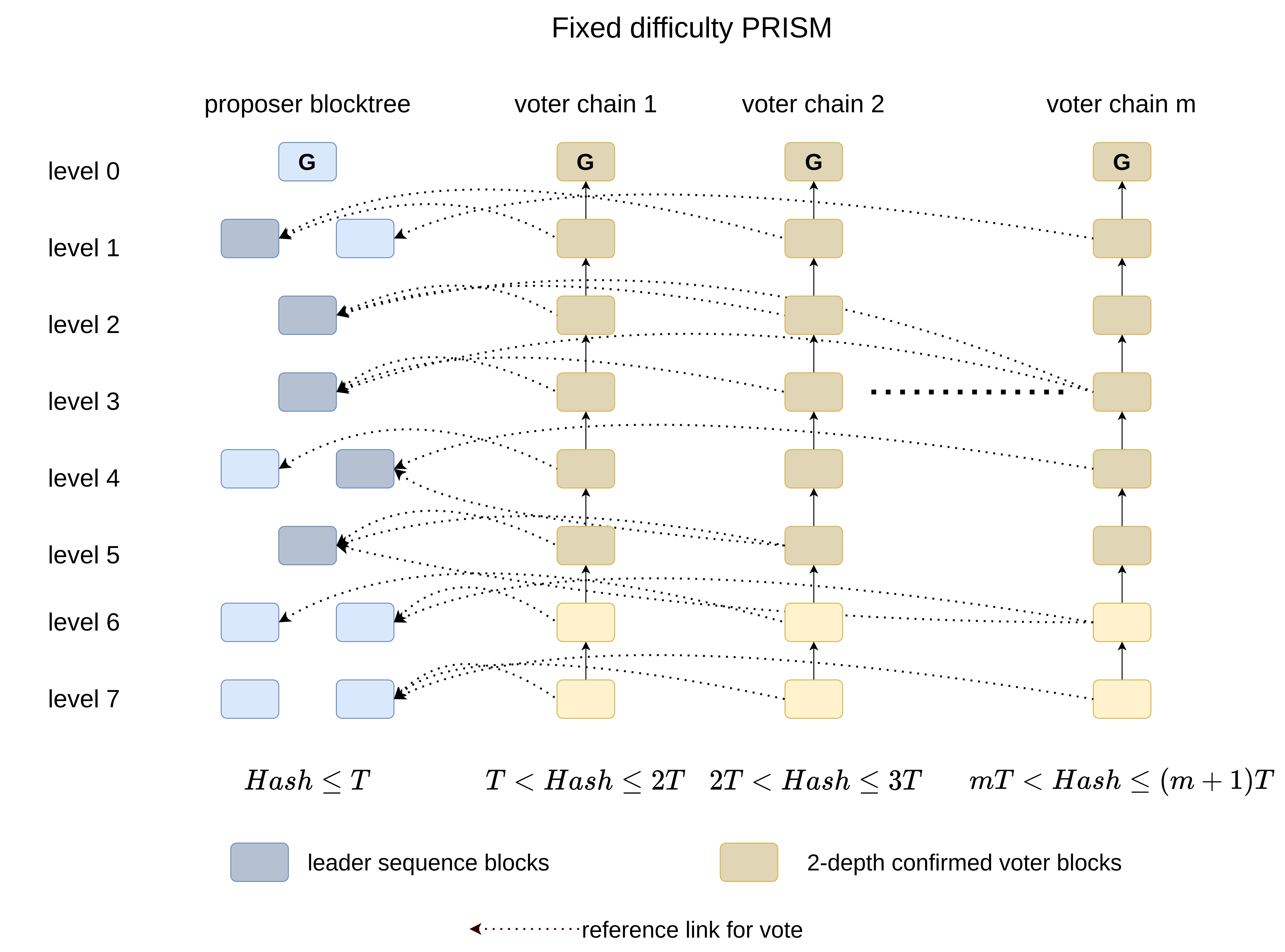}
    \vspace{-0.1in}
    \caption{Fixed difficulty \schemenosp. Snapshot of a miner's view of \scheme block-trees. The confirmed blocks are darker in color and the votes are shown using dotted arrows.} 
    \label{fig:prism_fig}
    \vspace{-0.2in}
\end{figure}

The structure of this voting scheme in \scheme enables low confirmation latency. The high-level idea is that votes accrue only sequentially in \bitcoinnosp, whereas in \schemenosp, votes accrue parallelly. Thus for a given amount of security, confirmation only needs to wait for a much shorter amount of time (since voting blocks are created in parallel), thus reducing the amount of latency. 



\subsection{Natural Approaches Are Insecure}
\label{sec:strawman}

{\noindent \bf Different difficulty adjustment for different chains (no M1).} To add support for variable mining power to \schemenosp, a natural first approach is to replace the longest chain rule \cite{backbone} by the heaviest chain rule \cite{full2020} in all the parallel chains, and adjust the mining difficulty in each chain separately. However, miners in \scheme use cryptographic sortition to mine blocks on all chains at the same time, and having different thresholds for different chains depending on the state will require complex coupling across chains. Furthermore, since the mining power variation is the same across different chains, it is efficient to have a single difficulty threshold across the entire system. 

As we explained in the introduction, our general methodology for converting a fixed-difficulty protocol into a variable-difficulty protocol  comprises of three attributes \textbf{M1}: Pivot-chain, \textbf{M2}: Monotonicity and \textbf{M3}: Translation. We will now explore the subtleties inherent in this process and show why a subset of these attributes is insufficient for \schemenosp.


{\noindent \bf M1 without M2 $\Rightarrow$ Safety failure.} To make the cryptographic sortition technique applicable, a straightforward approach is to use the proposer chain as a pivot chain for difficulty adjustment (\textbf{M1}): the difficulty of the proposer blocks is adjusted according to the \bitcoin rule \cite{full2020}, and the difficulty of voter blocks tracks that of the proposer chain by reference links. However, if we allow the miners to use the difficulty of any proposer block for the voter blocks, then {\bf safety failures} may occur on the voter chains. We demonstrate an example safety failure here. Let the honest parties maintain the difficulty $d_0$ throughout the execution, and the adversary mines private proposer blocks with timestamps in rapid succession to increase the difficulty to $d_0*k$, where $k$ is a desired security parameter on the voter chains (i.e., we hope that a $k$-deep voter block will be stable forever). Even if the adversary cannot keep up the chain difficulty of its private proposer chain with the heaviest public chain, at the current time on the voter chain, the adversary will refer to this very difficult block on the proposer tree to create a very difficult block on the voter tree. If the adversary is lucky and mines one voter block with difficulty $d_0*k$, the probability of which is a constant rather than exponentially decaying in $k$ (via the same anti-concentration argument in Appendix~\ref{app:attack}), then it can overtake the heaviest voter chain and reverse a $k$-deep voter block. This attack is described in Figure ~\ref{fig:attack1}. To address this issue, we require that on each voter chain the referred proposer blocks should have non-decreasing chain difficulty (\textbf{M2}), so that the adversary can no longer adopt an old mining difficulty from the proposer chain. With {\bf M2}, although the adversary may not refer to the tip of the proposer chain, both our analysis in Section \ref{sec:non-pivot} and the simulation in Section \ref{sec:exp-attack} show that the security of the voter chain can still be guaranteed.
\begin{figure}
    \centering
    \includegraphics[width=6cm]{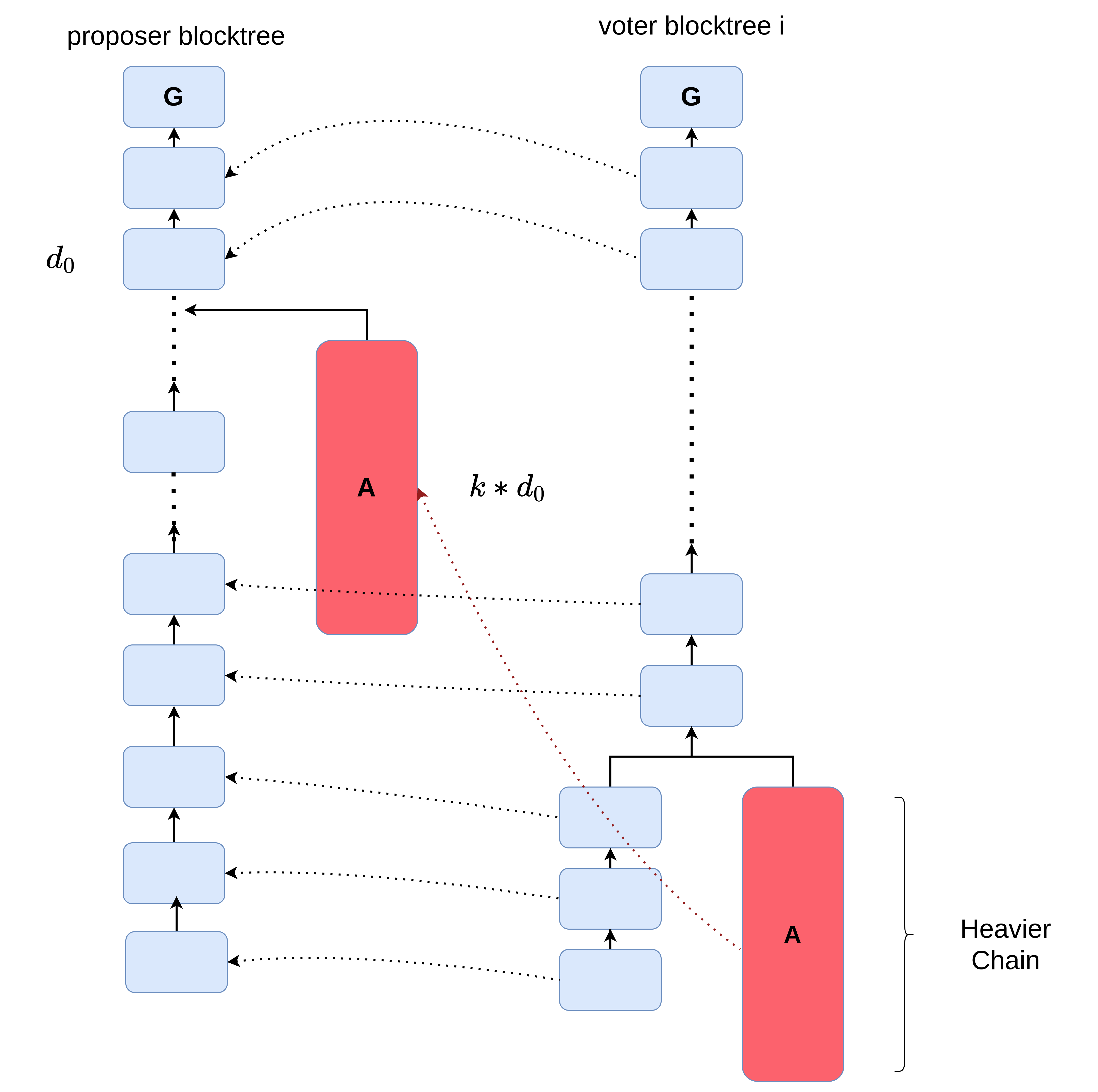}
    \vspace{-0.15in}
    \caption{Attack on safety when we enforce M1 but not M2. At the current time, the adversary will choose a very difficult proposer block in a less heavier chain as its proposer-parent in the voter chain, hurting the ledger's security. The dotted arrows represent the relation between the voter blocks and their proposer parents.} 
    \label{fig:attack1}
    \vspace{-0.15in}
\end{figure}

{\noindent \bf Voting rule: No M3 $\Rightarrow$ Liveness Failure.} In fixed difficulty \schemenosp, a voter block votes on all levels in the proposer tree that are unvoted by the voter block’s ancestors. In the variable difficulty algorithm, while the notion of ``level'' on the proposer chain is well-defined an adversary can always mine a very long but easy proposer chain. As a result, if we still order proposer blocks by level, the leader sequence will be full of adversarial blocks, which may cause liveness failure as described in Figure ~\ref{fig:attack2}. A natural generalization would be that voter blocks vote for each difficulty value rather than level and the leader sequence is also decided for each difficulty value (\textbf{M3}). See the complete algorithm in the next subsection.
\begin{figure}
    \centering
    \includegraphics[width=5cm]{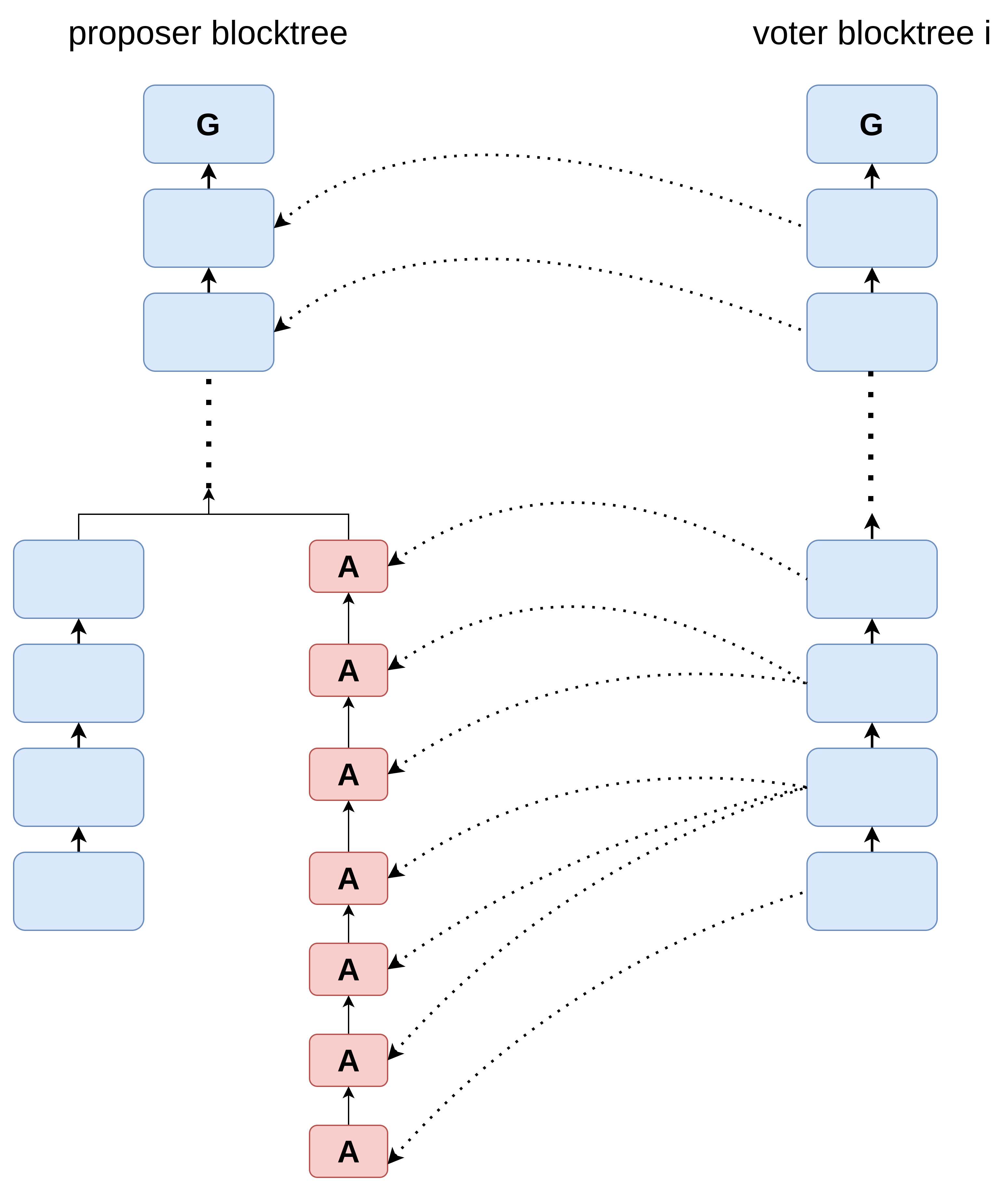}
    \vspace{-0.2in}
    \caption{Attack on liveness when M1, M2 are enforced but not M3. The adversary lowers the block difficulty and advances in level on the proposer tree. The voter blocks will not be able to vote for any honest blocks, hurting the ledger's liveness.}
    \label{fig:attack2}
    \vspace{-0.2in}
\end{figure}

\subsection{Variable Difficulty Algorithm}
\label{sec:prism-full-cfm-rule}

We now describe the full \scheme protocol for the variable difficulty setting constructed using our general methodology. We refer the reader to Appendix~\ref{app:pseudocode} for a pseudocode of the algorithm. There are two types of blocks in \scheme blockchain: proposer blocks and voter blocks. Proposer blocks contain transactions that are proposed to be included in the ledger, and constitutes the skeleton of \scheme blockchain.
Voter blocks are mined on $m$ separate voter blocktrees, each with its own genesis block. We say a voter block votes on a proposer block $B$ if it includes a pointer to $B$ in its payload.

{\bf \noindent Block proposal rule.}
The proposer chain follows the heaviest chain rule, and the difficulty adjustment uses the target calculation function defined in \cite{full2020} with parameter $\Phi$ and $\tau$, where $\Phi$ is the length of an epoch in number of blocks and $\tau \geq 1$ is the dampening filter (line \ref{code:difficult} in Algorithm \ref{alg:prism_minining}). All $m$ voter chains also follow the heaviest chain rule, but the difficulty adjustment on voter chains is more tricky and we will discuss it soon when introducing {\it sortition}.

Whereas \bitcoin miners mine on a single blocktree, \scheme miners simultaneously mine one proposer block and $m$ voter blocks via cryptographic sortition. 
More precisely, while mining, each miner selects $m+1$ parent blocks, which are the tips of the heaviest chains on the proposer tree and the $m$ voter trees. We call these tips {\it proposer parent} and {\it voter parents} separately. 
And the miner maintains outstanding content for each of the $m+1$ possible mined blocks: For the proposer block, the content is a list of transactions; For the voter block on the $i$-th voter tree, the content is a list of hashes of proposer blocks at each difficulty in the  proposer blocktree that has not yet received a vote in the heaviest chain of the $i$-th voter tree. More precisely, on the $i$-th voter tree, if the last proposer block voted by the heaviest chain covers the difficulty interval $(a_0,b_0]$ and the proposer parent covers $(a^*,b^*]$, then a valid voter block on the $i$-th voter tree must satisfy the following conditions (see Algorithm \ref{alg:prism_fucntion}).

\begin{itemize}
    \item If $b^* = b_0$, then it should contain no vote.
    \item If $b^* > b_0$, then it should vote for an arbitrary number of proposer blocks $B_1, B_2, \cdots, B_n$, each covering $(a_1,b_1],$ $(a_2,b_2],$ $\cdots,$ $(a_n, b_n]$, such that $a_i < b_{i-1} <b_i$ for all $1 \leq i \leq n$ and $b_n = b^*$.  
\end{itemize}

\begin{figure}
    \centering
    \includegraphics[width=8cm]{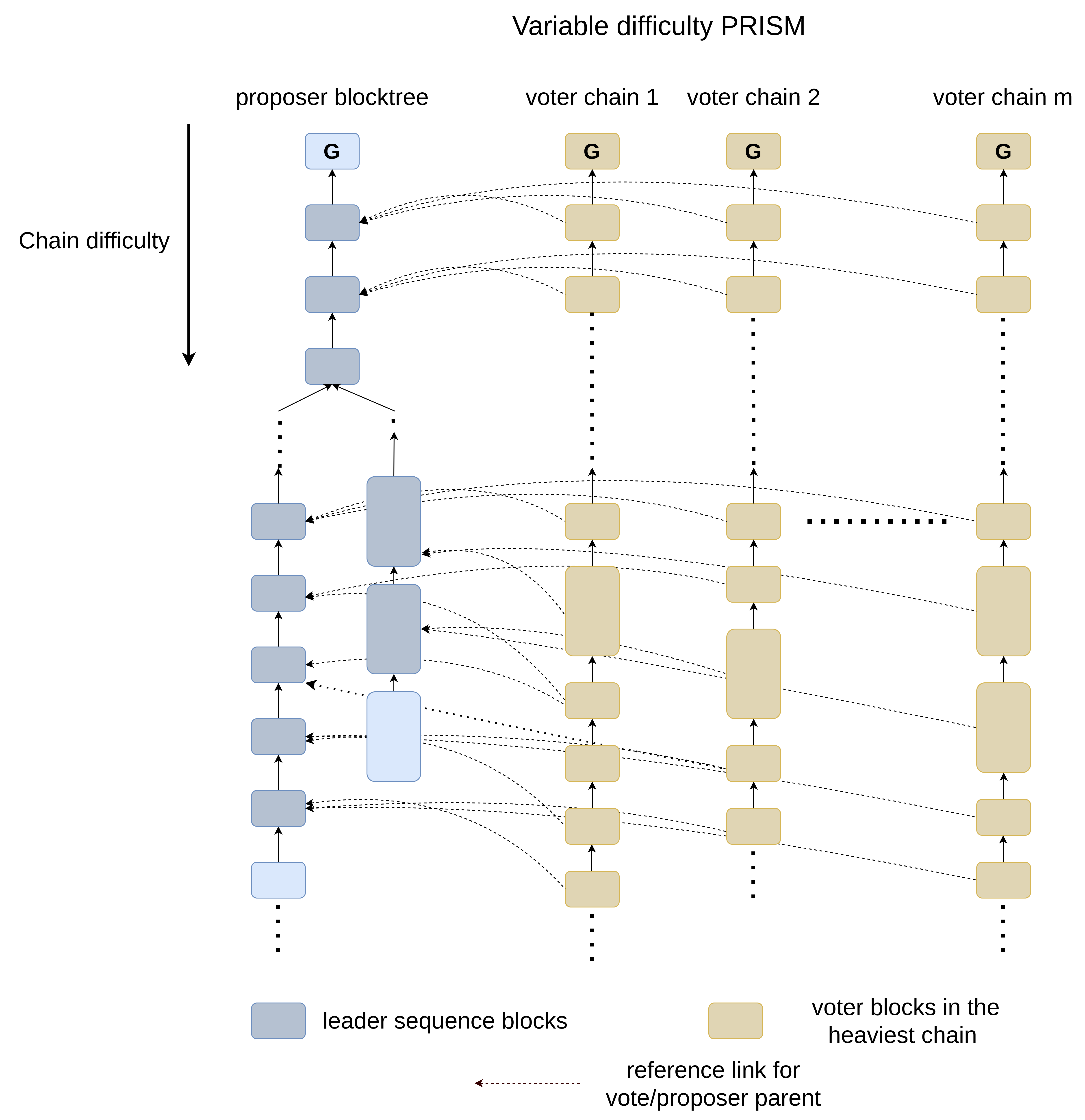}
    \vspace{-0.2in}
    \caption{A miner's view of \scheme block-trees with variable chain difficulty. The confirmed blocks are darker in color. Both the votes and proposer-parent links are shown using the same dotted arrows.}
    \label{fig:variable_prism}
    \vspace{-0.25in}
\end{figure}

Upon collecting this content, the miner tries to generate a block with target according to the proposer parent via proof-of-work (\textbf{M1}).  
Once a valid nonce is found, the output of the hash is deterministically mapped to either a voter block in one of the $m$ trees or a proposer block (lines \ref{code:sortitionStart}-\ref{code:sortitionEnd} in  Algorithm \ref{alg:prism_minining}).

While mining, nodes may receive blocks from  the network, which are processed in much the same way as \bitcoinnosp. For a received voter block to be valid, the chain difficulty of its proposer parent must be at least that of the proposer parent of its voter parent (\textbf{M2}).
Upon receiving a valid voter block, the miner updates the heaviest chain if needed, and updates the vote counts accordingly.
Upon receiving a valid proposer block $B$ with chain difficulty higher than the previous heaviest chain, the miner makes $B$ the  new proposer parent, and updates all $m$ voter trees to vote for chain difficulties until $B$.

{\bf \noindent Ledger formation rule.} 
Note that all the voters on one voter chain may cover overlapping intervals. So we first sanitize them into disjoint intervals: For $n$ consecutive valid votes $(a_1,b_1],$ $(a_2,b_2],$ $\cdots,$ $(a_n, b_n]$ on a voter chain, we sanitize them into new intervals $(a_1,b_1], (b_1,b_2], \cdots, (b_{n-1}, b_n]$. In this way, we make sure that each real-valued difficulty $d$ is voted at most once by each voter chain, hence $d$ can receive at most $m$ votes. Since voter blocks vote for each difficulty value rather than level, the ledger is also generated based on difficulty values (\textbf{M3}). Let $v_i(d)$ be the proposer block with interval containing $d$ voted by the heaviest chain on the $i$-th voter tree. Let $\ell(d)$ be the leader block of difficulty $d$, which is the plurality of the set $\{v_i(d)\}_{i=1}^m$. For each proposer block $B_p$ in the proposer tree, define $g(B_p)$ as
\begin{equation}
\label{eqn:grade}
    g(B_p) = \inf_{d \geq 0}\{d: \ell(d) = B_p \}.
\end{equation}
Note that if $\{d: \ell(d) = B_p \}$ is empty, then $g(B_p) = \infty$. Finally, by sorting all proposer blocks by $g(\cdot)$, we get the leader sequence of the proposer blocks. A concrete example of this ledger formation rule is shown in Figure \ref{fig:prism_ledger}.

Operationally, we only need to count votes for intervals in the atomic partition of all intervals covered by the proposer blocks. After finding the leader block for each atomic interval, we can get the leader sequence by sanitizing the repeated proposer blocks.

\begin{figure}
    \centering
    \includegraphics[width=7cm]{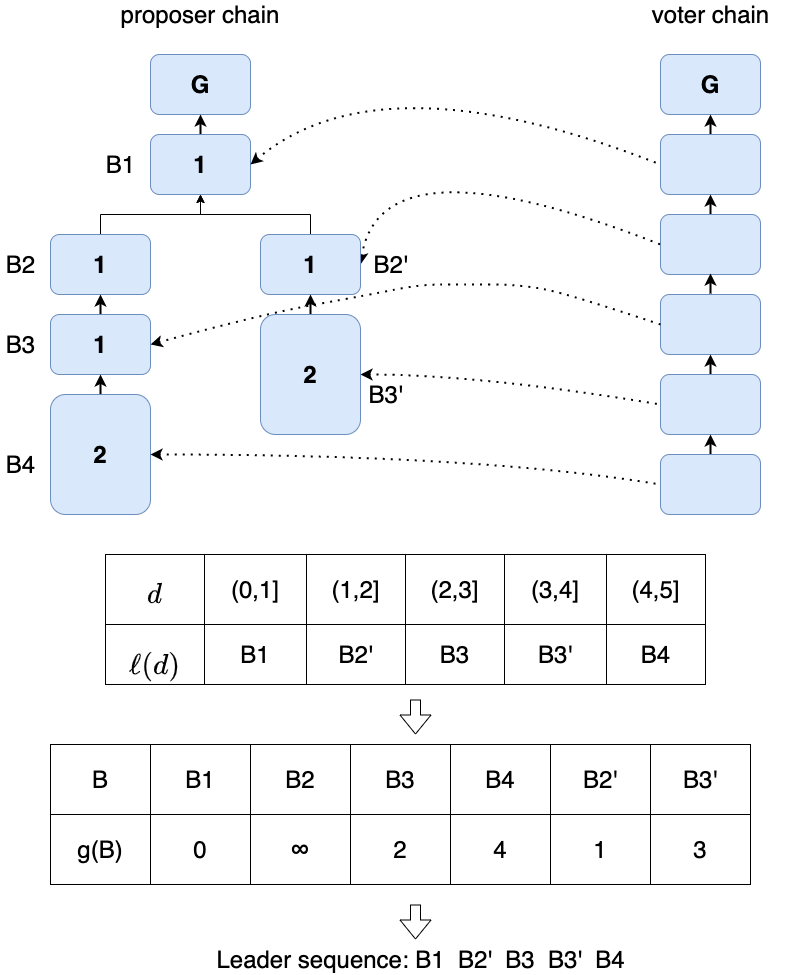}
    \vspace{-0.1in}
    \caption{An example of the ledger formation rule in \schemenosp. For simplicity, we only have one voter chain in the example. The number inside each proposer block is the block difficulty. In this example, the heaviest proposer chain has chain difficulty 5. We find the leader block $\ell(d)$ for each difficulty level $d$ in $(0,5]$ according to the votes (as shown in the first table). Then we find the grade $g(\cdot)$ of each proposer block by Equation (\ref{eqn:grade}) as shown in the second table. Finally, the proposer blocks are ordered by their grades.}
    \label{fig:prism_ledger}
    \vspace{-0.25in}
\end{figure}

\noindent{\bf Main result: persistence and liveness of \scheme (Informal)} 
We show that \scheme generates a transaction ledger that satisfies \textit{persistence} and \textit{liveness} in a variable mining power setting in Theorem \ref{thm:prism}.  

\section{OHIE}
\label{sec:ohie}

\subsection{Fixed Difficulty Algorithm}

\ohie \cite{yu2020ohie} composes $m$ parallel instances of \bitcoin longest chains. Each chain has a distinct genesis block, and the chains have ids from $0$ to $m-1$. Similar to \schemenosp, \ohie also uses cryptographic sortition to ensure that miners extend the $m$ chains concurrently and they do not know which chain a new block will extend until the PoW puzzle is solved.
    
Each individual chain in \ohie inherits the proven security properties of longest chain protocol \cite{backbone}, and all blocks on the $m$ chains confirmed by the longest chain confirmation rule (eg. the $k$-deep rule) are called {\em partially-confirmed}. However, this does not yet provide a total ordering of all the confirmed blocks across all the $m$ chains in \ohienosp. The goal of \ohie is to generate a {\em sequence of confirmed blocks} (SCB) across all $m$ parallel chains. Once a partially-confirmed block is added to SCB, it becomes {\em fully-confirmed}.

In \ohienosp, each block has two additional fields used for ordering blocks across chains, denoted as a tuple $({\tt rank},\,\, {\tt next\_rank})$. In SCB, the blocks are ordered by increasing ${\tt rank}$ values, with tie-breaking based on the {\em chain ids}. For any new block $B$ that extends from its parent block denoted as ${\tt parent}(B)$, we directly set $B$'s ${\tt rank}$ to be the same as ${\tt parent}(B)$'s ${\tt next\_rank}$. A genesis block always has ${\tt rank}$ of $0$ and ${\tt next\_rank}$ of $1$. Properly setting the ${\tt next\_rank}$ of a new block $B$ is the key  design in \ohienosp. Let $\mathcal{B}$ be the set of all tips of the $m$ longest chains before $B$ is added to its chain, then the ${\tt next\_rank}$ of $B$ is given by 
\begin{equation*}
    {\tt next\_rank}(B) = \max\{{\tt rank}(B)+1,\max_{B' \in \mathcal{B}} \{{\tt next\_rank}(B')\}\}.
\end{equation*}
If $B$ copies the ${\tt next\_rank}$ of a block $B'$ on a chain with different id, then a reference link to $B'$ (or the hash of $B'$) is added into $B$. In the example of Figure \ref{fig:ohie_fixed}, when  B11 is mined, B04 has the highest ${\tt next\_rank}$, so B11 copies the ${\tt next\_rank}$ of B04 and has a reference link to B04.

\ohie generates a SCB in the following way. Consider any given honest node at any given time and its local view of all the $m$ chains. Let $y_i$ be the ${\tt next\_rank}$ of the last partially-confirmed block on chain $i$ in this view. Let ${\tt confirm\_bar} \leftarrow \min_{i=1}^k y_i$. All partially-confirmed blocks whose rank is smaller than ${\tt confirm\_bar}$ are deemed fully-confirmed and included in SCB. Finally, all the fully-confirmed blocks will be ordered by increasing ${\tt rank}$ values, with tie-breaking favoring smaller chain ids. As an example, in Figure ~\ref{fig:ohie_fixed}, we have $y_0 = 4, y_1 =7, y_2=9$, hence ${\tt confirm\_bar}$ is $4$. Therefore, the $8$ partially-confirmed blocks whose ${\tt rank}$ is below $4$ become fully-confirmed.

\begin{figure}
    \centering
    \includegraphics[width=7cm]{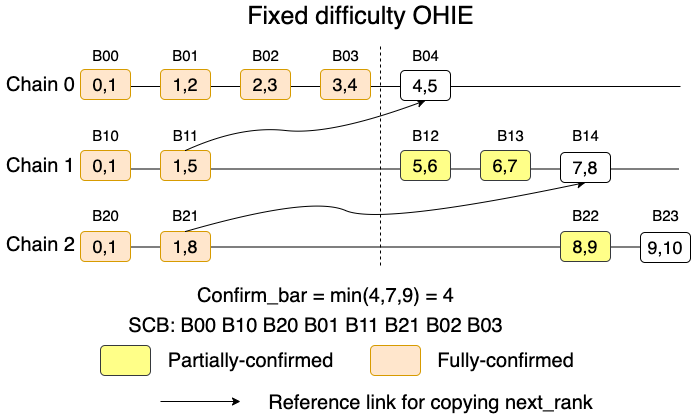}
    \vspace{-3mm}
    \caption{\ohie with fixed difficulty. Each block has a tuple $({\tt rank},\,\, {\tt next\_rank})$. In this figure, a block that is at least 2-deep in its chain is partially-confirmed. The blocks arrive in this order: B00, B10, B20, B01, B02, B03, B04, B11, B12, B13, B14, B21, B22, B23.}
    \label{fig:ohie_fixed}
    \vspace{-3mm}
\end{figure}

\begin{figure}
    \centering
    \includegraphics[width=8cm]{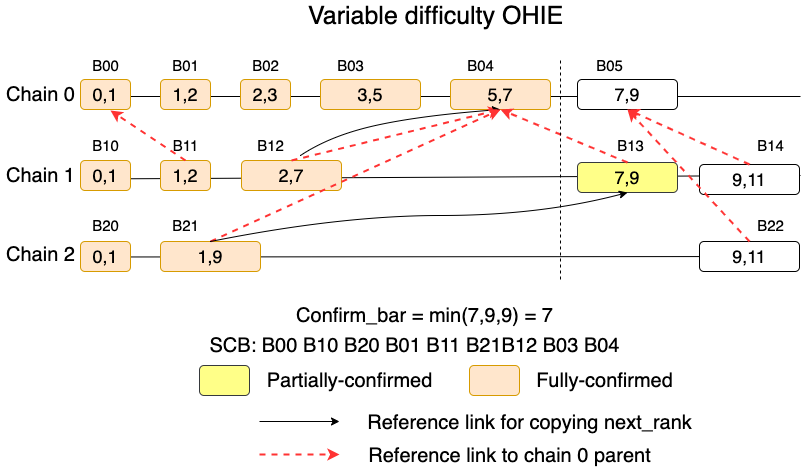}
    \vspace{-3mm}
    \caption{\ohie with variable difficulty. Each block has a tuple $({\tt rank},\,\, {\tt next\_rank})$. In this figure, a block that is at least 2-deep in its chain is partially-confirmed. The width of a block represents its mining difficulty. Different from the fixed difficulty algorithm, the mining difficulty is adjusted every 3 blocks on chain 0; Each block $B$ on chain $i$ ($i>0$) has a chain 0 parent (shown by the red reference link), which decides the mining difficulty of $B$. The blocks arrive in this order: B00, B10, B20, B11, B01, B02, B03, B04, B12, B13, B21, B05, B14, B22.}
    \label{fig:ohie_variable}
    \vspace{-5mm}
\end{figure}

\subsection{Variable Difficulty Algorithm}

Following the same meta-principle of designing variable difficulty \schemenosp, we can also turn the fixed difficulty \ohie into a variable difficulty algorithm by making the following changes. 
\begin{itemize}
    \item Each individual chain follows the heaviest chain rule instead of the longest chain rule. 
    \item The mining difficulty of chain 0 is adjusted the same way as the \bitcoin rule \cite{full2020}.
    \item Following our design principle \textbf{M1}, each block $B$ on chains $1,2,..., (m-1)$ will also have a {\em chain 0 parent} $\hat{B}$ (assigned before mining). The mining difficulty of $B$ is the same as the difficulty used to mine a child block of $\hat{B}$. To prevent the adversary from adopting an old mining difficulty from chain 0, we require that on each chain the referred chain 0 parent should have non-decreasing chain difficulty (\textbf{M2}). As an example in Figure ~\ref{fig:ohie_variable}, each block on chain 1 and chain 2 refers to (shown in red dashed arrow) a chain 0 parent with non-decreasing chain difficulty, which decides the mining difficulty of the block.
    \item A straightforward adoption on how to decide the ${\tt next\_rank}$ of a block would follow from our design principle \textbf{M3}. Let $\mathcal{B}$ be the set of all tips of the $m$ heaviest chains before $B$ is added to its chain, then the ${\tt next\_rank}$ of $B$ is given by 
    \begin{equation*}
        {\tt next\_rank}(B) = \max\{{\tt rank}(B)+{\rm diff}(B),\max_{B' \in \mathcal{B}} \{{\tt next\_rank}(B')\}\}.
    \end{equation*}
    If $B$ copies the ${\tt next\_rank}$ of a block $B'$ on a chain with different id, then a reference link to $B'$ (or the hash of $B'$) is added into $B$. Note that $B'$ may be different from $B$'s chain 0 parent, eg. B21 in Figure ~\ref{fig:ohie_variable}. We point out that this design is not necessary for the security analysis, but it is a very natural choice. 
    
\end{itemize}

\noindent{\bf Main result: persistence and liveness of \ohie (Informal)} 
We show that \ohie generates a transaction ledger that satisfies \textit{persistence} and \textit{liveness} in a variable mining power setting in Theorem \ref{thm:ohie}.  
\section{FruitChains}
\label{sec:fruit}

\subsection{Fixed Difficulty Algorithm}

The \fruitchains protocol was developed in order to solve the selfish mining problem and develop incentives which are approximately a Nash equilibrium. A key underlying step in \fruitchains is to ensure that a node that controls a certain fraction of mining power receives reward nearly proportional to its mining power, irrespective of adversarial action. \fruitchains runs an instance of Nakamoto consensus but instead of directly putting the transactions inside the blockchain, the transactions are put inside ``fruits'' and fruits are included by blocks. Mining fruits also requires solving some PoW puzzle. Similar to \scheme and \ohienosp, the \fruitchains protocol also uses cryptographic sortition to ensure that miners mine blocks and fruits concurrently and they do not know the type of the blocks until the puzzle is solved. Additionally, a fruit is required to ``hang'' from a block which is not too far from the block which includes the fruit. 

In \fruitchainsnosp, each of the fruit will have two parent blocks, we call them fruit parent and block parent: the fruit parent is a recently stabilized/confirmed block that the fruit is hanging from; the block parent should be the tip of the longest chain. A block will also have a fruit parent because the fruit mining and block mining are piggybacked atop each other, but a block actually does not care about this field. See Figure ~\ref{fig:fruit} for illustration. We say that a fruit $B_f$ is recent w.r.t. a chain $\C$ if the fruit parent of $B_f$ is a block that is at most $R$ deep in $\C$, where $R$ is called the recency parameter. The \fruitchains protocol requires that blocks only include recent fruits. Intuitively, the reason why fruits need to be recent is to prevent the ``fruit withhold attack'': without it, an attacker could withhold fruits, and suddenly release lots of them at the same time, thereby creating a very high fraction of adversarial fruits in some segment of the chain.

\begin{figure}
    \centering
    \includegraphics[width=6cm]{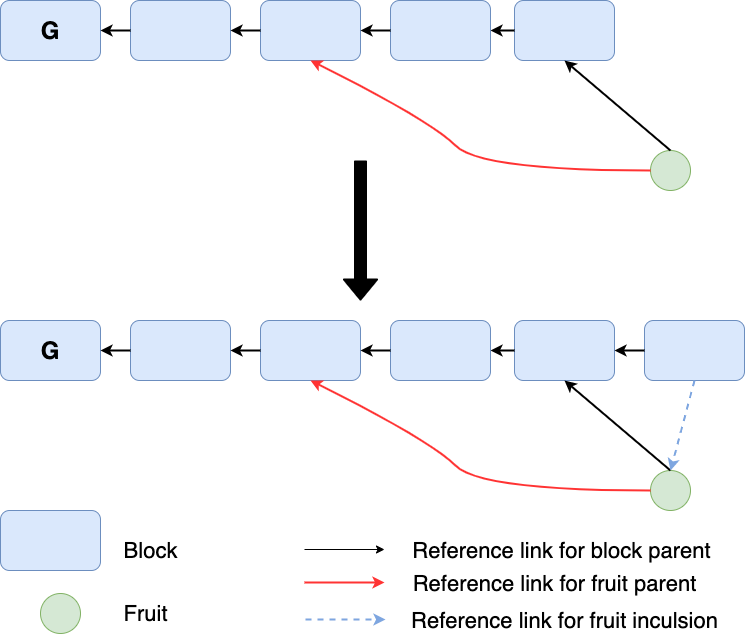}
    \vspace{-0.15in}
    \caption{The \fruitchains protocol. }
    \label{fig:fruit}
    \vspace{-0.25in}
\end{figure}

We term a blockchain protocol as fair if players controlling a $\phi$ fraction of the computational resources will reap a $\phi$ fraction of the rewards. Intuitively, the reason why the \fruitchains protocol guarantees fairness is that even if an adversary tries to ``erase'' some block mined by an honest player (which contains some honest fruits), by the liveness of the longest chain protocol, eventually an honest player will mine a new block including those fruits and the block will be stable -- in fact, by setting the recency parameter $R$ reasonably large, we can make sure that any fruit mined by an honest player will be included sufficiently deep in the chain. And further, if rewards and transaction fees are evenly distributed among the fruits in the long segment of the chain, then the \fruitchains protocol guarantees fairness.

\subsection{Variable Difficulty Algorithm}

Following our meta-principles, we can also turn the fixed difficulty \fruitchains into a variable difficulty algorithm by making the following changes. 
\begin{itemize}
    \item The underlying blockchain protocol follows the heaviest chain rule instead of the longest chain rule, i.e., the block parent of a block/fruit is the tip of the heaviest chain.
    \item The mining difficulty is adjusted the same way as the \bitcoin rule \cite{full2020}, and the block/fruit mining will use the same mining difficulty, or the difficulties of fruit and block will remain the same ratio (\textbf{M1}). 
    \item A fruit $B_f$ is recent w.r.t. a chain $\C$ at round $r$ if the fruit parent of $B_f$ is in $\C$ and has timestamp at least $r-R$, where $R$ is called the recency parameter. And again, blocks only include recent fruits, i.e., a block $B$ with timestamp $r$ is valid if for all fruits $B_f \in B$, the fruit parent of $B_f$ has timestamp at least $r-R$.
\end{itemize}

If rewards and transaction fees are designed to distribute proportional to the fruit difficulty in a sufficiently long segment of the chain, then the variable difficulty \fruitchains protocol guarantees fairness under a variable mining power setting. This is where the meta-principle \textbf{M3} kicks in. In the fixed difficulty setting, the reward is distributed equally among all fruit miners {\bf equally} in a window of blocks. In the variable difficulty setting, the reward is distributed {\bf proportional to the difficulty of the fruits}. To model this in our calculation of fairnesss, we say that the variable difficulty protocol is fair if the {\bf fraction of difficulty} of fruits of a given miner in a window is approximately proportional to its mining power. Note that monotonicity condition (\textbf{M2}) does not apply to variable difficulty \fruitchains as there is no chaining structure among the fruits. But the recency condition on the fruits has the same effect and does prevent the adversary from adopting an old mining difficulty for fruits.

\noindent{\bf Main result: persistence, liveness and fairness of \fruitchains (Informal)} 
We show that \fruitchains generates a transaction ledger that satisfies \textit{persistence}, \textit{liveness} and \textit{fairness} in a variable mining power setting in Appendix~\ref{app:fruit}.

\section{Security Analysis}
\label{sec:proof}

\subsection{Desired Security  Properties}
\label{sec:prop}
\begin{nota}
\label{notation:chains}
We denote by $\C^{\lceil \ell}$ the chain resulting from ``pruning'' the blocks with timestamps within the last $\ell$ rounds. If $\C_1$ is a prefix of $\C_2$, we write $\C_1 \prec \C_2$. The latest block in the chain $\mathcal{C}$ is called the head of the chain and is denoted by head$(\mathcal{C})$. We denote by $\C_1 \cap \C_2$ the common prefix of chains $\C_1$ and $\C_2$. We say that a chain $C$ is held by or belongs to an honest party if it is one of the heaviest chains in its view. 
\end{nota}
The following two properties called common prefix and chain quality, are essential in proving the persistence and liveness of the transaction ledger. The common prefix property states that any two honest parties’ chains at two rounds have the earlier one subsumed in the later as long as the last a few blocks are removed, while chain quality quantifies the contributions of the honest parties to any sufficiently long segment of the chain.

\begin{defn}[Common Prefix]
\label{def:common}
The common prefix property with parameter $\ell_{\rm cp} \in \mathbb{N}$ states that for any two honest players 
holding chains $\C_1$, $\C_2$ at rounds $r_1$, $r_2$, with $r_1 \leq r_2$, it holds that $\C_1^{\lceil \ell_{\rm cp}} \prec \C_2$.
\end{defn}

\begin{defn}[Chain Quality]
\label{def:quality}
The chain quality property is defined for two parameters $\ell_{\rm cq} \in \mathbb{N}$ and $\mu \in \mathbb{R}$. Let $\mathcal{C}$ be a chain held by any honest party at round $r$ and let $S_0 \subseteq [0,r]$ be an interval with at least $\ell_{\rm cq}$ consecutive rounds. Let $\C(S_0)$ be the segment of $\C$ containing blocks with timestamps in $S_0$ and $d$ be the total difficulty of all blocks in $\C(S_0)$. The chain quality property states that the honest blocks in $\C(S_0)$ have a total difficulty of at least $\mu d$. 
\end{defn}

In the context of \schemenosp, let $\texttt{LedSeq}_d(r)$ be the \textit{leader sequence}  up to difficulty level $d$ at round $r$. And the leader sequence at the end of round $\rmax$, the end of the protocol execution, is the \textit{final leader sequence}, $\texttt{LedSeq}_d(\rmax)$. Then similar to a single chain, we can define the following properties on the leader sequence.

\begin{defn}[Leader Sequence Common Prefix ]
\label{def:leader-prefix}
The leader sequence common prefix property with parameter $\ell_{\rm lscp} \in \mathbb{N}$ states that for a fixed difficulty level $d$, let $R_d$ be the first round in which a proposer block covering $d$ was received by all honest players, then it holds that
\begin{equation}
    \texttt{LedSeq}_{d}(r) = \texttt{LedSeq}_{d}(\rmax) \quad \forall r\geq R_d+\ell_{\rm lscp}.
\end{equation}
\end{defn}

\begin{defn}[Leader Sequence Quality]
\label{def:leader-quality}
The leader sequence property is defined for two parameters $\ell_{\rm lsq} \in \mathbb{N}$ and $\mu \in \mathbb{R}$. Let $\mathcal{C}$ be a proposer chain held by any honest party at round $r$ and let $D$ be the difficulty range covered by all blocks in $\C$ with timestamps in the last $\ell_{\rm lsq}$ rounds. The leader sequence quality property states that leader blocks mined by honest players cover at least $\mu$ fraction of $D$.
\end{defn}
Our goal is to generate a robust transaction ledger that satisfies {\em persistence}  and {\em liveness} as defined in  \cite{backbone,yu2020ohie}.

\begin{defn}[from \cite{backbone,yu2020ohie}]
    \label{def:public_ledger}
    A protocol $\Pi$ maintains a robust public transaction ledger if it organizes the ledger as a blockchain of transactions and it satisfies the following two properties:
    \begin{itemize}
        \item (Persistence) Consider the confirmed ledger $L_1$ on any node $u_1$ at any round $r_1$, and the confirmed ledger $L_2$ on any node $u_2$ at any round $r_2$ (here $u_1$ ($r_1$) may or may not equal $u_2$ ($r_2$)).  If $r_1 +\Delta < r_2$, then $L_1$ is a prefix of $L_2$.
        \item (Liveness) Parameterized by $u \in \mathbb{R}$, if a transaction {\sf tx} is received by all honest nodes for more than $u$ rounds, then all honest nodes will contain {\sf tx} in the same place in the confirmed ledger.
    \end{itemize}
\end{defn}

\subsection{Proof Sketch}
\label{sec:sketch}

\begin{figure}
    \centering
    \includegraphics[width=8cm]{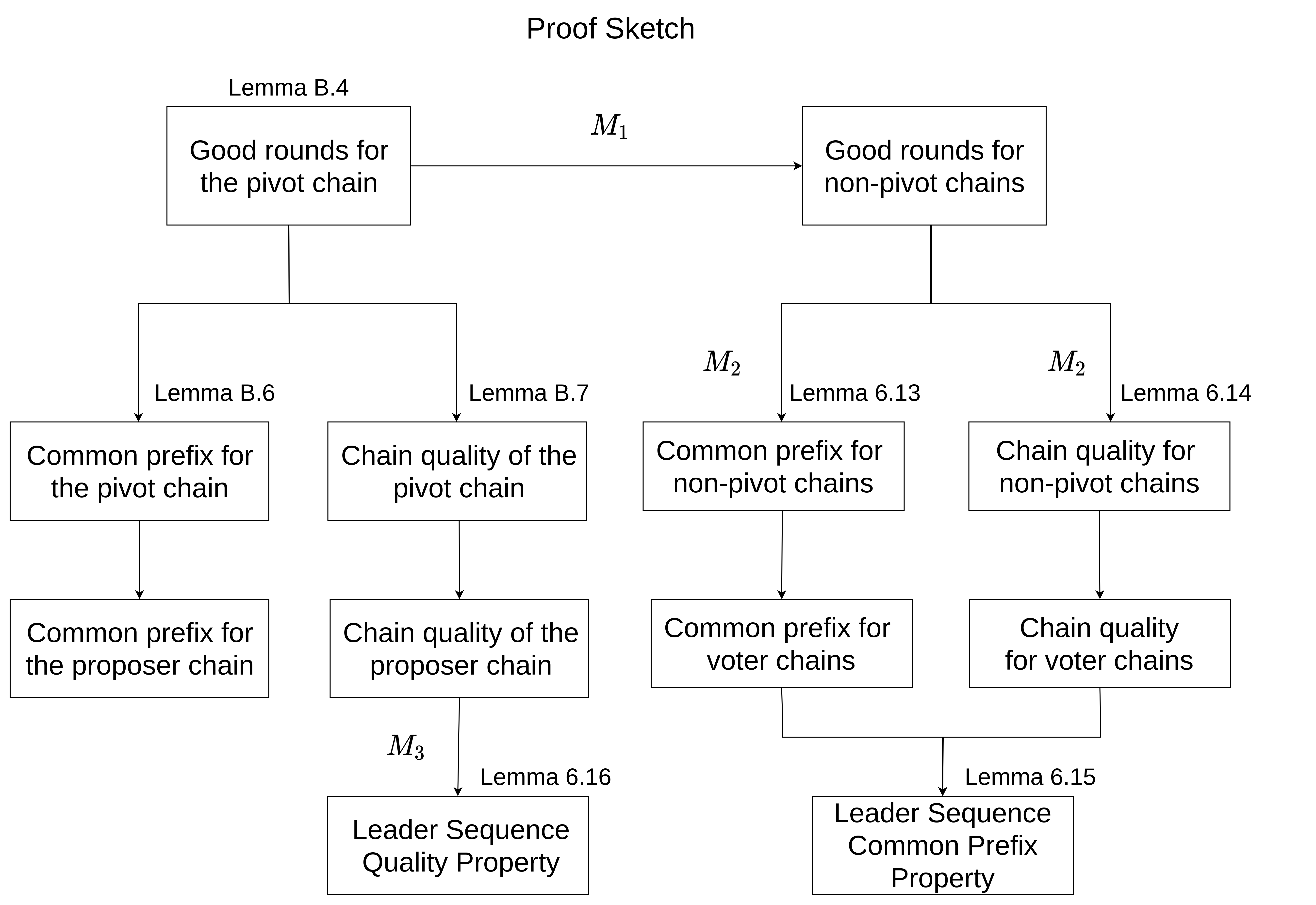}
    \vspace{-0.15in}
    \caption{Proof sketch for \schemenosp. M1, M2 and M3 are crucial in proving these properties for the leader sequence.} 
    \label{fig:proof_sketch}
    \vspace{-0.2in}
\end{figure}

Since there is a pivot chain in all three protocols (by {\bf M1}), the first step of our analysis is to prove some desired properties (including chain growth, common prefix, and chain quality) of the pivot chain.
As the pivot chain just follows the difficulty adjustment rule as in \bitcoinnosp, we can directly borrow results from a beautiful paper \cite{full2020}. The key step is to show that by adopting the heaviest pivot chain, honest nodes are always mining with ``reasonable'' block difficulties (this is formally defined as Good round/chain in Section \ref{sec:def}). We state all the useful lemmas and summarize the proof from \cite{full2020} in Appendix~\ref{app:backbone}.

The key technical challenge involves analyzing the properties of the non-pivot chains. Unlike in a pivot chain where all blocks in an epoch will have the same block difficulty, the block difficulties may experience sudden changes in non-pivot chains. This presents a significant barrier to surmount in our analysis, and differs from previous work in this area. Recall that {\bf M1} ensures that an honest party chooses the target of the next block in a non-pivot chain from the tip of the heaviest pivot chain in its view. Hence, the targets used by an honest party for the non-pivot chains are also reasonable. 
Then how about the non-pivot-chain blocks mined by the adversary?
As discussed in Section \ref{sec:prism}, allowing the miners to choose arbitrary mining difficulty in a non-pivot chain is risky. So we use the monotonicity condition {\bf M2} to ensure that non-pivot-chain blocks also have ``reasonable'' block difficulties even if the adversary mines them.

Then we prove that any two heaviest non-pivot chains cannot diverge for too long to prove the common prefix property.  We do this by considering two non-pivot chains $C_1$ and $C_2$ (in one of the non-pivot block tree) that diverge for too long and consider the last common honest block $B$ of $C_1$ and $C_2$. \textbf{M2} ensures that the blocks arriving after $B$ should refer to a pivot-chain block with monotonically non-decreasing chain difficulty than the one referred by $B$. We also argue that the chain difficulty intervals covered by uniquely successful honest blocks (defined as honest blocks that are mined more than $\Delta$ rounds apart) in chains $C_1$ and $C_2$ do not overlap similar to the analysis for the common prefix in \cite{full2020}. To make $C_1$ and $C_2$ diverge, the adversary has to accumulate an enormous total difficulty compared to uniquely successful honest blocks. 

When the number of adversarial queries is high in the chains $C_1$ and $C_2$ after the block $B$, we bound the difficulty accumulated by the adversary via concentration. When it is low, the variance is high; we prove this by dividing the problem into 5 cases. Since the adversary cannot contribute an enormous total difficulty compared to uniquely successful honest blocks in one of the heaviest chains, the chain quality property also holds. The full proof can be found in Section \ref{sec:non-pivot}.

The last step of our proof is using the desired proprieties on each individual chain to show the security of the full parallel-chain protocol. Since each parallel-chain protocol has its own way of forming the transaction ledger, the proof also has to differ. By properly turning the concept of block level to block's chain difficulty ({\bf M3}), we make sure that our proof works out for all three protocols. We complete the proof of persistence and liveness for \scheme in Section \ref{sec:proof_prism} (and the flowchart of the proof sketch can be found in Figure \ref{fig:proof_sketch}), while the proof for \ohie can be found in Appendix \ref{app:ohie}. In addition, we define and prove block reward fairness of \fruitchains under a variable mining setting in Appendix \ref{app:fruit}.

\subsection{Definitions}
\label{sec:def}

\begin{table}[h]
\vspace{-0.2in}
\centering
\begin{tabular}{ |c l| }
\hline
    $m \in \mathbb{N}$ & number of voter/parallel chains in \schemenosp/\ohie \\
    $n_r$ & number of honest parties mining in round $r$ \\
    $t_r$ & number of corrupted parties mining in round $r$ \\
    $\delta$ & advantage of honest parties ($t_r < (1 - \delta)n_r$ for all $r$) \\
    $\Delta$ & network delay in rounds \\
    $\kappa$ & security parameter; length of the hash function output  \\
    $\Phi \in \mathbb{N}$ & the length of an epoch in number of blocks\\
    $\tau \geq 1$ & the dampening filter (Definition~\ref{def:targetrecalc})\\
    $(\gamma, s)$ & restrictions on the fluctuation of the number of \\
    & parties across rounds (Definition~\ref{def:respecting}) \\
    $f$ & expected mining rate in number of blocks per round \\
    $\epsilon$ & quality of concentration of random variables \\
    $\lambda$ & related to the properties of the protocol \\
    $\ell$ & minimum number of rounds for concentration bounds\\
    $\rmax$ & total number of rounds in the execution\\
 \hline
\end{tabular}
\caption{The parameters used in our analysis.}
\vspace{-0.3in}
\label{tbl:notation}
\end{table}

Let $T, \Lambda, \Phi$ and $n$ denote the target of a block, duration of an epoch, epoch length and number of honest parties respectfully. Throughout the analysis, the block difficulty of a block with target $T$ is set to be $ 1/T$. The chain difficulty of a chain is equal to the sum of all block difficulties that comprise the chain. 
The following is the target recalculation function for the pivot chain which is the same function used in Bitcoin.
\begin{defn}[from \cite{backbone}]
\label{def:targetrecalc}
Consider a pivot chain of $v$ blocks with timestamps $(r_1 \ldots r_v)$. For fixed constants $\kappa, \tau, \Phi, n_0$ the initial number of participants, $T_0$ the initial target, the target calculation function $\mathcal{T}: \mathbb{Z}^* \rightarrow \mathbb{R}$ is defined as

\begin{align*}
    \mathcal{T}(\emptyset) &= T_0,\\
    \mathcal{T}(r_1 \ldots r_v) &= \left\{ \begin{array}{ll}
         \frac{1}{\tau}T& \text{if } \frac{n_0}{n(T,\Lambda)}T_0 < \frac{1}{\tau}T  \\
         \tau T&  \text{if } \frac{n_0}{n(T,\Lambda)}T_0 > {\tau}T \\
         \frac{n_0}{n(T,\Lambda)}T_0 & \text{ otherwise}
    \end{array}\right.
\end{align*}
where $n(T,\Lambda) = 2^\kappa \Phi/T\Lambda $, with $\Lambda = r_{\Phi'} - r_{\Phi' -\Phi}$, $T = \mathcal{T}(r_1 \ldots r_{\Phi'-1})$, and $\Phi' = \Phi \lfloor v/\Phi \rfloor$.

\end{defn}

We now define a notion of ``good'' properties such as good round and good chain. These properties will bound the targets used by the honest parties, which will help us prove chain quality and common prefix. 
\begin{defn} [Good round, from \cite{full2020}]
Let $T^{min}_r$ and $T^{max}_r$ denote the minimum and the  maximum targets the $n_r$ honest parties are querying the oracle for in round $r$. Round r is good if $f/2\gamma^2 \leq pn_rT^{min}_r$ and $ pn_rT^{max}_r \leq (1 + \delta)\gamma^2f$.
\end{defn}
\begin{defn} [Good chain, from \cite{full2020}]
Round $r$ is a target-recalculation point of a pivot chain $\mathcal{C}$, if $\mathcal{C}$ has a block with timestamp $r$ and height a multiple of $\Phi$. A target-recalculation point $r$ is good if the target $T$ of the next block satisfies $f/2\gamma \leq pn_rT \leq (1+\delta)\gamma f$. A pivot chain $\mathcal{C}$ is good if all its target-recalculation points are good.
\end{defn}

We will use the superscript $P$ to denote the variables, blocks, chains and sets corresponding to the pivot chain/tree and $i$ to denote the ones of the $i^{th}$ non-pivot chain/tree.

At any round $r$ of an execution, the adversary may keep chains in private that have the potential to be adopted by an honest party (because the private chains are heavier than the heaviest chain adopted by the honest party). So, we expand our chains of interest beyond the chains that belong to an honest party. For every non-pivot tree and the pivot tree, we define a set of valid chains $\mathcal{S}_r^{P} $ and $ \mathcal{S}_r^{i}$ \cite{full2020} that include the chains that belong to or have the potential to be adopted by an honest party. 

We will be dealing with random variables to quantify the difficulty accumulated by the honest parties and the adversary in our analysis. At round $r$, define the real random variable $D^P_r$ equal to the sum of the difficulties of all pivot-chain blocks computed by honest parties. Also, define $Y^P_r$ to equal the maximum difficulty among all pivot-chain blocks computed by honest parties, and $Q^P_r$ to equal $Y^P_r$ when $D^P_u = 0$ for all $r < u < r + \Delta$ and 0 otherwise. We call an honest block uniquely successful if it is mined at round $r$ such that $Q_r>0$. Similarly define $D^i_r, Y^i_r$ and $Q^i_r$ for the $i$-th non-pivot chain ($1\leq i \leq m$ in \scheme and $1\leq i \leq m-1$ in \ohienosp). 
For a set of rounds $S$, we define $D^P(S) = \sum_{r\in S}D^P_r, Q^P(S)= \sum_{r\in S}Q^P_r$ and $D^i(S)= \sum_{r\in S}D^i_r, Q^i(S)= \sum_{r\in S}Q^i_r$ for all $i$.

Regarding the adversary, for a set of $J$ adversarial queries to the oracle, let $T(J)$ be target associated with the first query in $J$. Define the real random variable $A^P(J)$, as the sum of difficulties of all the adversarial blocks created during queries in $J$ with difficulty less than $\tau/T(J)$. For all $i$, define $A^i(J)$ as the sum of difficulties of all the adversarial blocks created during queries in $J$ with difficulty less than $b^i(J) = max_{j\in J}sup\{A^i_j - A^i_{j-1} | \mathcal{E}_{j-1} = E_{j-1}\}$, a function associated with the set of queries $J$ (defined according to Theorem 8.1 in \cite{dubhashi2009concentration}). $A^P_j$ is the difficulty of the pivot-chain block with difficulty at most $\tau/T(J)$ obtained at the $j^{th}$ query of $J$. $A^i_j$ is the difficulty of the block obtained at $j^{th}$ query of $J$ for non-pivot chain $i$.

Let $\mathcal{E}$ denote the entire execution and let $\mathcal{E}_r$ be the execution just before round $r+1$. 
To obtain meaningful concentration of our random variables, we should be considering a sufficiently long sequence of at least

\begin{align}
\label{eq:re1}
    \ell \triangleq \frac{4(1+3\epsilon)}{\epsilon^2f[1 - (1+\delta)\gamma^2f]^{\Delta+1}} \max\{\Delta,\tau\}\gamma^3 \lambda 
\end{align}
consecutive rounds. 

We require $\Phi$ the duration of an epoch to be large enough in order to obtain meaningful security bounds:
\begin{align}
\label{eq:re2}
 \Phi \geq  4(1+\delta)\gamma^2 f(\ell + 3\Delta)\epsilon.
\end{align}

In order for the proofs for the security analysis to work, the parameters of the protocol should satisfy the following conditions:
\begin{align}
\label{eq:re3}
    [1 - (1+\delta)\gamma^2f]^{\Delta} \geq 1-\epsilon, ~~8\epsilon \leq \delta \leq 1.
\end{align}
Note that Equations (\ref{eq:re2}) and (\ref{eq:re3}) can always be satisfied by setting $\Phi$ to be large enough and $f$ to be small enough. Also note that (\ref{eq:re2}) and (\ref{eq:re3}) are not tight bounds on the parameters and are just sufficient conditions for the analysis to work.

We now define what a typical execution, which will help us bound the random variables in our analysis.
\begin{defn}[Typical Execution]
\label{def:typicality}
For any set $S$ of at least $\ell$ consecutive good rounds, any set of $J$ consecutive adversarial queries and $\alpha(J) = 2(\frac{1}{\epsilon} + \frac{1}{3})\lambda/T(J)$, an execution $E$ is typical if 
\begin{align*}
    &(1-\epsilon)[1 - (1+\delta)\gamma^2 f]^\Delta pn(S) < Q^{i/P}(S)\leq D^{i/P}(S) < (1+\epsilon)pn(S),\\
    &A^P(J) < p|J| + \max\{\epsilon p|J|, \tau \alpha(J)\}, \\
    &A^i(J) < p|J| + \max\{\epsilon p|J|, b^i(J)\lambda(\frac{1}{\epsilon} + \frac{1}{3}) \},
\end{align*}
where $b^i(J) = max_{j\in J}sup\{A^i_j - A^i_{j-1} | \mathcal{E}_{j-1} = E_{j-1}\}$.
\end{defn}
We now show that a typical execution is a high-probability event.
\begin{thm}
\label{thm:votertypicality}
For an execution $\mathcal{E}$ of $\rmax$ rounds, in a $(\gamma,s)$-respecting environment, the probability of the event  ``$\mathcal{E}$ not typical'' is bounded by $\mathcal{O}(\rmax^2)e^{-\lambda}$.
\end{thm}

The proof for Theorem ~\ref{thm:votertypicality} can be found in Appendix \ref{app:typicalproof}


\subsection{Non-pivot chain properties}
\label{sec:non-pivot}
Pivot chain behaves similar to the Bitcoin chain and its properties can be found in Appendix ~\ref{app:backbone}

Next we prove some desired properties for the non-pivot chains. 

\begin{lem}[Chain growth for non-pivot chain, from \cite{full2020}]
\label{lem:chaingrowth1}
Suppose that at round $u$ of an execution $E$, an honest party broadcasts a $i$-th non-pivot chain of difficulty $d$. Then, by round $v$, every honest party receives a chain of difficulty at least $d + Q^{i}(S)$, where $S = \{ r:u + \Delta \leq r \leq v - \Delta \}$.
\end{lem}

The proof of Lemma ~\ref{lem:chaingrowth1} is identical to Lemma \ref{lem:chaingrowth}.

At round $r$, to mine on a non-pivot chain block, an honest party picks a target from the tip of a pivot chain in $\mathcal{S}^P_r$ which has good targets at round $r$ because of Lemma \ref{lem:allgoodrounds}. So, as a consequence of \textbf{M1}, all the targets used by the honest parties on a non-pivot chain also satisfies $f/2\gamma^2 \leq pn_rT_r \leq f(1+\delta)\gamma^2$. 

\begin{lem}[Common prefix for non-pivot chains]
\label{lem:voter-prefix}
For a typical execution in a $(\gamma, 2(1 + \delta)\gamma^2 \Phi/f)$-respecting environment, each non-pivot chain satisfies the common-prefix property with parameter $\ell_{\rm cp} = \ell +2\Delta$.
\end{lem}

The proof of Lemma ~\ref{lem:voter-prefix} is in Appendix ~\ref{app:proof}

\begin{lem}[Chain quality for non-pivot chains]
\label{lem:voter-quality}
For a typical execution in a $(\gamma, 2(1 + \delta)\gamma^2 \Phi/f)$-respecting environment, each non-pivot chain satisfies the chain-quality property with parameter $\ell_{\rm cq} = \ell +2\Delta$ and $\mu = \delta - 3\epsilon$.
\end{lem}

The proof of Lemma~\ref{lem:voter-quality} is in Appendix~\ref{app:proof2}.

\subsection{Persistence and Liveness of \scheme}
\label{sec:proof_prism}

\begin{lem}[Leader sequence common prefix]
\label{lem:leader-prefix}
For a typical execution in a $(\gamma, 2(1 + \delta)\gamma^2 \Phi/f)$-respecting environment, the leader sequence satisfies the leader-sequence-common-prefix property with parameter $\ell_{\rm lscp} = 2\ell +4\Delta$.
\end{lem}

\begin{lem}[Leader sequence quality]
\label{lem:leader-quality}
For a typical execution in a $(\gamma, 2(1 + \delta)\gamma^2 \Phi/f)$-respecting environment, the leader sequence satisfies the leader-sequence-quality property with parameter $\ell_{\rm lsq} = \ell +2\Delta$ and $\mu = \delta - 3\epsilon$.
\end{lem}

The proofs of Lemma~\ref{lem:leader-prefix} and Lemma~\ref{lem:leader-quality} are in Appendix~\ref{app:proof3}.

\begin{thm}[Persistence and liveness of \schemenosp]
\label{thm:prism}
For a typical execution in a $(\gamma, 2(1 + \delta)\gamma^2 \Phi/f)$-respecting environment, \scheme satisfies persistence and liveness with parameter $u = \frac{4(1+\epsilon)\gamma^2(\ell + 2\Delta)}{(\delta - 3\epsilon)(1-\epsilon)^2}$.
\end{thm}


\begin{proof}
By our definition, the persistence of \scheme is equivalent to the leader sequence common prefix property proved in Lemma~\ref{lem:leader-prefix}.

We next prove the liveness property. Suppose a transaction {\sf tx} is received by all honest nodes before or at round $r_0$. Let $r \geq r_0 + u$ be current time and we shall prove that {\sf tx} is contained in the permanent leader sequence of all honest nodes at round $r$. As shown in Figure \ref{fig:prismliveness}, let $S_1 = \{r_0, \cdots, r\}$, $S_2 = \{r-2\ell-4\Delta, \cdots, r\}$, and $J$ be the adversarial queries in $S_2$. By Lemma \ref{lem:leader-prefix}, for a difficulty level $d$, if $d$ is covered by an honest block mined in $S1 \setminus S2$, then the block covering $d$ will be permanent in the leader sequence at round $r$. We know that the difficulty level grows at least $Q^P(S_1) \geq (1-\epsilon)^2p n_r u/\gamma$ in $S_1$. By Lemma \ref{lem:leader-quality}, we have that among the chain growth in $S_1$, different difficulty levels with size at least $(\delta - 3\epsilon)(1-\epsilon)^2p n_r u/\gamma$ is covered by honest leader blocks (which may not be permanent at round $r$). On the other hand, the proposer blocks that are not permanent (mined in $S_2$) cover different difficulty levels with size at most 
\begin{figure}
    \centering
    \includegraphics[width=7cm]{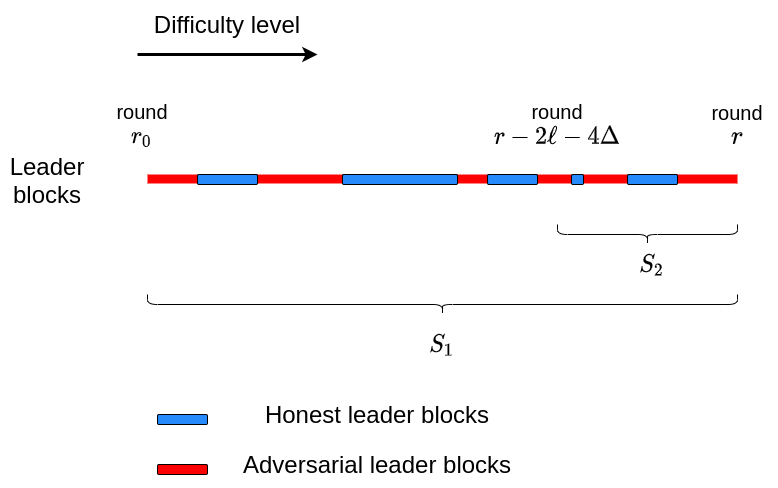}
    \vspace{-0.2in}
    \caption{The leader blocks at each difficulty level in the proposer tree.}
    \vspace{-0.2in}
    \label{fig:prismliveness}
\end{figure}
\begin{align*}
    D^P(S_2) + A^P(S_2) &< 2D^P(S_2) \leq 2(1+\epsilon)p\gamma n_r(2\ell+4\Delta) \\
    &= 4(1+\epsilon)p\gamma n_r(\ell+2\Delta).
\end{align*}
Hence at least one honest proposer block $B$ mined after $r_0$ is permanent in the leader sequence at round $r$. Since either $B$ or some proposer block referred by $B$ will contain $\sf tx$, in both case we can conclude the proof.

\end{proof}

\section{Evaluation}

In our evaluation, we  answer the following questions. 
\begin{itemize}
    \item Is the proposed scheme effective in matching the mining difficulty and the miner hash power?
    \item Does the blockchain forking rate remain low under our scheme, even with changing miner hash power?
    \item Does our scheme ensure that non-pivot chains adopt the difficulty of pivot chains, even with presence of the adversary?
    \item Does our scheme cause major computation and communication overhead when applied? 
\end{itemize}

\begin{figure*}[ht]
\minipage[t]{0.32\textwidth}
  \centering
    \includegraphics[width=\textwidth]{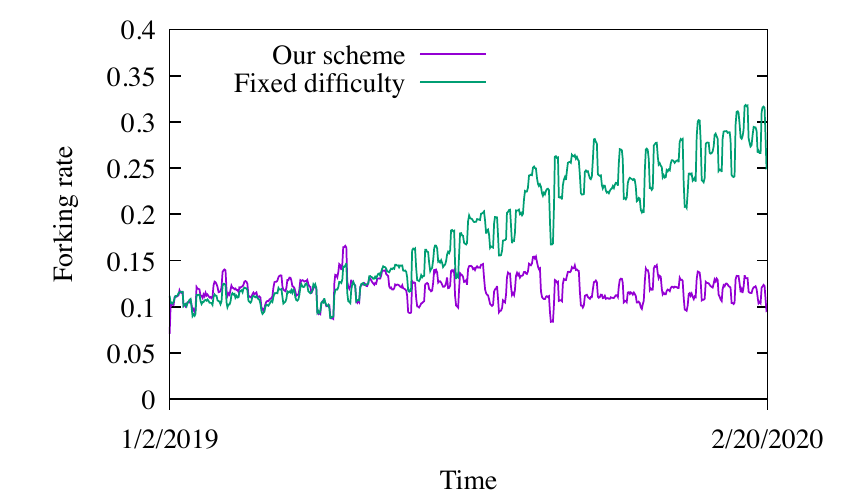}
    \vspace{-0.25in}
    \caption{Forking rate of all parallel chains in two simulations, one using our scheme and one using fixed difficulty.}
    \label{fig:health}
\endminipage\hfill
\minipage[t]{0.32\textwidth}
   \centering
    \includegraphics[width=\textwidth]{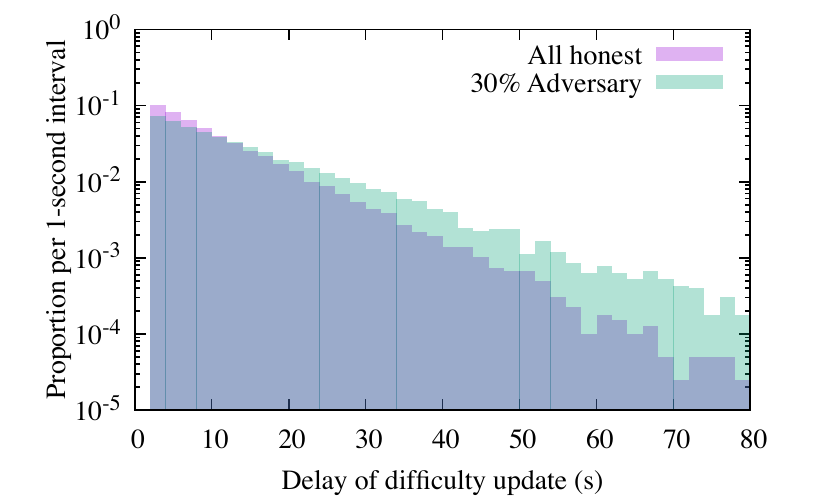}
    \vspace{-0.25in}
    \caption{Frequency histogram of the delay where non-pivot chains update their difficulty to follow that of the pivot chain. Note the y-axis is log scale.}
    \label{fig:delay}
\endminipage\hfill
\minipage[t]{0.32\textwidth}%
  \centering
    \includegraphics[width=\textwidth]{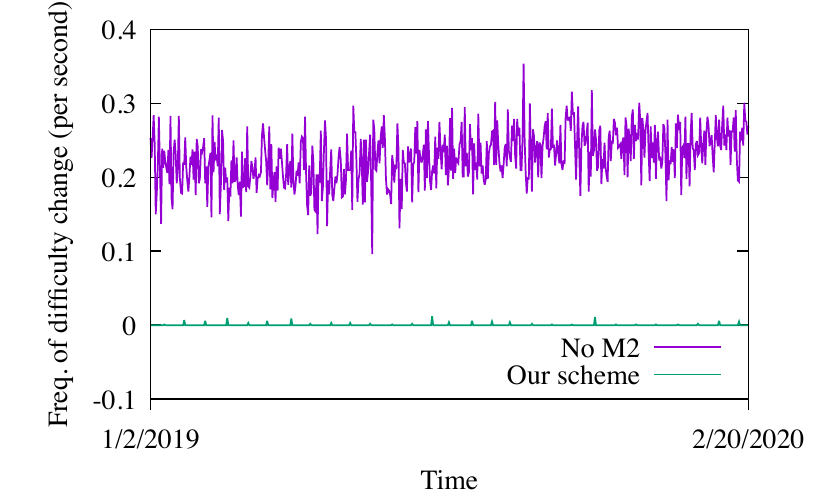}
    \vspace{-0.25in}
    \caption{Frequency of difficulty change on a non-pivot chain where 30\% of miner power is adversarial. }
    \label{fig:advdiff}
\endminipage
\end{figure*}

\subsection{Experimental Setup}
\label{sec:exp-setup}

\noindent\textbf{Simulator.} To evaluate our scheme, we build a mining simulator for parallel-chain protocols in Golang. The simulator uses a round-by-round model with an adjustable round interval. In each round, blocks are mined on each of the parallel chains, and the number of blocks mined is determined by drawing from independent Poisson random variables with mean set to the product of the round interval and the per-chain mining rate. Miners receive newly-mined blocks after an adjustable network latency.

\noindent\textbf{Simulated protocol.} Our simulator does not consider the \textit{interpretation} of the chains, such as transaction confirmation, ledger formation, etc. We only simulate the mining process. As a result, our evaluation is not tied to any particular protocol. Meanwhile, it is meaningful broadly to all PoW parallel-chain protocols, because they share this mining process.

There are 1 pivot chain and 1000 non-pivot chains. We simulate PoW mining on each of the chains at the same mining rate $f$. Each pivot-chain block contains its timestamp, difficulty, and parent. Each non-pivot-chain block also contains all these fields, plus a reference to a pivot-chain block (\textbf{M1}). We simulate two parties of miners: honest and adversary. Honest miners follow the general methodology described in section~\ref{sec:intro} by always referring to the best block in the pivot chain. They enforce the rules \textbf{M1}, \textbf{M2} by rejecting any non-compliant block . We design different adversarial miners to simulate attacks, and we provide more details later in Section \ref{sec:exp-attack}.

\noindent\textbf{Parameters.} The round interval and the network latency are set to 2 seconds according to data collected in large-scale experiments of \scheme \cite{prism-system}. The target mining rate $f$ is set to 0.1 block per second per chain according to \cite{prism-system}. The epoch length $\Phi$ is set to 2016 blocks, and the dampening filter $\tau$ is set to 4 according to \bitcoin. We replay the historical \bitcoin mining power data  \cite{hashratedata} during the simulation.

\subsection{Adaptation to Changing Miner Power}

The main purpose of our scheme is to ensure the mining difficulty
adapts to changing mining power. To show that, we simulate our scheme while varying the mining power according to the historical \bitcoin miner hash rate trace from Jan 2, 2019 to Feb 20, 2020. Figure ~\ref{fig:normal-trace} shows that even though the miner hash power has tripled during the simulated period, the mining difficulty of every chain keeps tracking the miner hash power very closely. Also, at any point in time, the max and min difficulty of all chains are very close. This demonstrates that the mining difficulty of all chains are always closely coupled, and no single chain experiences unstable difficulty or vulnerability.

\begin{figure}
    \centering
    \vspace{-6mm}
    \includegraphics[width=0.43\textwidth]{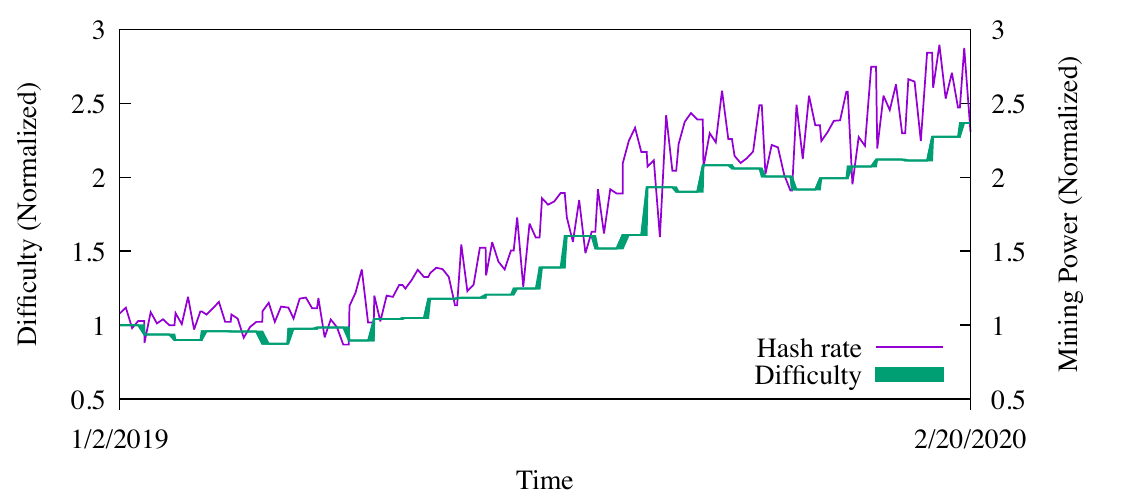}
    \vspace{-5mm}
    \caption{Miner hash power and mining difficulty of each chain when simulating our scheme over the historical \bitcoin miner power trace. Difficulty is plotted as a region to show the max and min difficulty across all chains. Both metrics are normalized over their initial values.}
    \label{fig:normal-trace}
    \vspace{-6mm}
\end{figure}

As mentioned in Section~\ref{sec:intro}, support for variable miner power is crucial to keeping the blockchain secure. If the miner hash power increases while the mining difficulty stays the same, the forking rate will increase due to decreased block inter-arrival time. To show our scheme is effective in keeping the blockchain secure, we compare the forking rate of two simulations: one using our scheme and one using a fixed mining difficulty. We use the same \bitcoin mining power data as in the previous experiment, and Figure ~\ref{fig:health} shows the results. Here, we report the forking rate as the ratio of the number of blocks not on the longest chain, to the number of blocks on the longest chain. If a fixed difficulty is used, the forking rate quickly increases as the miner power increases, to almost tripling towards the end of the simulation. In comparison, our scheme keeps the forking rate low across all parallel chains for the whole simulation. This is because the mining difficulty and the miner hash power are closely matched under our scheme, so the block mining rate stays at a safe level.

\subsection{Difficulty Update on Non-pivot Chains}

\label{sec:exp-attack}

One major challenge in designing our scheme is to ensure non-pivot chains adopt the pivot chain difficulty quickly after a new epoch begins, and we achieve it with the \textbf{M2} (Monotonicity, cf. Section~\ref{sec:intro}). To show that adversarial miners cannot delay this process, we simulate our scheme where 30\% of miners are adversarial. Adversarial miners do not voluntarily refer to the latest block on the pivot chain after a new epoch begins, but rather try to stay in the previous epoch (and mining difficulty) for as long as possible. We also simulate an all-honest scenario for comparison. We measure how soon non-pivot chains adopt new difficulty by tracking the delay from the last block of the previous epoch on the pivot chain to the first block of the new epoch of the non-pivot chain. Figure~\ref{fig:delay} shows the results. In either scenario, the difficulty of non-pivot chains is updated within 1--5 block intervals (0--50 seconds in real time). Although adversarial presence does delay the update of difficulty, the delay is not significant. This demonstrates that our mechanism ensures in-time update of non-pivot-chain difficulty.

We demonstrate that \textbf{M2} is essential to ensuring the mining difficulty does not vary too frequently on non-pivot chains. We compare two simulations where 30\% of miner power is adversarial. In one case, we apply our full scheme. In the other case, we disable \textbf{M2} so that the adversary is free to choose whatever block on the pivot chain to refer when mining non-pivot chain blocks. Specifically, the adversary always tries to mine blocks with the lowest difficulty possible by referring to the genesis pivot-chain block. We focus on one non-pivot chain, and track the \textit{frequency} of difficulty change. Difficulty change is defined as a block on the longest chain having different difficulty than its parent. Figure~\ref{fig:advdiff} shows the results. Under our scheme, non-pivot chain difficulty does not change for most of the time, and only changes swiftly at the beginning of new epochs, so the curve for our scheme stays close to zero. On the contrary, if we disable \textbf{M2}, the difficulty oscillates violently, as frequently as 0.2 times per second on average. This shows that our design is essential to maintain stable mining difficulty of non-pivot chains.

\subsection{Analysis of Overhead}
\label{sec:overhead}

Finally, we analyze and show that our schemes will cause minimal overhead when implemented on existing parallel-chain protocols. 

\noindent\textbf{Communication and storage.} Every block on the non-pivot chains needs to refer to a block on the pivot chain (\textbf{M1}), which takes the size of a hash (usually 32 bytes). This is a very small overhead compared to the size of the blockchain. For example, in \schemenosp, the size of a voter (non-pivot-chain) block is 534 bytes \cite{prism-system}. The pivot-chain reference constitutes to an increase of 6\% in communication and storage cost for voter blocks. Notice
that voter blocks themselves only make up for 0.21\% of the size of the Prism blockchain~\cite{prism-system}, so the overhead of pivot-chain referencing is negligible, regardless of the parameters.

\noindent\textbf{Computation.} Our scheme changes the mining and the transaction confirmation process of parallel chain protocols. For mining, notice that the pivot chain follows the same difficulty adjustment rule as \bitcoinnosp, which is proven practical by its real-world deployment. Mining on non-pivot chains uses the same difficulty as the pivot chain, so there is no additional bookkeeping.  

For transaction confirmation, we use \scheme as a concrete example (note that no computation overhead exists in transaction confirmation for \ohie and \fruitchainsnosp). 
Under static difficulty, \scheme selects a leader for every \textit{level} of the proposer tree.
With \textbf{M3}, we partition the proposer tree into real-valued
difficulty \textit{intervals} such that no interval is \textit{partially} occupied by any proposer block. We need to select a leader for each of such \textit{intervals} (section \ref{sec:prism-full-cfm-rule}). To determine the overhead, we need to answer: how many more intervals are there compared to levels?

We simulate the mining process of Prism with 1000 voter chains, epoch length $\Phi = 2016$ blocks, target mining rate $f=0.1$ block per second, and found the number of intervals is only 0.12\% more than the number of levels. That is, our scheme incurs a confirmation overhead of 0.12\%. This is expected, because only forks that happen at the beginning of an epoch will lead to extra intervals, and such a fork rarely exists with $\Phi=2016$ and $f=0.1$. Decreasing $\Phi$ may cause the overhead to increase because there are more epochs and it is more likely to fork at the beginning of an epoch. Table \ref{table:confirmation-overhead} plots the confirmation overhead for different $\Phi$; we see that  even at $\Phi=10$, the overhead is smaller than 1\%.

\begin{table}[t]
	\centering
	\caption{Confirmation overhead vs epoch length $\Phi$}
	\vspace{-0.1in}
	\begin{tabular}{  c || c | c | c | c } 
	 \hline
     $\Phi$ & 10  & 100 & 1000 & 2016 \\
	 \hline
	 Overhead & 0.43\% & 0.07\% & 0.11\% & 0.12\% \\ 
	 \hline
	\end{tabular}
	\vspace{-0.2in}
	\label{table:confirmation-overhead}
\end{table}

\section{Discussion}
We presented a general methodology by which any parallel chain protocol can be converted from the fixed difficulty to the variable difficulty setting. We also proved the safety, liveness, and performance of the proposed scheme using novel proof method that analyzes the coupling between the pivot and non-pivot chains. There are several open directions of research. 1) In our design methodology, we proposed using a single chain as a pivot chain to set the difficulty target for all blocks. However, if we can use the information (for example, inter-block arrival times) from all the chains together to determine the difficulty target, we can get much better statistical averaging. This can lead to protocols which can adapt to much more aggressive mining power variation than is possible with a single-chain protocol. Such a protocol needs to be designed with care since it leads to strong coupling across all the chains. In particular, every chain needs to know the state of all other chains in order to check the correctness of the difficulty target. Since other chains can have forking in the meanwhile, it may lead to unintended complex interactions. 
2) We analyzed various protocols under the variable difficulty setting. One new protocol, called Ledger-combiners \cite{fitzi2020ledger} uses parallel-chains for robustly combining multiple ledgers as well as for achieving low latency. Analyzing that protocol in the variable difficulty setting is an interesting direction for future work.

\section{Acknowledgements}
This research is supported in part by a gift from IOHK Inc., an Army Research Office grant W911NF1810332 and by the National Science Foundation under grants CCF 17-05007 and CCF 19-00636.

%
%
%


\bibliographystyle{plain}
\bibliography{reference}

\begin{thebibliography}{10}

\bibitem{hashratedata}
Blockchain charts - total hash rate.
\newblock \url{https://www.blockchain.com/charts/hash-rate}.

\bibitem{badertscher2018ouroboros}
Christian Badertscher, Peter Ga{\v{z}}i, Aggelos Kiayias, Alexander Russell,
  and Vassilis Zikas.
\newblock Ouroboros genesis: Composable proof-of-stake blockchains with dynamic
  availability.
\newblock In {\em Proceedings of the 2018 ACM SIGSAC Conference on Computer and
  Communications Security}, pages 913--930, 2018.

\bibitem{bagaria2019prism}
Vivek Bagaria, Sreeram Kannan, David Tse, Giulia Fanti, and Pramod Viswanath.
\newblock Prism: Deconstructing the blockchain to approach physical limits.
\newblock In {\em Proceedings of the 2019 ACM SIGSAC Conference on Computer and
  Communications Security}, pages 585--602, 2019.

\bibitem{bahack2013theoretical}
Lear Bahack.
\newblock Theoretical bitcoin attacks with less than half of the computational
  power (draft).
\newblock {\em arXiv preprint arXiv:1312.7013}, 2013.

\bibitem{bonneau2015sok}
Joseph Bonneau, Andrew Miller, Jeremy Clark, Arvind Narayanan, Joshua~A Kroll,
  and Edward~W Felten.
\newblock Sok: Research perspectives and challenges for bitcoin and
  cryptocurrencies.
\newblock In {\em 2015 IEEE symposium on security and privacy}, pages 104--121.
  IEEE, 2015.

\bibitem{chan2020Varying}
T-H.~Hubert Chan, Naomi Ephraim, Antonio Marcedone, Andrew Morgan, Rafael Pass,
  and Elaine Shi.
\newblock Blockchain with varying number of players.
\newblock Cryptology ePrint Archive, Report 2020/677, 2020.
\newblock \url{https://eprint.iacr.org/2020/677}.

\bibitem{croman2016scaling}
Kyle Croman, Christian Decker, Ittay Eyal, Adem~Efe Gencer, Ari Juels, Ahmed
  Kosba, Andrew Miller, Prateek Saxena, Elaine Shi, Emin~G{\"u}n Sirer, et~al.
\newblock On scaling decentralized blockchains.
\newblock In {\em International conference on financial cryptography and data
  security}, pages 106--125. Springer, 2016.

\bibitem{dubhashi2009concentration}
Devdatt~P Dubhashi and Alessandro Panconesi.
\newblock {\em Concentration of measure for the analysis of randomized
  algorithms}.
\newblock Cambridge University Press, 2009.

\bibitem{eyal2016bitcoin}
Ittay Eyal, Adem~Efe Gencer, Emin~G{\"u}n Sirer, and Robbert Van~Renesse.
\newblock Bitcoin-ng: A scalable blockchain protocol.
\newblock In {\em 13th $\{$USENIX$\}$ symposium on networked systems design and
  implementation ($\{$NSDI$\}$ 16)}, pages 45--59, 2016.

\bibitem{fitzi2020ledger}
Matthias Fitzi, Peter Ga{\v{z}}i, Aggelos Kiayias, and Alexander Russell.
\newblock Ledger combiners for fast settlement.
\newblock In {\em Theory of Cryptography Conference}, pages 322--352. Springer,
  2020.

\bibitem{backbone}
Juan Garay, Aggelos Kiayias, and Nikos Leonardos.
\newblock The bitcoin backbone protocol: Analysis and applications.
\newblock In {\em Annual International Conference on the Theory and
  Applications of Cryptographic Techniques}, pages 281--310. Springer, 2015.

\bibitem{garay2017bitcoin}
Juan Garay, Aggelos Kiayias, and Nikos Leonardos.
\newblock The bitcoin backbone protocol with chains of variable difficulty.
\newblock In {\em Annual International Cryptology Conference}, pages 291--323.
  Springer, 2017.

\bibitem{full2020}
Juan Garay, Aggelos Kiayias, and Nikos Leonardos.
\newblock Full analysis of nakamoto consensus in bounded-delay networks.
\newblock Cryptology ePrint Archive, Report 2020/277, 2020.
\newblock \url{https://eprint.iacr.org/2020/277}.

\bibitem{gilad2017algorand}
Yossi Gilad, Rotem Hemo, Silvio Micali, Georgios Vlachos, and Nickolai
  Zeldovich.
\newblock Algorand: Scaling byzantine agreements for cryptocurrencies.
\newblock In {\em Proceedings of the 26th Symposium on Operating Systems
  Principles}, pages 51--68, 2017.

\bibitem{li2018scaling}
Chenxing Li, Peilun Li, Dong Zhou, Wei Xu, Fan Long, and Andrew Yao.
\newblock Scaling nakamoto consensus to thousands of transactions per second.
\newblock {\em arXiv preprint arXiv:1805.03870}, 2018.

\bibitem{li2020taiji}
Songze Li and David Tse.
\newblock Taiji: Longest chain availability with bft fast confirmation.
\newblock {\em arXiv preprint arXiv:2011.11097}, 2020.

\bibitem{mazieres2015stellar}
David Mazieres.
\newblock The stellar consensus protocol: A federated model for internet-level
  consensus.
\newblock {\em Stellar Development Foundation}, 2015.

\bibitem{nakamoto2008bitcoin}
Satoshi Nakamoto.
\newblock Bitcoin: A peer-to-peer electronic cash system.
\newblock Technical report, 2008.

\bibitem{negy2020selfish}
Kevin~Alarc{\'o}n Negy, Peter~R Rizun, and Emin~G{\"u}n Sirer.
\newblock Selfish mining re-examined.
\newblock In {\em International Conference on Financial Cryptography and Data
  Security}, pages 61--78. Springer, 2020.

\bibitem{pass2017analysis}
Rafael Pass, Lior Seeman, and Abhi Shelat.
\newblock Analysis of the blockchain protocol in asynchronous networks.
\newblock In {\em Annual International Conference on the Theory and
  Applications of Cryptographic Techniques}, pages 643--673. Springer, 2017.

\bibitem{pass2017fruitchains}
Rafael Pass and Elaine Shi.
\newblock Fruitchains: A fair blockchain.
\newblock In {\em Proceedings of the ACM Symposium on Principles of Distributed
  Computing}, pages 315--324, 2017.

\bibitem{pass2017hybrid}
Rafael Pass and Elaine Shi.
\newblock Hybrid consensus: Efficient consensus in the permissionless model.
\newblock In {\em 31st International Symposium on Distributed Computing (DISC
  2017)}. Schloss Dagstuhl-Leibniz-Zentrum fuer Informatik, 2017.

\bibitem{schwartz2014ripple}
David Schwartz, Noah Youngs, Arthur Britto, et~al.
\newblock The ripple protocol consensus algorithm.
\newblock {\em Ripple Labs Inc White Paper}, 2014.

\bibitem{sompolinsky2015secure}
Yonatan Sompolinsky and Aviv Zohar.
\newblock Secure high-rate transaction processing in bitcoin.
\newblock In {\em International Conference on Financial Cryptography and Data
  Security}, pages 507--527. Springer, 2015.

\bibitem{prism-system}
Lei Yang, Vivek Bagaria, Gerui Wang, Mohammad Alizadeh, David Tse, Giulia
  Fanti, and Pramod Viswanath.
\newblock Prism: Scaling bitcoin by 10,000 x.
\newblock {\em arXiv preprint arXiv:1909.11261}, 2019.

\bibitem{yu2020ohie}
Haifeng Yu, Ivica Nikoli{\'c}, Ruomu Hou, and Prateek Saxena.
\newblock Ohie: Blockchain scaling made simple.
\newblock In {\em 2020 IEEE Symposium on Security and Privacy (SP)}, pages
  90--105. IEEE, 2020.

\end{thebibliography}

\appendices
\section*{Appendix}
\section{The Difficulty Raising Attack}
\label{app:attack}

\bitcoin set its target recalculation using a ``dampening filter''-like adjustment (as defined in Definition~\ref{def:targetrecalc}). It turns out that this design is surprisingly foresighted. If we make a relaxation of the adjustment mechanism by removing the dampening filter, then it is subject to an attack called difficulty raising attack firstly discovered in \cite{bahack2013theoretical}. At a high level, in this attack the adversary mines private blocks with timestamps in rapid succession, and induce one block with arbitrarily high difficulty in the private chain; via an anti-concentration argument, a sudden adversarial advance that can break agreement amongst honest parties cannot be ruled out. In this appendix, we describe this attack in detail and explain why having a ``dampening filter'' in the target recalculation function could resolve it.

\noindent {\bf A simple attack.} As a prelim, we first look at a simple attack if the protocol lets miners to choose their own difficulty and use the heaviest chain rule. At a first glance, this rule appears kosher - the heaviest chain rule seems to afford no advantage to any miner to manipulate their difficulty. However, this lack of advantage only holds in expectation, and the variance created by extremely difficult adversarial blocks can thwart a confirmation rule that confirms deeply-embedded blocks, no matter how deep, with non-negligible probability. We give a simple calculation here. For simplicity, we using the difficulty defined in the genesis block as the difficulty unit and the expected inter-block time (10 minutes in \bitcoinnosp) as the time unit. Let $n$ be number of honest queries to the hash function per unit time and $t$ be the number of adversarial queries per unit time. Then we know that to mine a block with unit difficulty, each query solves the PoW puzzle with probability $1/n$. We further assume that $n$ and $t$ don't change over time and the network delay among honest nodes is zero. Note that these assumptions only make the adversary weaker. The goal of the adversary is to double-spend a coin by mining a heavier chain than the public honest chain from the genesis.

Suppose honest miners are adopting the initial mining difficulty as defined in the genesis block, hence on average it take $k$ units of time to mine a honest chain with $k$ blocks. To mine a heavier chain, the adversary only needs to mine one block which has difficulty $k$ (See Figure ~\ref{fig:bahack0} for illustration), within $k$ unit of time. The adversarial can make $tk$ queries in $k$ units of time, and each query succeeds with probability $1/nk$. Hence the success probability of this attack would be
$$\mathbb{P}({\rm attack~succeeds}) = 1 - (1-\frac{1}{nk})^{tk} \approx 1- e^{t/n},$$
since $n$ and $t$ are large in PoW mining. Note that the success probability is a constant independent of $k$, therefore any $k$-deep confirmation rule will fail.

\begin{figure}
    \centering
    \includegraphics[width=6cm]{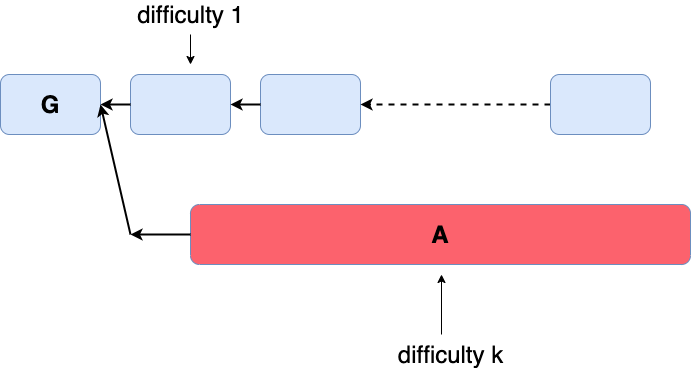}
    \caption{A simple attack if allowing miners to choose their own difficulty. The adversary mines one block which is as difficult as $k$ honest blocks.}
    \label{fig:bahack0}
\end{figure}

\noindent {\bf Difficulty raising attack.} However, even if we adopts a epoch based difficulty adjustment rule as in \bitcoin (but without the ``dampening filter''), there is still a difficulty raising attack. We using the difficulty of the first epoch (defined in the genesis block) as the difficulty unit and the expected inter-block time (10 minutes in \bitcoinnosp) as the time unit. Let $\Phi$ be the length of an epoch in number of blocks (2016 in \bitcoinnosp). And we define $n$ and $t$ the same as above.

Note that the adversary can put any timestamp in its private blocks, so the difficulty of the second epoch in its private chain can be arbitrary value as long as the adversary completes the first epoch. Let $B$ with difficulty $X$ be the first block of the second epoch in the private chain (that is each query solves the PoW puzzle with probability $1/nX$), then $B$ has chain difficulty $\Phi +X$. See Figure ~\ref{fig:bahack} for illustration. To mine an honest chain with chain difficulty $\Phi +X$, on average it takes $\Phi + X$ time. On the other hand, it takes on average $n\Phi/t$ time for the adversary to complete the first epoch in its private chain. Therefore, to succeed in this attack, the adversary needs to mine the block $B$ within $\Phi + X - n\Phi/t$ time, which happens with probability:
\begin{align*}
    \prob{({\rm attack~succeeds})} &= 1 - (1-\frac{1}{nX})^{(\Phi + X - n\Phi/t)t} \\
                                 &= 1 - (1-\frac{1}{nX})^{Xt - (n-t)\Phi} \\
                                 &\approx 1- e^{t/n},
\end{align*}
if $X \gg \Phi \gg 1$. Note that the success probability is independent of the length of the public longest chain, hence any $k$-deep confirmation rule will fail.

\begin{figure}
    \centering
    \includegraphics[width=8cm]{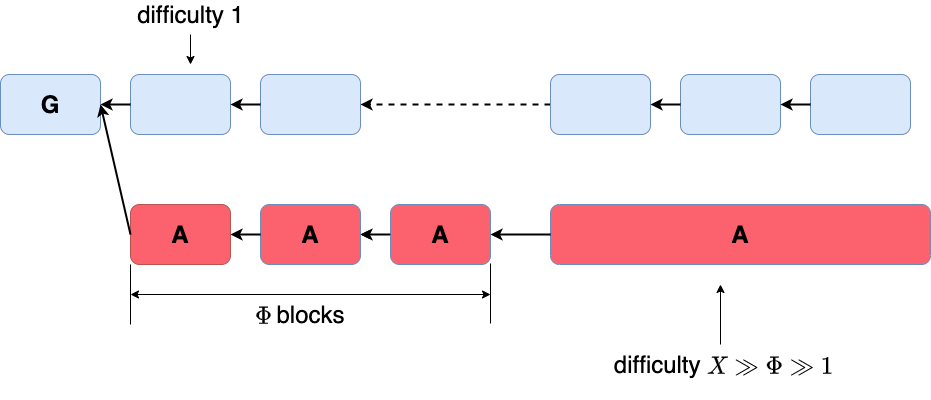}
    \caption{The difficulty raising attack. The adversary raises the difficulty to extremely high in the second epoch by faking timestamps.}
    \label{fig:bahack}
\end{figure}

However, \bitcoin is saved by the dampening filter in the target recalculation function. As in Definition \ref{def:targetrecalc}, the difficulty can
be increased by a factor of at most $\tau$ between two consecutive epochs ($\tau = 4$ in \bitcoinnosp). Then we shall analyze the difficulty raising attack under the same assumptions made above. Since the epoch size $\Phi \gg 1$, the time for the adversary to complete one epoch or mine $\Phi$ blocks with the same difficulty will satisfy the concentration bound of binomial random variables. Hence if the adversary always rises the difficulty by $\tau$ in each epoch, then it takes on average $\frac{n}{t}\sum_{i=0}^{\ell-1} \tau^i \Phi$ time for the adversary to complete $\ell$ epochs in its private chain, and the public honest chain will on average have difficulty $\frac{n}{t}\sum_{i=0}^{\ell-1} \tau^i \Phi$ during this time. Since the private chain has chain difficulty $\sum_{i=0}^{\ell-1} \tau^i \Phi$, the gap of chain difficulties between the public honest chain and the private chain will be
$$(\frac{n}{t}-1) \sum_{i=0}^{\ell-1} \tau^i \Phi = (\frac{n}{t}-1)\frac{\tau^\ell - 1}{\tau-1}\Phi.$$
Each block of the $(\ell+1)$-th epoch in the private chain will have difficulty $\tau^\ell$, hence the adversary still needs to mine approximately $\frac{n-t}{t(\tau-1)}\Phi$ blocks in order to catch up the honest chain. As $\Phi \gg 1$, the time for the adversary to catch up is still controlled by the concentration bound, and the success probability of this attack will be at most $e^{-\theta(\Phi)}$. By setting $\Phi$ large enough, the difficulty raising attack can be ruled out. 

While this specific attack could in principle be thwarted, to have security guarantee we still need to consider all possible attacks in the presence of a full-blown adversary. A full and beautiful analysis of \bitcoin rule is provided in \cite{full2020} and we shall give a proof sketch in Appendix~\ref{app:backbone}.
\section{Bitcoin Backbone Properties Revisited}
\label{app:backbone}
 
 We will briefly revisit the analysis in \cite{full2020} because the pivot chain is identical to the \bitcoin chain.

We will additionally define a stale chain and accuracy related to timestamps of the blocks.
 \begin{defn}[from \cite{full2020}]
 A block created at round $u$ is accurate is it has a timestamp $v$ such that $|u-v|\leq \ell + 2\Delta$. A chain is accurate if all its blocks are accurate. A chain is stale if for some $u\geq \ell+2\Delta$ it does not contain a honest block with timestamp $v \geq u - \ell -2\Delta$.
 \end{defn}

Recall that we define $\mathcal{S}^P_r$ as the set of pivot chains that belong to or have the potential to be adopted by an honest party at round $r$ in Section \ref{sec:def}. Now we define a series of useful predicates with respect to $\mathcal{S}^P_r$.

\begin{defn}[from \cite{full2020}]
For a round $r$,\\
$GoodRounds(r):=$ ``All rounds $u \leq r$ are good.''\\
$GoodChains(r):=$ ``For all rounds $u \leq r$, every chain in $\mathcal{S}^P_u$ is good.''\\
$NoStaleChains(r):=$ ``For all rounds $u \leq r$, there is no stale chain in $\mathcal{S}^P_u$.''\\
$Accurate(r):=$ ``For all rounds $u \leq r$, all chains in $\mathcal{S}^P_u$ are accurate.''\\
$Duration(r):=$ ``For all rounds $u \leq r$ and duration$\Lambda$ of an epoch of any chain in $\mathcal{S}^P_u$, $\frac{1}{2(1+\delta)\gamma^2}\frac{m}{f} \leq \Lambda \leq 2(1+\delta)\gamma^2 \frac{m}{f}$.''\\
\end{defn}

The following lemma provides a lower bound on the progress of the honest parties, which holds irrespective of any adversary.

\begin{lem}[Chain growth for pivot chain, from \cite{full2020}]
\label{lem:chaingrowth}
Suppose that at round $u$ of an execution $E$, an honest party broadcasts a pivot chain of difficulty $d$. Then, by round $v$, every honest party receives a chain of difficulty at least $d + Q^{P}(S)$, where $S = \{ r:u + \Delta \leq r \leq v - \Delta \}$.
\end{lem}

In order to prove properties like common prefix and chain quality for the pivot chain, we need all rounds in a typical execution to be good.
\begin{lem}[All rounds in a typical execution are good, Theorem 2 from \cite{full2020}]
\label{lem:allgoodrounds}
Consider a typical execution in a $(\gamma,2\gamma^2(1+\delta)\Phi/f)$-respecting environment. If the protocol is initiated such that the first round it good, and all the conditions \ref{eq:re1}, \ref{eq:re2} and \ref{eq:re3} are satisfied, then all rounds are good. 

\end{lem}
\begin{proof}
The elaborate proof can be found in \cite{full2020} and we summarize it as follows.

We will use an induction argument. In a $(\gamma, s)$-respecting environment, $s \geq 2(1+\delta)\gamma^2m/f$ covers at least the first epoch. It is easy to see that if the initial target is good, the rounds in the first epoch are good, and the first target recalculation point is good. We will prove that the subsequent rounds and target recalculation points are good using an induction argument shown in Figure \ref{fig:goodroundsproof}. The predicates are defined as follows.

\begin{figure}
    \centering
    \includegraphics[width=8cm]{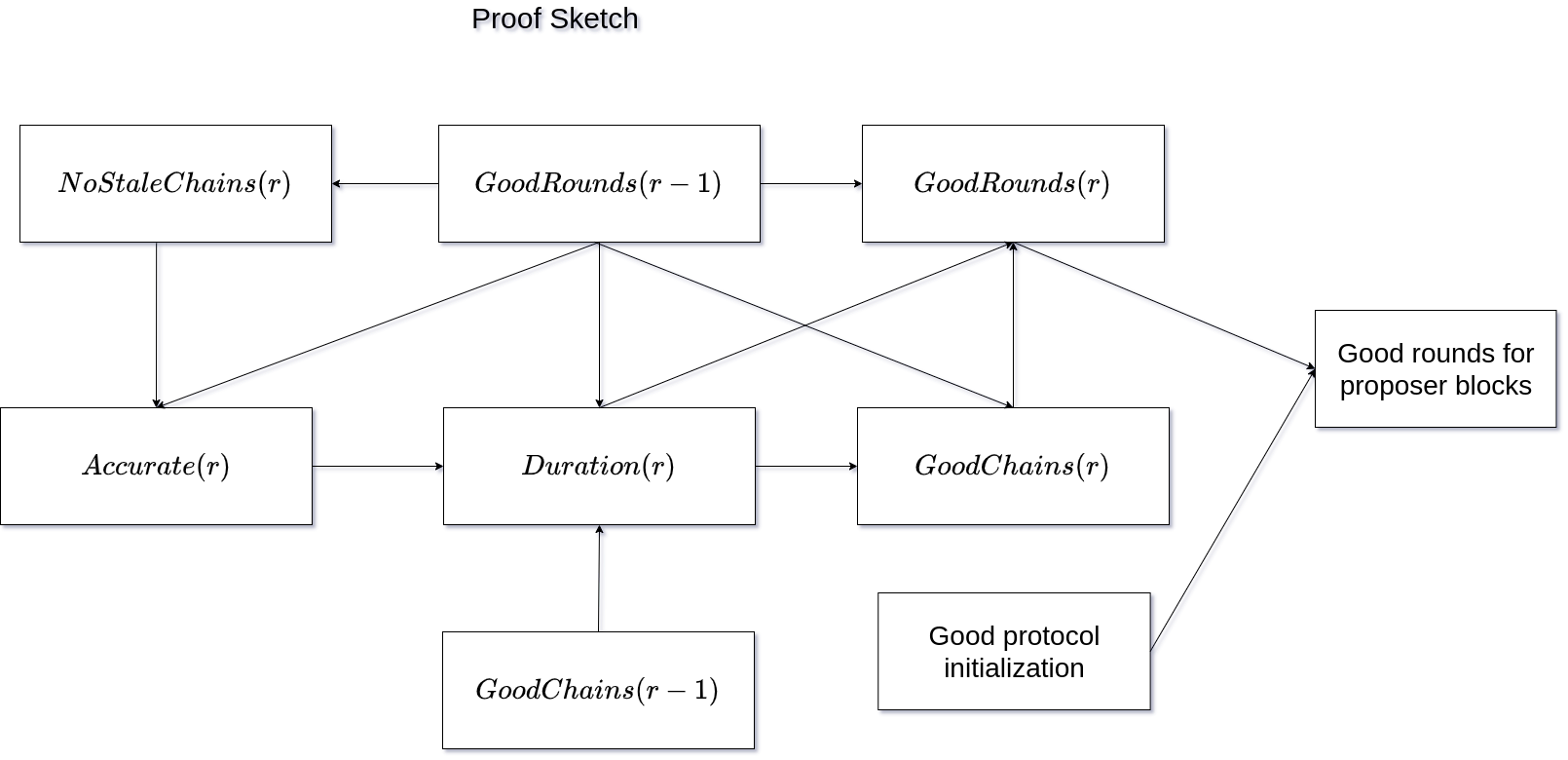}
    \caption{An induction argument to prove that all rounds in a typical execution are good.}
    \label{fig:goodroundsproof}
\end{figure}
We prove $NoStaleChains(r)$ from $GoodRounds(r-1)$ using typicality bounds, showing that the adversary cannot accumulate more difficulty than the lower bound of the minimum chain growth, $Q^P$. Let $w$ be the timestamp of the last honest block on the stale chain. Set $U= \{u: w \leq u \leq r\}$, $S= \{u: w+\Delta \leq u \leq r-\Delta\}$ and $J$ be the adversarial queries in $U$. We will first consider the case where the chain has more than one target recalculation point. In this case we divide $J$ into sub-queries $J_i$ such that each subset covers at least $m/2$ blocks and has exactly one target recalculation point in it. In this case, we have $A^P(J) = \sum_{i}A^P(J_i) < \sum_{i} (1+\epsilon)p|J_i| = (1+\epsilon)p|J|$. We arrive at a contradiction by showing  $(1+\epsilon)p|J|$ is lower than $Q^P(S)$'s lower bound. In case of at most one target recalculation point, if $A(J)<(1+\epsilon)p|J|$ applies, the argument from the previous case applies. If $A(J) < (1+1/\epsilon)(1/3 +1/\epsilon)\lambda\tau/T(J)$, we prove that the lower bound of $Q^P(S)$ considering only the first $l$ rounds will cover $(1+1/\epsilon)(1/3 +1/\epsilon)\lambda\tau/T(J)$. $Accurate(r)$ follows from $NoStaleChains(r)$. 

We then prove the bound on duration by contradiction, assuming that the previous target recalculation point is good using property $GoodChains(r-1)$. The lower bound is contradicted by showing that even if the adversary and honest party join forces, they can't produce $e$ blocks in less than $\frac{1}{2(1+\delta)\gamma^2}\frac{m}{f}$. The upper bound is contradicted by showing that the lower bound of $Q^P$ produces at least $e$ blocks in $2(1+\delta)\gamma^2 \frac{m}{f}$ rounds. To prove $GoodChains(r)$, we prove that the next target-recalculation point is good. This is proved again, using a contradiction for both the bounds of a good target recalculation point. Finally, we use $Duration(r)$, $GoodChains(r)$ and the $(\gamma, 2(1+\delta)\gamma^2m/f)$-respecting environment assumption to prove $GoodRounds(r)$.
\end{proof}

The following lemma from \cite{full2020} is useful to prove common prefix and chain quality of the pivot chain. 
\begin{lem}[Lemma 2(c) from \cite{full2020}]
\label{lem:proptypicalrel}
Consider a typical execution in a $(\gamma,s)$-respecting environment. Let $S = \{r: u \leq r \leq v\}$ be a set of rounds with at least $\ell$ rounds and $J$ be the set of adversary queries in $U= \{r: u-\Delta \leq r \leq v+\Delta\}$. If $w$ is a good round such that $|w-r|\leq s$ for any $r \in S$, then $A^P(J) < (1-\delta +3\epsilon)Q^P(S)$
\end{lem}

The following properties of the pivot chain are from \cite{full2020}. 

\begin{lem}[Common prefix for pivot chain]
\label{lem:proposer-prefix}
For a typical execution in a $(\gamma, 2(1 + \delta)\gamma^2 \Phi/f)$-respecting environment, the pivot chain satisfies the common-prefix property with parameter $\ell_{\rm cp} = \ell +2\Delta$.
\end{lem}

\begin{lem}[Chain quality for pivot chain]
\label{lem:proposer-quality}
For a typical execution in a $(\gamma, 2(1 + \delta)\gamma^2 \Phi/f)$-respecting environment, the pivot chain satisfies the chain-quality property with parameter $\ell_{\rm cq} = \ell +2\Delta$ and $\mu = \delta - 3\epsilon$.
\end{lem}

\section{Proof for Section 6}
\label{app:proof}

\subsection{Proof for typical execution}
\label{app:typicalproof}
The following concentration bound on a martingale is helpful to bound the probability of a not typical execution.
\begin{thm}[from \cite{full2020}] 
\label{thm:martingale}
Let $(X_1,X_2,\ldots)$ be a martingale with respective the sequence $(Y_1,Y_2,\ldots)$, if an event $G$ implies $X_k - X_{k-1} \leq b$ and $V = \sum_{k}var[X_k - X_{k-1}| Y_1, \ldots, Y_{k-1}] \leq v$, then for non-negative $n$ and $t$
\begin{align*}
    P(X_n -X_0 \geq  t, G ) \leq e^{-\frac{t^2}{2v + \frac{2bt}{3}}}.
\end{align*}
\end{thm}

And for the proof of Theorem \ref{thm:votertypicality}
\begin{proof}
The proof for $Q^{P}(S), D^{P}(S)$ and $A^{P}(J) $ can be found in \cite{full2020} and the same proof follows for $Q^i(S)$ and $D^i(S)$. We will prove it for $A^i(J)$.
For each $j \in J$, let $A_j$ be the difficulty of the block obtained with the $j^{th}$ query as long as the target was at least $1/b^i(J)$. Define 
\begin{align*}
    &X_0 = 0,\\
    &X_k = \sum_{j \in [k]} A_j - \sum_{j \in [k]} \mathbb{E}[A_j| \mathcal{E}_{j-1}], k \in [|J|],
\end{align*}
which is a martingale with respect to the sequence $\mathcal{E}_{j-1}, j\in J$. 
For the above martingale, for all $k \in [|J|]$, we have $X_k - X_{k-1} \leq b^i(J)$, using the definition of $b^i(J)$ and $ var[X_k - X_{k-1}] \leq pb^i(J)$ and $\mathbb{E}[A_j | \mathcal{E}_{j-1}] \leq p$. We will apply Theorem \ref{thm:martingale} with $t =\max\{\epsilon p|J|,\\
b^i(J)\lambda(\frac{1}{\epsilon} + \frac{1}{3}) \} \geq b^i(J)\lambda(\frac{1}{\epsilon} + \frac{1}{3})  $ and $v = b^i(J)p|J|$ to obtain
\begin{align*}
    Pr[ \sum_{j \in J} A_j \geq p|J| + t ] \leq exp \{ -\frac{t}{2b^i(J)(\frac{1}{3}+ \frac{1}{\epsilon})}\} \leq e^{-\lambda}.
\end{align*}

\end{proof}

\begin{lem}[Proposition 2 from \cite{full2020}]
\label{lem:envbounds}
In a $(\gamma, s)$-respecting environment, let $U$ be a set of at most $s$ consecutive rounds and $S \subseteq U$ then, for any $n \in \{n_r : r \in U\}$ we have
\begin{align*}
    \frac{n}{\gamma} &\leq \frac{n(S)}{|S|} \leq \gamma n, \\
    n(U) &\leq (1 + \frac{\gamma |U \setminus S|}{|S|})n(S). 
\end{align*}
\end{lem}

\subsection{Proof of Lemma \ref{lem:voter-prefix}}
\label{app:proof1}
By the definition of typical execution, we have the following lemma that will be useful in the proof.
\begin{lem}
\label{lem:voterproprel}
Under a typical execution, for the set of rounds $S$ with $|S|\geq \ell$, let $Q^P(S)$ correspond to the pivot tree and $Q^i(S)$ correspond to any non-pivot tree then, $Q^i(S) > Q^P(S)(1-\epsilon)[1 - (1+\delta)\gamma^2 f]^\Delta / (1 + \epsilon)$.
\begin{proof}
This follows from the definition of typicality, we use the following inequalities 
\begin{align*}
     (1 + \epsilon)pn(S) &>Q^P(S), \\
    Q^i(S) &> (1 - \epsilon)[1 - (1+\delta)\gamma^2 f]^\Delta pn(S).
\end{align*}
\end{proof}
\end{lem}

The following proposition will be useful in the proof of non-pivot chain's common prefix.
\begin{prop}
\label{prop:typicalvoter}
In a typical execution, we have the following bound

\begin{align*}
    A^i(J) < (1 + \epsilon)p|J|
\end{align*}
for $p|J| \geq \frac{2b^i(J)\lambda}{\epsilon}(\frac{1}{3}+\frac{1}{\epsilon}).$
\begin{align*}
    A^i(J) < (1 + \epsilon)\frac{2b^i(J)\lambda}{\epsilon}(\frac{1}{3}+\frac{1}{\epsilon}) < \frac{(1 - \epsilon^2)(\frac{\epsilon}{3} +1)\epsilon \Phi}{8\gamma^5(1 + \delta)(1 + 3\epsilon)}\frac{b^i(J)}{\tau}
\end{align*}
for $p|J| < \frac{2b^i(J)\lambda}{\epsilon}(\frac{1}{3}+\frac{1}{\epsilon})$, the second inequality follows from the bound on $\ell$.

\end{prop}

For the proof of Lemma ~\ref{lem:voter-prefix}
\begin{proof}
Consider the $i^{th}$ non-pivot chain, suppose common prefix fails for two chains $\C_1$ and $\C_2$ held by honest players at rounds $r_1 \leq r_2$ respectively, that is, $\exists B \in \C_1^{\lceil \ell + 2\Delta}$, s.t. $B \notin \C_2$. It is not hard to see that in such a case there was a round $r \leq r_2$ and two honest held chains $\C$ and $\C'$ in $\mathcal{S}_r^i$, such that $B \in \C^{\lceil \ell + 2\Delta}$ but $B \notin \C'$. Then we know $B$ is a descendant of head$(\C \cap \C')$, and hence head$(\C \cap \C') \in  \C^{\lceil \ell + 2\Delta}$. Therefore, the timestamp of head$(\C \cap \C')$ is less than $r - \ell -2\Delta$. 

Let $v < r - \ell -2\Delta$ be the timestamp of head$(\C \cap \C')$ and $w\leq v$ be the timestamp of the last honest block $B^i_h$ on $(\C \cap \C')$. Let $U^i =\{u: w \leq u \leq r\}, S^i =\{u: w + \Delta \leq u \leq r - \Delta \}$ and let $J^i$ be the adversarial queries in rounds $U^i$.
Let $\mathcal{S}^P_{r,w-\Delta}$ be the collection of pivot chains heavier than at least one chain in $\mathcal{S}_{w-\Delta}$. And for $j \in J^i$, let $\mathcal{S}^P_{j,w-\Delta}$ be the collection of pivot chains heavier than at least one chain in $\mathcal{S}_{w-\Delta}$. Due to condition \textbf{M2}, all the difficulties of the blocks in $C$ or $C'$ that come after $B_0^i$ are extending $\mathcal{S}_{r,w-\Delta}^P$. 
We have $b=b^i(J) = \max_{j\in J^i} \sup\{A^i_j - A^i_{j-1}| \mathcal{E}_{j-1} = E_{j-1}\} = \max_{j\in J^i} sup\{A^i_j - A^i_{j-1}| \mathcal{E}_{j-1} = E_{j-1}\} = \max_{j\in J^i} \sup\{ {\rm diff}(C^PB^*) |C^* \in \mathcal{S}_{j,w-\Delta}  \}  = \sup_{C^P \in \mathcal{S}_{r,w-\Delta}}\{  {\rm diff}(C^PB^*)  \} $.
The last equality applies because, for $j \in J^i$, $\mathcal{S}_{j,w-\Delta} \subseteq \mathcal{S}_{r,w-\Delta}$. Let $C^* \in \mathcal{S}_{r,w-\Delta}$ be the chain for which ${\rm diff}(C^*B^*) = b$. In case such $C^*$ doesn't exist, there exists a sequence of chains $C^*_n$, such that ${\rm diff}(C^*_nB^*)$ approaches $ b$ in limit. Let the block $B^P_h$ be the last honest block on $C^*$ with timestamp $x$.

We claim that if $r > \ell + 2\Delta + w$, then $A^i(J^i) < (1 + \delta + 3\epsilon)Q^i(S^i)$. The proof is as follows.
When $p|J^i| \geq \frac{2b\lambda}{\epsilon}(\frac{1}{3}+\frac{1}{\epsilon})$, the concentration bound $A^i(J^i)<(1+\epsilon)p|J^i|$ applies. We have $n(U^i) \leq n(S^i)(1 + \gamma|U\setminus S|/|S|)<(1 + \epsilon^2/2)n(S)$ and
\begin{align*}
    A^i(J^i)&<(1+\epsilon)(1-\delta)pn(U^i) <(1+\epsilon)(1 + \epsilon^2/2)(1-\delta)pn(S^i)\\  
    &< (1-\delta + \epsilon)pn(S^i)< (1 -\delta + 3\epsilon)Q^i(S^i) 
\end{align*}
We will prove this when $p|J^i| < \frac{2b\lambda}{\epsilon}(\frac{1}{3}+\frac{1}{\epsilon})$.

\begin{center}
    \begin{tabular}{ | p{0.8cm} | p{7cm}|  } 
        \hline
        Case 1 &  $C^*$ has at most one target-recalculation point after $B_h^P$ and $w \leq x \leq r \Rightarrow f/2\gamma^2\tau < \frac{1}{b}n_x p$\\ 
        \hline
        Case 2 & $C^*$ has at least two target-recalculation point after $B_h^P$ and $w \leq x \leq r$  \\ 
        \hline
        Case 3 & $x<w, |w-x| > \gamma^2 (1 + \delta)\Phi/f -\ell -2 \Delta$
        \\
        \hline
        Case 4 & $x<w, |w-x| < \gamma^2 (1 + \delta)\Phi/f -\ell -2 \Delta$, and $C^*$ has at most one target-recalculation point after $B_h^P$ 
        \\
        \hline
        Case 5 & $x<w, |w-x| < \gamma^2 (1 + \delta)\Phi/f -\ell -2 \Delta$, and $C^*$ has at least two target-recalculation point after $B_h^P$
        \\
\hline
    \end{tabular}
\end{center}

\begin{figure}
    \centering
    \includegraphics[width=\columnwidth]{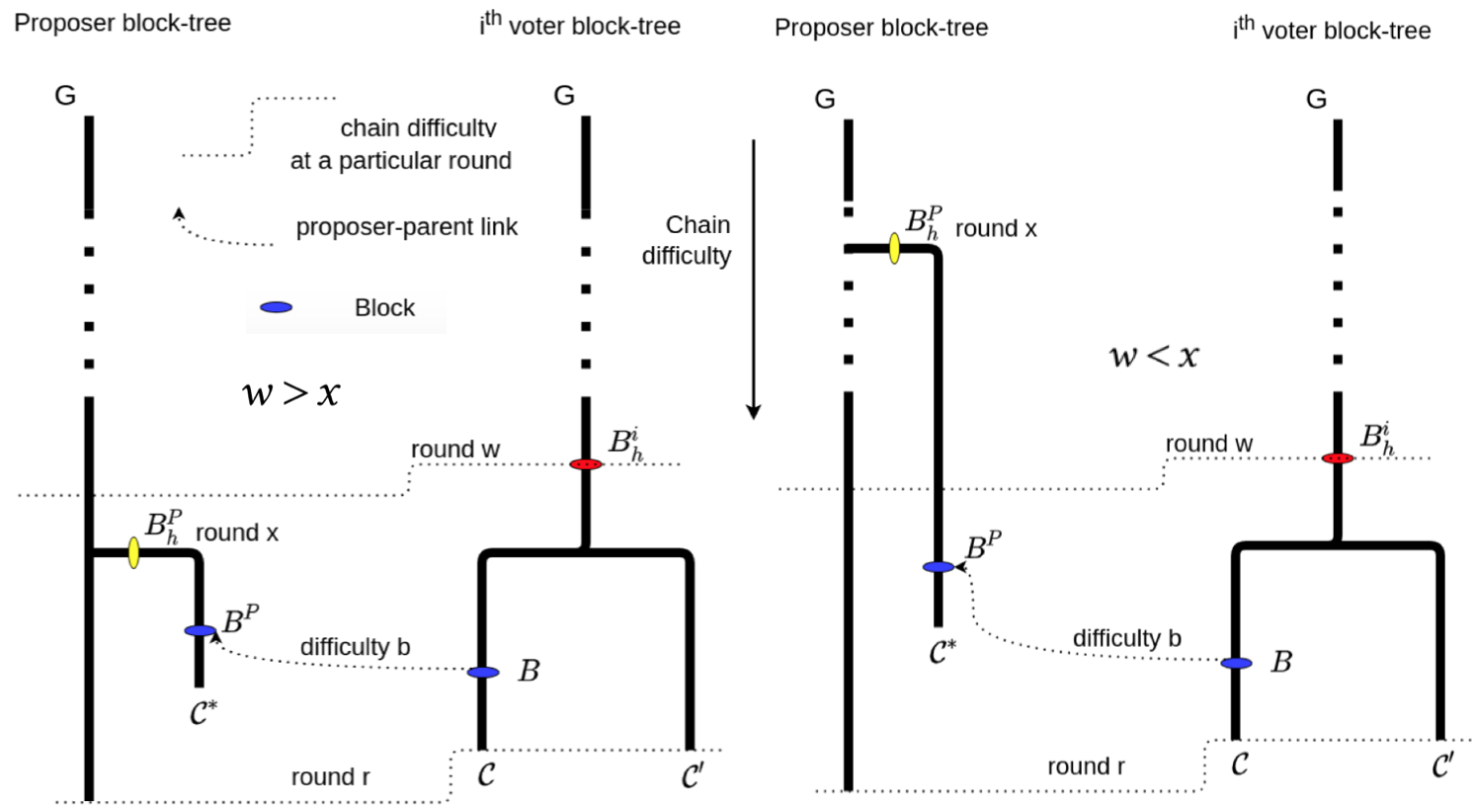}
    \vspace{-2mm}
    \caption{Common Prefix Proof (Left): $w<x$, (Right): $w>x$}
    \label{fig:vchaincp1}
\end{figure}


Cases 1, 2 are shown in left and cases 3,4,5 in right of  Figure \ref{fig:vchaincp1}.

{\bf \noindent Case 1:} The last honest block $B^P_h$ in the chain $C^*$ has a timestamp $x\geq w$. We will look at the case when $C^*$ has at most one target recalculation point after $B^P_h$. In this case the difficulty $b$ satisfies  $\frac{f}{2\gamma^2\tau} < \frac{1}{b}n_x p$ since the difficulty can raise by at most a factor of $\tau$ and considering the first $\ell$ rounds in $S^i$, we have ${n(S^i)} > \frac{n_x}{\gamma}\ell$. Using typicality we have, $p|J| \leq (1-\delta + \epsilon^2/2)pn(S)$ and 
\begin{align*}
    \epsilon(1-2\epsilon)pn(S^i) > \epsilon(1-2\epsilon) \frac{pn_x\ell b}{\gamma b} > \frac{\epsilon (1-2\epsilon) f\ell b}{2\gamma^3\tau} \geq 2b\lambda(\frac{1}{\epsilon} + \frac{1}{3}),\\
    A^i(J^i) < p|J| + 2b\lambda(\frac{1}{\epsilon} + \frac{1}{3}) \leq (1 - \delta +\epsilon)pn(S^i) < (1 -\delta + 3\epsilon)Q^i(S^i)
\end{align*}

{\bf \noindent Case 2:} The last honest block $B^P_h$ in the chain $C^*$ has a timestamp $x\geq w$. We will look at the case when $C^*$ has more than one target recalculation point after $B^P_h$. Let $U^P$  be the set of rounds $\{ u: x \leq u \leq r\}$, $S^P$ be the set of rounds $\{ u: x +\Delta \leq u \leq r - \Delta\}$ and $J^P$ be the queries made by the adversary for the proposer chain in $U^P$. In this case difficulty accumulated by the adversary in $J^P$ queries is at least $\frac{b}{\tau} \Phi$. Using typicality, we have $p|J^P| > \frac{b\Phi}{\tau (1 + \epsilon)}$ and using honest party's advantage we have $n(S^i) \geq n(S^P) > \frac{|J^P|}{(1-\delta)(1 + \epsilon^2/2)}  $. 

\begin{align*}
    (1 -\delta +3\epsilon)Q^i(S^i) &> (1-\epsilon)[1 - (1+\delta)\gamma^2 f]^\Delta \frac{(1 -\delta +3\epsilon)p|J^P|}{(1-\delta)(1 + \epsilon^2/2)} \\
    &> \frac{b\Phi}{\tau} \geq A^i(J^i).
\end{align*}
The last inequality implies from Proposition \ref{prop:typicalvoter}.

{\bf \noindent Case 3:} Consider the case when $x<w$ and $|w-x|>s/2 - \ell  - 2\Delta$. Let $S' := \{u: x+\Delta \leq u \leq w - \Delta \} $, $U^P = \{u: x \leq u \leq r\}$ and $J^P$ be the set of adversarial queries for the proposal tree in the rounds $U^P$. The difficulty accumulated in the chain $C^*$ in $J^P$ queries is more than that of the chain growth in $S'$.
\begin{align*}
    A^P(J^P) \geq ChainGrowth^P(S') \geq Q^P(S'). 
\end{align*}
Since $s = \frac{2\gamma^2(1+\delta)\Phi}{f}$, we have $|S'| \geq (1+\delta)(1-\epsilon)\gamma^2\Phi/f$. 
Considering the first $s/2 - \ell  $ rounds in $U^P \setminus U^i$  , if $T_x$ is the target used by the honest party in round $x$, then $\frac{n(S')}{|S'|} \geq \frac{n_x}{\gamma}$ and $ T_xn_xp \geq \frac{f}{2\gamma^2}$. Using these, we have
\begin{align*}
    Q^P(S') &> (1-\epsilon) [1 - (1+\delta)\gamma^2 f]^\Delta pn(S') \\
    &\geq (1+\delta)(1-\epsilon)^3\frac{\gamma \Phi pn_xT_x}{2fT_x} \geq (1+\delta)(1-\epsilon)^3\frac{\Phi}{2\gamma T_x} > \frac{\Phi}{2T_x \gamma} 
\end{align*}
Note that starting with target $T_x$, if $C^*$ has at most one target recalculation point after $B_h^P$, then the accumulated difficulty is at least $\frac{\Phi}{2 \gamma}\frac{b}{\tau}$, which is a smaller quantity than $\frac{\Phi}{2T_x \gamma}$. If the chain has more than one target recalculation point, then the accumulated difficulty is at least $m\frac{b}{\tau}$ which is larger than $\frac{\Phi}{2 \gamma}\frac{b}{\tau}$. Hence, the accumulated difficulty will be at least $ \frac{\Phi}{2 \gamma}\frac{b}{\tau}$ in any case. 
\begin{align*}
     |J^P|p(1+\epsilon) > A^P(J^P) \geq  Q^P(S'), \\
     A^P(J^P)> \frac{\Phi}{2 \gamma}\frac{b}{\tau} 
\end{align*}
We have $n(S^i) + n(S')> \frac{|J^P|}{(1-\delta)(1 + \epsilon^2/2)}$ and
\begin{align*}
    Q^i(S^i) + Q^i(S') &> (1-\epsilon)[1 - (1+\delta)\gamma^2 f]^\Delta \frac{p|J^P|}{(1-\delta)(1 + \epsilon^2/2)},\\ 
    Q^i(S') &< Q^P(S') \frac{(1+\epsilon)}{(1-\epsilon)[1 - (1+\delta)\gamma^2 f]^\Delta} \\
    &< \frac{(1+\epsilon)^2}{(1-\epsilon)[1 - (1+\delta)\gamma^2 f]^\Delta} p|J^P| 
\end{align*}
Combining both we have,
\begin{align*}
    Q^i(S^i) > p|J^P|\frac{((1-\epsilon)[1 - (1+\delta)\gamma^2 f]^\Delta)^2 - (1+\epsilon)^2(1-\delta)(1 + \epsilon^2/2)}{(1-\epsilon)[1 - (1+\delta)\gamma^2 f]^\Delta (1-\delta)(1 + \epsilon^2/2)},
\end{align*}
and then
\begin{align*}
    &(1 -\delta + 3\epsilon)Q^i(S^i) \\
    >& (1 + 3\frac{\epsilon}{1-\delta})\frac{(\delta (1+\epsilon^2/2)(1+\epsilon)^2 - 6\epsilon)}{(1-\epsilon^2)[1 - (1+\delta)\gamma^2 f]^\Delta (1+\epsilon^2/2)}\frac{\Phi}{2 \gamma} \frac{b}{\tau}\\
    >& A^i(J^i),
\end{align*}
Where the last inequality follows from the condition $\delta > 8\epsilon$. 

{\bf \noindent Case 4:} Consider the case the last honest block in the chain containing $B$'s proposer parent has a timestamp $x<w$ and $|w-x|<s/2 - \ell  - 2\Delta$ and $C^*$ has at most one target recalculation point after $B^P_h$. Let $S' := \{u: x+\Delta \leq u \leq w - \Delta \} $. The difficulty accumulated $C^*$ in $J^P$ queries is more than that of the chain growth in $S'$. Considering just first $\ell $ rounds in $S^i$, we have $n(S^i) >\ell n_x/\gamma$ and $b$ satisfies  $\frac{f}{2\gamma^2\tau} < \frac{1}{b}n_x p$. Using these bounds and Lemma \ref{lem:chaingrowth}, we have

\begin{align*}
    \epsilon(1-2\epsilon)pn(S^i) > \epsilon(1-2\epsilon) \frac{pn_x\ell b}{\gamma b} > \frac{\epsilon (1-2\epsilon) f\ell b}{2\gamma^3\tau} \geq 2b\lambda(\frac{1}{\epsilon} + \frac{1}{3}),\\
    A^i(J^i) < p|J| + 2b\lambda(\frac{1}{\epsilon} + \frac{1}{3}) \leq (1 - \delta +\epsilon)pn(S^i) < (1 -\delta + 3\epsilon)Q^i(S^i)
\end{align*}

{\bf \noindent Case 5:} Consider the case the last honest block in the chain containing $B$'s proposer parent has a timestamp $x<w$ and $|w-x|<s/2 - \ell  - 2\Delta$. Let $S' := \{u: x+\Delta \leq u \leq w - \Delta \} $. The difficulty accumulated by $C^*$ in $J^P$ queries is more than that of the chain growth in $S'$. We will consider the case $C^*$ has more than one target recalculation point after $B^P_h$.  The adversary accumulates more than $\frac{b}{\tau}\Phi$ difficulty in $J^P$ queries and similar to \textbf{Case 4}, we have
\begin{align*}
     & |J^P|p(1+\epsilon) > A^P(J^P) \geq   \Phi\frac{b}{\tau},\\
     & A^P(J^P) \geq   Q^P(S'), \\
     & n(S^i) + n(S')> \frac{|J^P|}{(1-\delta)(1 + \epsilon^2/2)},
\end{align*}
and 
 \begin{align*}
    Q^i(S^i) + Q^i(S') &> (1-\epsilon)[1 - (1+\delta)\gamma^2 f]^\Delta \frac{p|J^P|}{(1-\delta)(1 + \epsilon^2/2)},\\ 
    Q^i(S') &< Q^P(S') \frac{(1+\epsilon)}{(1-\epsilon)[1 - (1+\delta)\gamma^2 f]^\Delta}\\ &< \frac{(1+\epsilon)^2}{(1-\epsilon)[1 - (1+\delta)\gamma^2 f]^\Delta} p|J^P| 
\end{align*}
Combining both we have
\begin{align*}
    Q^i(S^i) &> p|J^P|\frac{((1-\epsilon)[1 - (1+\delta)\gamma^2 f]^\Delta)^2 - (1+\epsilon)^2(1-\delta)(1 + \epsilon^2/2)}{(1-\epsilon)[1 - (1+\delta)\gamma^2 f]^\Delta (1-\delta)(1 + \epsilon^2/2)},\\
\end{align*}
and then
\begin{align*}
    &(1 -\delta + 3\epsilon)Q^i(S^i) \\
    >& (1 + 3\frac{\epsilon}{1-\delta})\frac{(\delta (1+\epsilon^2/2)(1+\epsilon)^2 - 6\epsilon)\Phi b}{(1-\epsilon^2)[1 - (1+\delta)\gamma^2 f]^\Delta (1+\epsilon^2/2)\tau} \\
    >& A^i(J^i),
\end{align*}
The last inequality follows from the condition $\delta > 8\epsilon$.

We also claim that, if $r-w > \ell + 2\Delta$, then $2Q^i(S^i) \leq D^i(U^i) + A^i(J^i)$, which leads to a contradiction as $D^i(U^i) < (1+5\epsilon)Q^i(S)$ and $A^i(J^i) < (1 -\delta +3\epsilon)Q^i(S^i)$.

Towards proving the claim above, associate with each $r \in S$ such that $Q^i_r > 0$ an arbitrary honest block that is computed at round $r$ for difficulty $Q^i_r$. Let $\mathcal{B}$ be the set of these blocks and note that their difficulties sum to $Q^i(S)$. Then consider a block $B \in \mathcal{B}$ extending a chain $\C^*$ and let $d = {\rm diff}(\C^*B)$. If $d \leq {\rm diff}(\C \cap \C')$ (note that $u < v$ in this case and head$(\C \cap \C')$ is adversarial), let $B_0$ be the block in $\C \cap \C'$ containing d. Such a block clearly exists and has a timestamp greater than $u$. Furthermore, $B_0 \notin \mathcal{B}$, since $B_0$ was an adversarial block. If $d > {\rm diff}(\C \cap \C')$, note that there is a unique $B \in \mathcal{B}$ such that $d \in B$. Since $B$ cannot simultaneously be on chain $\C$ and $\C'$, there is a $B_0 \notin \mathcal{B}$ either on $\C$ or on $\C'$ that contains $d$. Hence there exists a set of blocks $\mathcal{B}'$ computed in $U$ such that $\mathcal{B} \cap \mathcal{B}' = \empty$ and $\{d \in B: B \in \mathcal{B}\} \subseteq \{d \in B : B \in \mathcal{B}'\}$. Because each block in $\mathcal{B}'$ contributes either to
$D^i(U) - Q^i(S)$ or to $A^i(J)$, we have $Q^i(S^i) \leq D^i(U^i) - Q^i(S) + A^i(J^i)$.

\end{proof}

\subsection{Chain Quality of Non-pivot Chains}
\label{app:proof2}

\begin{proof}[Proof of Lemma \ref{lem:voter-quality}]
Without loss of generality, we focus on the first non-pivot chain. 
Let $B_i$ denote the $i$-th block of $\C$ and consider $K$ consecutive blocks $B_u, \cdots, B_v$ in $\C$ with timestamp in $S_0$. Define $K_0$ as the least number of consecutive blocks $B_{u'},\cdots, B_{v'}$ that include the $K$ given ones (i.e., $u' \leq u$ and $v \leq v'$) and have the properties (1) that the block $B_{u'}$ was mined by an honest party at some round $r_1$ or is the genesis block in case such block does not exist, and (2) that there exists a round $r_2$ such that the chain ending at block $B_{v'}$ is adopted by some honest node at round $r_2$. Let $d'$ be the total difficulty of these $K'$ blocks. Define $U = \{r_1,\cdots, r_2\}$, $S = \{r_1 + \Delta,\cdots, r_2 - \Delta\}$, and $J$ the adversarial queries in $U$ associated with the $K'$ blocks. Then we have $|S| = |U| - 2\Delta \geq |S_0| -2\Delta \geq \ell$. Then following the same argument from Lemma~\ref{lem:voter-prefix}, we have $A^1(J) < (1 - \delta + 3\epsilon)Q^1(S)$.
Let $x$ denote the total difficulty of all the blocks from honest parties that are included in the $K$ blocks and—towards a contradiction—assume $x < \mu d \leq \mu d'$. In a typical execution, all the $K'$ blocks have been mined in $U$. But then we have the following contradiction
\begin{equation*}
    A^1(J) \geq d' - x > (1 - \mu)d' \geq (1 - \mu)Q^1(S) = (1 - \delta + 3\epsilon)Q^1(S).
\end{equation*}
Therefore, we can conclude the proof.
\end{proof}

\subsection{Common Prefix and Chain Quality of the Leader Sequence}
\label{app:proof3}

\begin{proof}[Proof of Lemma \ref{lem:leader-prefix}]
Let $r \geq R_d + 2\ell +4\Delta$ be the current round. For $1 \leq i \leq m$, let $\C_i$ be the heaviest voter chain $i$ in an honest node $u$'s view at round $r$. By the common prefix property in Lemma~\ref{lem:voter-prefix}, blocks in $\C_i^{\lceil \ell +2\Delta}$ remain unchanged until $\rmax$. In addition, by the chain quality property in Lemma~\ref{lem:voter-quality}, we know that for $1 \leq i \leq m$, there exists at least one honest block $B_i$ on chain $\C_i$ whose timestamp is in the interval $(r-2\ell-4\Delta, r-\ell-2\Delta)$,i.e., $B_i$ is on the chain $\C_i^{\lceil \ell +2\Delta}$. As $B_i$ is an honest block mined after $R_d$, $B_i$ or an ancestor of $B_i$ must have voted for the difficulty level $d$. Therefore the leader sequence remains unchanged up to difficulty level $d$ until $\rmax$.
\end{proof}

\begin{proof}[Proof of Lemma \ref{lem:leader-quality}]
Let $r$ be the current round, $\C$ be the proposer chain held by honest player $P$, and $d = \rm{diff}(\C)$. Let interval $D = (d',d]$ be the difficulty range covered by all blocks in $\C$ with timestamp in last $\ell + 2\Delta$ rounds. Define:
\begin{align*}
    d^* &:= \max \big( \tilde{d} \leq d' \;\;s.t\;\; \text{ the honest players mined} \\
     &\text{the first proposer block covering~} \tilde{d}\big)
\end{align*}
Let $r^*$ be the round in which the first proposer block covering $d^*$ was mined. $r^* = 0$ and $d^* = 0$ if such proposer block does not exists. Define $U = \{r^*,\cdots, r\}$, $S = \{r^* + \Delta,\cdots, r - \Delta\}$, and $J$ the adversarial queries in $U$. Then we have $|S| = |U| - 2\Delta \geq \ell$. From the definition of $d^*$ we have the following two observations:
\begin{enumerate}
    \item All difficulties in $(d^*,d']$ are covered by at least one adversarial proposer block.
    \item All the proposer blocks covering $(d^*,d]$ are mined in the interval $U$ because there are no proposer blocks covering $d^*$ before round $r^*$ and hence no player can mine a proposer block covering a difficulty level greater than $d^*$ before round $r^*$.
\end{enumerate}

Let $L_h$ be  the size of difficulty range covered by honest leader blocks in the range $(d',d]$ and say
\begin{equation}
    L_h < \mu (d-d') \leq \mu(d-d^*). \label{eqn:local_contradition_1}
\end{equation}
Let $L_h'$ be  the size of difficulty range covered by honest leader blocks in the range $(d^*,d']$. The adversarial leader blocks have covered difficulty ranges with size $d-d^* - L_h - L_h'$ in the interval $U$. From our first observation, we know that adversarial proposer blocks in the difficulty range $[d^*,d']$ which are \textit{not} leader blocks cover difficulty ranges with size at least $L_h'$, and from our second observation, these proposer blocks are mined in the interval $U$. Therefore, we have the following bound on $A^P(J)$
\begin{align}
    A^P(J)&\geq(d-d^*-L_h-L_h')+L_h' \nonumber \\
     &= d-d^*-L_h \nonumber \\
     \left(\text{From Equation } \eqref{eqn:local_contradition_1}\right)& > d-d^* - \mu(d-d^*) \nonumber \\
     &= (1-\mu) (d-d^*). \label{eqn:local_454}
\end{align}
From the chain growth, we know that $d-d^* \geq Q^P(S)$ and combining this with Equation \eqref{eqn:local_454} gives us 
\begin{equation}
A^P(J) >  (1-\mu)Q^P(S) = (1 - \delta + 3\epsilon)Q^P(S), \label{eqn:chain_quality_contradiction1}   
\end{equation}
which contradicts Lemma \ref{lem:proptypicalrel}.
\end{proof}
\section{Persistence and Liveness of \ohie}
\label{app:ohie}
\begin{lem}[Common prefix and chain quality for each individual chain]
\label{lem:ohie-ind-quality}
For a typical execution in a $(\gamma, 2(1 + \delta)\gamma^2 m/f)$-respecting environment, each chain $i$ ($0 \leq i \leq m-1$) satisfy the the common-prefix property with parameter $\ell_{\rm cp} = \ell +2\Delta$ and chain-quality property with parameter $\ell_{\rm cq} = \ell +2\Delta$ and $\mu = \delta - 3\varepsilon$.
\end{lem}

\begin{proof}
The proof directly follows Lemmas \ref{lem:proposer-prefix},\ref{lem:proposer-quality}, \ref{lem:voter-prefix} and \ref{lem:voter-quality}.
\end{proof}

We call all blocks on a chain except the blocks with timestamp in last $\ell + 2\Delta$ rounds as partially confirmed. Then by Lemma \ref{lem:ohie-ind-quality}, we know that the ordering of these partially confirmed blocks on their chain will no longer change in the future.
Recall that \ohie generates a SCB in the following way. Consider any given honest node at any given time and its local view of all the $m$ chains. Let $y_i$ be the ${\tt next\_rank}$ of the last partially-confirmed block on chain $i$ in this view. Let ${\tt confirm\_bar} \leftarrow \min_{i=1}^k y_i$. Then all partially-confirmed blocks whose rank is smaller than ${\tt confirm\_bar}$ are deemed as fully-confirmed, and included in SCB. Finally, all the fully-confirmed blocks will be ordered by increasing ${\tt rank}$ values, with tie-breaking favoring smaller chain ids. Next, we will prove that the SCB generated in this way satisfies persistence and liveness properties.

\begin{thm}[Persistence and Liveness of \ohie]
\label{thm:ohie}
For a typical execution in a $(\gamma, 2(1 + \delta)\gamma^2 m/f)$-respecting environment, \ohie satisfies persistence and liveness with parameter $u = 2\ell +4\Delta$.
\end{thm}

\begin{proof}
We first prove the persistence property. 
Let $L_1$ and $L_2$ be the SCB held by two honest node $u_1$ and $u_2$ at round $r_1$ and $r_2$ respectively, with $r_2 > r_1 +\Delta$. Let $x_1$ and $x_2$ be the ${\tt confirm\_bar}$ of $L_1$ and $L_2$, respectively. Then we will prove that $L_1$ is a prefix of $L_2$. 

Let $F_1(i)$ be the sequence of partially-confirmed blocks on chain $i$ in $u_1$'s view at round $r_1$. Let $G_1(i)$ be the prefix of $F_1(i)$ such that $G_1(i)$ contains all blocks in $F_1(i)$ whose ${\tt rank}$ is smaller than $x_1$. Similarly define $F_2(i)$ and $G_2(i)$, where $G_2(i)$ contains those blocks in $F_2(i)$ whose ${\tt rank}$ is smaller than $x_2$. We first prove the following two claims:
\begin{itemize}
\item For all $i$, $0 \le i\le m-1$, $G_1(i)$ is a prefix of $G_2(i)$. By round $r_2$, node $u_2$ must have seen all blocks seen by $u_1$ at round $r_1$. By the common prefix property in Lemma~\ref{lem:ohie-ind-quality}, all partially-confirmed blocks on $u_1$ at round $r_1$ must also be partially-confirmed on $u_2$ at round $r_2$, i.e., $F_1(i)$ is a prefix of $F_2(i)$. Then we have $x_1 \leq x_2$ and also $G_1(i)$ should be a prefix of $G_2(i)$.
\item For all $i$, $0 \le i\le m-1$, and all block $B \in G_2(i)\setminus G_1(i)$, $B$'s ${\tt rank}$ must be no smaller than $x_1$. We prove this claim via a contradiction and assume that $B$'s ${\tt rank}$ is smaller than $x_1$. Together with the fact that $B$ is in $G_2(i)\setminus G_1(i)$, we know that $B$ is in $F_2(i)$ but not in $F_1(i)$. Let $x_0$ be the ${\tt rank}$ of the last block $B_0$ in $F_1(i)$. Since $F_1(i)$ is a prefix of $F_2(i)$, both $B_0$ and $B$ must be in $F_2(i)$, and $B$ must be a descendent of $B_0$ in $F_2(i)$. Since the blocks in $F_2(i)$ must have increasing ${\tt rank}$ values, we know that $B_0$'s ${\tt rank}$ must be smaller than $B$'s, hence we have $x_0 < x_1$. On the other hand, since $B_0$'s ${\tt next\_rank}$ must be greater than its ${\tt rank}$, by the design of ${\tt confirm\_bar}$, we know that $x_1 < x_0$. This yields a contradiction.
\end{itemize}

Now we can use the above two claims to prove that $L_1$ is a prefix of $L_2$. Note that $L_1 = \cup_{0 \leq i \leq m-1} G_1(i)$ and $L_2 = \cup_{0 \leq i \leq m-1} G_2(i)$. Since $G_1(i)$ is a prefix of $G_2(i)$ for all $i$, we know that $L_1\subseteq L_2$. For all blocks in $L_1$, the sequence $L_1$ orders them in exactly the same way as the sequence $L_2$. For all block $B\in L_2\setminus L_1$, by the second claim above, we know that $B$'s ${\tt rank}$ must be no smaller than $x_1$. Hence in the sequence $L_2$, all blocks in $L_2\setminus L_1$ must be ordered after all the blocks in $L_2\cap L_1$ (whose ${\tt rank}$ must be smaller than $x_1$). This completes our proof that $L_1$ is a prefix of $L_2$.

We next prove the liveness property. Suppose a transaction {\sf tx} is received by all honest nodes before or at round $r_0$. Let $L$ be the SCB of an honest node $u$ at round $r \geq r_0 + 2\ell +4\Delta$ and $x$ be the ${\tt confirm\_bar}$ of $L$. For $0\leq i \leq m-1$, let $F(i)$ be the sequence of partially-confirmed blocks on chain $i$ in $u$'s view at round $r$ and $G(i)$ be the prefix of $F(i)$ such that $G(i)$ contains all blocks in $F(i)$ whose ${\tt rank}$ is smaller than $x$. Without loss of generality, assume $x$ is the ${\tt next\_rank}$ of the last block in $F(0)$. Then we have $G(0) = F(0) \subset L$. By the chain quality property in Lemma~\ref{lem:ohie-ind-quality}, we know that there exists at least one honest block $B$ on chain 0 whose timestamp is in the interval $(r-2\ell-4\Delta, r-\ell-2\Delta)$. As $B$ is an honest block mined after $r_0$, $B$ or an ancestor of $B$ must contain $\sf tx$. Therefore $L$ contains $\sf tx$ since $B$ and all its ancestors are fully confirmed.
\end{proof}

\section{Fairness of \fruitchains}
\label{app:fruit}
For our analysis to go through, we choose the following protocol parameters for \fruitchainsnosp:
\begin{equation}
\label{eqn:ell}
    \ell \triangleq \frac{4(1+3\epsilon)}{\epsilon^2f[1 - (1+\delta)\gamma^2f]^{\Delta+1}} \max\{\Delta,\tau\}\gamma^4 \lambda,
\end{equation}
\begin{equation}
\label{eqn:phi}
    \Phi = 4(1+\delta)\gamma^2 f(\ell +3\Delta)/\epsilon,
\end{equation}
and
\begin{equation}
\label{eqn:s}
    s = 2(1 + \delta)\gamma^2 \Phi/f.
\end{equation}

Note that in Equation (\ref{eqn:ell}), we define a larger $\ell$ than in Equation (\ref{eq:re1}) for \scheme and \ohienosp. This is for the proof of fairness to work. Since Equation (\ref{eqn:phi}) satisfies Equation (\ref{eq:re2}) and $s$ is selected the same as in \scheme and \ohienosp, all the pivot-chain properties stated in Appendix \ref{app:backbone} still hold for \fruitchainsnosp.
Therefore, the persistence and liveness of variable difficulty \fruitchains simply follow the persistence and liveness of \bitcoinnosp, as the protocols are identical if we ignore the content of the blocks. In this section, we formally prove the fairness property guaranteed by \fruitchainsnosp. To ease the analysis, we assume $t_r = (1-\delta)n_r$ for all rounds $r$, i.e., the number of adversarial queries in the worst case. We first give a definition of fairness.

\begin{defn}
We say that $\H$ is a $\phi$-fraction honest subset if players in $\H$ (that may change over time) are honest and $n^{\H}_r = \lceil \phi(n_r+t_r) \rceil = \lceil \phi (2-\delta) n_r \rceil$ for round $r$, where $n^{\H}_r$ is number of players in $\H$ at round $r$.
\end{defn}

\begin{defn}[Fairness]
We say the \fruitchains protocol has (approximate) $\sigma-$fairness if for any $0 < \phi <\frac{1}{2-\delta}$ and any $\phi-$fraction honest subset $\H$, there exists $W_0 \in \mathbb{N}$ such that for any honest player holding chain $\C$ at round $r$ and any interval $S_0 \in [0,r]$ with at least $W_0$ consecutive rounds, let $\C(S_0)$ be the segment of $\C$ containing blocks with timestamps in $S_0$ and $d$ be the total difficulty of all fruits included in $\C(S_0)$, then it holds that the fruits included in $\C(S_0)$ mined by $\H$ have total difficult at least $(1-\sigma) \phi d$.
\end{defn}

By the common prefix property proved in Lemma \ref{lem:proposer-prefix}, for a chain $\C$ held by an honest node at round $r$, the prefix $\C^{\lceil \ell+2\Delta}$ are stabilized, so to mine a fruit/block at round $r$ an honest miner will select the tip of $\C^{\lceil \ell+2\Delta}$ as the fruit parent by our proposed variable difficulty \fruitchains algorithm. And we set the recency parameter $R = 3\ell + 7\Delta$, i.e., a fruit $B_f$ is recent w.r.t. a chain $\C$ at round $r$ if the fruit parent of $B_f$ is in $\C$ and has timestamp at least $r-3\ell -7\Delta$. With this selection of the recency parameter, we can prove the following key property of the \fruitchains protocol: any fruit mined by an honest player will be incorporated into the stabilized chain (and thus never lost). We refer to this as the {\it Fruit Freshness Lemma}—fruits stay ``fresh'' sufficiently long to be incorporated. 

\begin{lem}[Fruit Freshness]
\label{lem:fresh}
For a typical execution in a $(\gamma, 2(1 + \delta)\gamma^2 \Phi/f)$-respecting environment, if $R = 3\ell +7\Delta$, then an honest fruit mined at round $r$ will be included into the stabilized chain before round $r+r_{\rm wait}$, where $r_{\rm wait} = 2\ell +5\Delta$.
\end{lem}

\begin{proof}
Suppose an honest fruit $B_f$ is mined at round $r_0$ with block parent being the tip of $\C_0$, then the fruit parent of $B_f$ is the tip of $\C_0^{\lceil \ell+2\Delta}$. Further, by the NoStaleChains property (defined in Appendix \ref{app:backbone}), the fruit parent of $B_f$ has timestamp $r_1 \geq r_0 - 2\ell -4\Delta$. By round $r_0 +\Delta$, all honest nodes will receive $B_f$. Let $r_2 = r_0 + 2\ell +5\Delta$ and $\C$ be any chain held by honest nodes at round $r_2$, then by the chain quality property in Lemma~\ref{lem:proposer-quality}, we know that there exists at least one honest block $B$ on $\C$ whose timestamp $r_3$ is in the interval $(r_2-2\ell-4\Delta, r_2-\ell-2\Delta)$. We check that $B_f$ is still recent at round $r_3$ as
$$r_3 - r_1 < (r_2 - \ell - 2\Delta) - (r_0 - 2\ell -4\Delta) = 3\ell +7\Delta = R.$$
As $B$ is an honest block mined after $r_0+\Delta$, $B$ or an ancestor of $B$ must include $B_f$. And since $B$ is stabilized in $\C$ at round $r_2$, we have that $r_{\rm wait} = r_2 - r_0 = 2\ell +5\Delta$
\end{proof}

Define the random variable $D_r$ equal to the sum of the difficulties of all fruits computed by honest parties at round $r$. And for fixed $\phi$ and $\phi-$fraction honest subset $\H$, define the random variable $D^{\H}_r$ equal to the sum of the difficulties of all fruits computed by parties in $\H$ at round $r$. For a set of rounds $S$, we define $D(S) = \sum_{r\in S}D_r$, $D^{\H}(S)= \sum_{r\in S}D^{\H}_r$, and $n^{\H}(S)= \sum_{r\in S}n^{\H}_r$. For a set of $J$ adversarial queries, define the random variable $A(J)$, as the sum of difficulties of all the fruits created during queries in $J$. 

Next we define the following fruit-typical execution.
\begin{defn}[Fruit-typical Execution]
\label{def:fruit_typicality}
An execution $E$ is fruit-typical if the followings hold

\noindent(a) For any set $S$ of at least $\ell$ consecutive good rounds,
\begin{align*}
     D(S) < (1+\epsilon)pn(S).
\end{align*}

\noindent(b) For any set $S$ of at least $\ell$ consecutive good rounds, let $J$ be the set of adversarial queries in $S$, if we further know each query in $J$ made at round $r$ with target $T$ satisfies $f/2\tau\gamma^3 \leq pn_{r}T \leq (1+\delta)\tau\gamma^3f$, then
\begin{equation*}
     D(S) + A(J) < (1+\epsilon)p(n(S)+|J|).
\end{equation*}

\noindent(c) For any set $S$ of at least $\ell/\phi(2-\delta)$ consecutive good rounds,
\begin{equation*}
    D^{\H}(S) > (1-\epsilon)pn^{\H}(S).
\end{equation*}
\end{defn}

\begin{thm}
For an execution $\mathcal{E}$ of $L$ rounds, in a $(\gamma,s)$-respecting environment, the probability of the event  ``$\mathcal{E}$ not fruit-typical'' is bounded by $\mathcal{O}(L)e^{-\lambda}$.
\end{thm}
\begin{proof}
The proof for (a) is the same as in \cite{full2020}. 
For (b), by the condition, for each query in $S$ (made either by an honest node or the adversary) at round $r$ with target $T$, we have $f/2\tau\gamma^3 \leq p(n_r+t_r)T\leq 2(1+\delta)\tau\gamma^3f$. Therefore, the proof is similar to \cite{full2020} (by setting $\ell$ to be slightly larger as we did in Eqn.(\ref{eqn:ell})).

For (c), let the execution be partitioned into parts such that each part has at least $l/{\phi(2-\delta)}$ rounds and at most $s$ rounds. We will prove that the statement fails with a probability less than $e^{-\lambda}$ for each part. 
Let $J$ denote the queries made by $\mathcal{H}$ in rounds $S$. We have $|J| =n^{\mathcal{H}}(S)= n(S)\phi (2-\delta)$. For $k \in [|J|]$, let $Z_i$ be the difficulty of any block obtained from query $j \in J$. Then
\begin{align*}
    &X_0 = 0 \\
    &X_k = \sum_{i \in [k]} Z_i - \sum_{i \in [k]} \mathbb{E}[Z_i| \mathcal{E}_{i-1}]  
\end{align*}
is a martingale with respect to $\mathcal{E}_0, \ldots, \mathcal{E}_{k}$. We have 
\begin{align*}
   X_k - X_{k-1} &= X_k - \mathbb{E}[X_k | \mathcal{E}_{k-1}] = Z_k - \mathbb{E}[Z_k | \mathcal{E}_{k-1}] \\
                 &\leq \frac{1}{T_k^{min}} = \frac{pn_k}{pn_kT_k^{min}} \leq 2\gamma^3 pn(S)/f|S| := b.
\end{align*}
Similarly 
\begin{align*}
    V &= \sum_k var [X_k - X_{k-1} | \mathcal{E}_{k-1}] \leq \sum_k \mathbb{E}[Z_k^2 | \mathcal{E}_{k-1}] \\
    & = \sum_k p T_k \frac{1}{T_k^2} \leq 2\gamma^3 p^2|J|n(S)/f|S| := v.
\end{align*}
Let the deviation $t = \epsilon p|J| = \epsilon \phi (2 -\delta)p n(S)$, then we have $b = \frac{2 \gamma^3 t}{\epsilon \phi (2-\delta)f |S|}$ and $v = \frac{2\gamma^3t^2}{\epsilon^2\phi (2-\delta) f |S|}$. Using the minimum value of $|S|$ is $l/\phi(2-\delta)$ and applying Theorem \ref{thm:martingale} to $-X_{|J|}$, we have
\begin{align*}
    P[D(S) < (1-\epsilon)pn^{\mathcal{H}}(S)] \leq \exp({-\frac{\epsilon t}{2b(1 + \epsilon/3)}}) \leq \exp({-\lambda}).
\end{align*}
This concludes the proof.

\end{proof}

\begin{thm}[Fairness]
For a typical and fruit-typical execution in a $(\gamma, 2(1 + \delta)\gamma^2 m/f)$-respecting environment, the \fruitchains protocol with recency parameter $R = 3\ell +7\Delta$ satisfies $\sigma-$fairness, where $\sigma = 4\epsilon$.
\end{thm}

\begin{proof}
Fixed $\phi$, a $\phi-$fraction honest subset $\H$, and an honest player holding chain $\C$.  Set $W_0 = \max\{s/2,\ell/\phi(2-\delta)\}+ 3\ell +7\Delta$. Let $S_0 = \{u:r_1\leq u \leq r_2\}$ be a window of at least $W_0$ consecutive rounds. Let $\C(S_0)$ be the segment of $\C$ containing blocks with timestamps in $S_0$, let $\F$ be all fruits included in $\C(S_0)$, and $d$ be the total difficulty of all fruits in $\F$. Then we have the following facts:
\begin{itemize}
    \item {\bf Fact 1.} For any $B_f \in \F$, $B_f$ is mined after $r_1 - 4\ell - 9\Delta$. Indeed by recency condition, $B_f$'s fruit/block parent has timestamp at least $r_1 - R$. By Accuracy property (defined in Appendix~\ref{app:backbone}), $B_f$'s fruit/block parent is mined after $r_1 - R - (\ell + 2\Delta)$. So $B_f$ must be mined after $r_1 - 4\ell - 9\Delta$.
    \item {\bf Fact 2.} For any $B_f \in \F$, $B_f$ is mined before $r_2 + \ell +2\Delta$. Indeed the block including $B_f$ has timestamp at most $r_2$, and by Accuracy property, the block including $B_f$ is mined before $r_2 + \ell +2\Delta$. So $B_f$ must be mined before $r_2 + \ell +2\Delta$.
    \item {\bf Fact 3.} If a fruit $B_f$ is mined by $\H$ after $r_1$ and before $r_2 - 3\ell -7\Delta$, then $B_f \in \F$. Indeed, by NoStaleChain property, the last honest block in $\C(S_0)$ has timestamp at least $r_2 - \ell - 2\Delta$. Hence by Lemma~\ref{lem:fresh}, all honest fruit mined after $r_1$ and before $r_2 - \ell - 2\Delta - r_{\rm wait}$ will be included into a block in $\C(S_0)$.  
\end{itemize} 

Let $S_1 = \{u: r_1-(4\ell + 9\Delta) \leq u \leq r_2 + (\ell +2\Delta)\}$, $S_2 = \{u: r_1 \leq u \leq r_2 - (3\ell +7\Delta)\}$, and $J$ be the set of adversary queries associated with $\F$ in $S_1$. Then by Fact 1 and Fact 2, we have all fruits in $\F$ are mined in $S_1$; by Fact 3, we have all fruits mined by $\H$ in $S_2$ are in $\F$. Also note that for each query in $J$, to mine a recent fruit $B_f$ with target $T$ at round $r$, the timestamp of $B_f$'s fruit parent $B$ must be at least $r-R$. Let $B_0$ with target $T_0$ be the last honest ancestor of $B$, then $B_0$ will have timestamp $r_0 \geq r - R -(\ell+2\Delta)$. By GoodRounds property (defined in Appendix \ref{app:backbone}), we know $f/2\gamma^2 \leq pn_{r_0}T_0 \leq (1+\delta)\gamma^2f$. By the bounded difficulty change rule, we have $T_0/\tau \leq T \leq \tau T_0$. And in a $(\gamma, 2(1 + \delta)\gamma^2 m/f)$-respecting environment, $n_{r_0}/\gamma \leq n_r \leq \gamma n_{r_0}$. Therefore we have $f/2\tau\gamma^3 \leq pn_{r}T \leq (1+\delta)\tau\gamma^3f$.
Further note that, to prove $\sigma$-fairness, it suffices to show that
\begin{equation*}
    D^{\H}(S_2) \geq (1-\sigma)\phi(D(S_1)+A(J)).
\end{equation*}

Under a fruit-typical execution, we have
\begin{equation*}
    D^{\H}(S_2) > (1-\epsilon)pn^{\H}(S_2) \geq (1-\epsilon)\phi(2-\delta) pn(S_2),
\end{equation*}
and
\begin{equation*}
    D(S_1) + A(J) < (1+\epsilon) p(n(S_1)+|J|) = (1+\epsilon) (2-\delta)pn(S_1).
\end{equation*}

By our choice of $W_0$, we have $|S_2| \geq s/2$. Furthermore, we may assume $|S_2| \leq s$. This is because we may partition $S_2$ in parts such that each part has size between $s/2$ and $s$, sum over all parts to obtain the desired bound. Then by Lemma~\ref{lem:envbounds}, we have
\begin{align*}
    n(S_1) &\leq (1+\frac{\gamma|S_1 \setminus S_2|}{|S_2|})n(S_2) \\
    & \leq (1+\frac{\gamma|8\ell + 18\Delta|}{s/2})n(S_2) \\ 
    & < (1+\frac{2\epsilon}{\gamma^3})n(S_2) \\
    & < (1+2\epsilon)n(S_2).
\end{align*}
Finally, by setting $\sigma = 4\epsilon$, we conclude the proof.

\end{proof}
\newgeometry{left=38mm, right=38mm} 
\section{Pseudocode of \scheme}
\label{app:pseudocode}
 
\begin{algorithm}[h]
{\fontsize{8pt}{8pt}\selectfont \caption{Prism: Main}\label{alg:prism_main}
\begin{algorithmic}[1]

\Procedure{Main}{ }
    \State \textsc{Initialize}()
    \While{True}
        \State $header, Ppf, Cpf$ =  \textsc{PowMining}() \label{code:blockBody1}
        \State \maincolorcomment{Block contains header, parent, content and merkle proofs}
        \If{header is a \textit{tx block}}
        \State $block \gets \langle header, txParent, txPool, Ppf, Cpf\rangle$
        \ElsIf{header is a \textit{prop block}}
        \State $block \gets \langle header, prpParent, unRfTxBkPool, Ppf, Cpf\rangle$\label{code:blockBody2}
        \ElsIf{header is a \textit{block in voter} blocktree $i$}
        \State $block \gets \langle header, vtParent[i], votesOnPrpBks[i], Ppf,\label{code:blockBody3} Cpf\rangle$
        \EndIf
        \State \textsc{BroadcastMessage}($block$)  \colorcomment{Broadcast to peers}

    \EndWhile
\EndProcedure

\Procedure{Initialize}{ } \colorcomment{ All variables are global}
\vspace{0.5mm} \State \maincolorcomment{Blockchain data structure $C  = (prpTree, vtTree) $}
\State $prpTree \gets genesisP$ \label{code:prpGenesis} \colorcomment{Proposer Blocktree}
\For{$i \gets 1 \; to \; m$}
\State $vtTree[i] \gets genesisM\_i$ \colorcomment{Voter $i$ blocktree} \label{code:voterGenesis}
\EndFor
\State \maincolorcomment{Parent blocks to mine on} \label{code:VarprpParent}
\State $prpParent$ $\gets genesisP $ \colorcomment{ Proposer block to mine on} 
\For{$i \gets 1 \; to \; m$}
\State $vtParent[i]$ $\gets genesisM\_i $ \colorcomment{Voter tree $i$ block to mine on}
\EndFor \label{code:VarvtParent}
\vspace{00.4mm}\State \maincolorcomment{Block content} \label{code:allcontent}
\State  $txPool$ $\gets \phi$ \colorcomment{Tx block content: Txs to add in tx bks} \label{code:txblockcontent}
\State $unRfTxBkPool$ $ \gets \phi$ \colorcomment{Prop bk content1: Unreferred tx bks} \label{code:prblockcontent1}
\State $unRfPrpBkPool$ $ \gets \phi$ \colorcomment{Prop bk content2: Unreferred prp bks} \label{code:prblockcontent2}

\For{$i \gets 1 \; to \; m$}
\State $votesOnPrpBks(i) \gets \phi$ \colorcomment{Voter tree $i$ bk content } \label{code:vtblockcontent}
 
\EndFor

\EndProcedure
\end{algorithmic}
}
\end{algorithm}

\begin{algorithm}[h]
{\fontsize{8pt}{8pt}\selectfont \caption{Prism: Mining}\label{alg:prism_minining}
\begin{algorithmic}[1]

\Procedure{PowMining}{ }
\While{True}
\State  $txParent \gets prpParent$ \label{code:initParent}
\State \maincolorcomment{Assign content for all block types/trees}
\For{$i \gets 1 \;  to \; m$} $vtContent[i] \gets votesOnPrpBks$[i] \label{code:voteIneffcient}
\EndFor
\State $txContent \gets txPool$
\State $prContent \gets (unRfTxBkPool, unRfPrpBkPool)$
\vspace{00.4mm}\State \maincolorcomment{Define parents and content Merkle trees}
\State $parentMT\gets$MerklTree($vtParent,txParent,prpParent$) 
\State $contentMT\gets$MerklTree($vtContent,txContent,prContent$) \label{code:endContent}
\State nonce $ \gets $ RandomString($1^\kappa$)
\State \maincolorcomment{Header is similar to Bitcoin}
\State header $\gets \langle$ $parentMT.$root, $contentMT.$root, nonce $\rangle$
\If{{chainLength($prpParent$) \% e == 0 }} \label{code:difficult}
    \State{$f_p^{new} \gets$ \textsc{RecalculateTarget}($f_p$)}
    \State{$f_v \gets$ ($f_v * f_p^{new} / f_p$)}
    \State{$f_t \gets$ ($f_t * f_p^{new} / f_p$)}
    \State{$f_p \gets$ $f_p^{new}$}
\EndIf
\vspace{00.4mm} \State \maincolorcomment{Sortition into different block types/trees } \label{code:sortitionStart}
\If {Hash(header)  $\leq mf_v$} \colorcomment{Voter block mined}
    \State $i\gets  \lfloor $Hash(header)/$f_v\rfloor$
    \textbf{and} \textit{break} \colorcomment{on tree $i$ }
\ElsIf{ $mf_v < $Hash(header)  $ \leq mf_v+f_t$} 
    \State  $i \gets m+1$ \textbf{ and} \textit{break}\colorcomment{Tx block mined}
\ElsIf{ $mf_v+f_t< $ Hash(header) $\leq mf_v+f_t+f_p$} 
    \State  $i \gets m+2$ \textbf{ and} \textit{break}\colorcomment{Prop block mined}
\EndIf
\EndWhile \label{code:sortitionEnd}
\maincolorcomment{Return header along with Merkle proofs}
\State \Return $\langle header, parentMT.$proof($i$),  $contentMT.$proof($i) \rangle$ \label{code:minedBlocm}
\EndProcedure
\end{algorithmic}
}
\end{algorithm}

\begin{algorithm}[h]
{\fontsize{8pt}{8pt}\selectfont \caption{Prism: Block and Tx handling}\label{alg:prism_handling}
\begin{algorithmic}[1]
\Procedure{ReceiveBlock}{\textsf{B}} \colorcomment{Get block from peers}    
\If {\textsf{B} is a valid \textit{transaction block}} \label{code:txblk_r1}
    \State $txPool$.removeTxFrom(\textsf{B}) \label{code:updateTxContent} 
    \State $unRfTxBkPool$.append(\textsf{B})\label{code:txblk_r2}

\ElsIf{\textsf{B} is a valid \textit{block on $i^{\text{th}}$ voter tree} 
\textbf{and} \textsc{ValidVote}(\textsf{B},i)} \label{code:vtblk_r1}
    \State $vtTree[i]$.append(\textsf{B}) \textbf{and} $vtTree[i]$.append(\textsf{B}.ancestors())
    \State \maincolorcomment{A vote is a range of difficulty along with the the corresponding proposer block}
    \If{{\textsf{B}.chaindiff $> vtParent[i]$.chaindiff}} 
    \State $vtParent[i] \gets \textsf{B} $ \textbf{and} $votesOnPrpBks$($i$).update(\textsf{B}) \label{code:updateVoteMine}
    \EndIf\label{code:vtblk_r2}

\ElsIf{\textsf{B} is a valid \textit{prop block}} \label{code:prpblk_r1}
    \If{\textsf{B}.diff $>prpParent$.diff}
    \State $prpParent \gets \textsf{B}$ \label{code:updatePropToMine}
    \For{$i \gets 1 \;  to \; m$}  
    \colorcomment{Add vote on level $\ell$ on all $m$ trees}
    \State $votesOnPrpBks($i$)[\textsf{B}.level] \gets \textsf{B}$   \label{code:addVote} 
    \EndFor
    \ElsIf{\textsf{B}.level $ > prpParent$.level+$1$}
        \State \maincolorcomment{Miner doesnt have block at level $prpParent$.level+$1$}
        \State \textsc{RequestNetwork(\textsf{B}.parent)} 
    \EndIf
    \State  $prpTree[\textsf{B} $.level].append(\textsf{B}),\; $unRfPrpBkPool$.append(\textsf{B})
    \State $unRfTxBkPool$.removeTxBkRefsFrom(\textsf{B}) \label{code:updateTxBkContent}
    \State $unRfPrpBkPool$.removePrpBkRefsFrom(\textsf{B}) \label{code:updatePrpBkContent}

\EndIf
\vspace{1mm}
\EndProcedure

\Procedure{ReceiveTx}{\textsf{tx}} 
\If {\textsf{tx} has valid signature }   $txPool$.append(\textsf{B})
\EndIf    
\EndProcedure
\end{algorithmic}
}
\end{algorithm}

\begin{algorithm}[h]
{\fontsize{8pt}{8pt}\selectfont \caption{Prism: Vote validation}\label{alg:prism_fucntion}
\begin{algorithmic}[1]

\Procedure{ValidVote}{\textsf{B},i } \colorcomment{ validate a vote}
\State \maincolorcomment{voter block can't vote for difficulty grater than its proposer parent}
\If {\textsf{B}.vtContent[i].latestBlock.chaindiff $>$ \textsf{B}.prpParent.chaindiff}
\State \Return False
\EndIf
\If {\textsf{B}.vtContent[i] has discontinuous votes}
\State \Return False
\EndIf
\If {\textsf{B}.vtContent[i].earliestBlock.parent.chaindiff  $>$ \textsf{B}.vtParent[i].chaindiff}
\State \Return False
\EndIf
\State \maincolorcomment{include the check where the difficulty ranges  of the votes should end at proposal blocks}
\State \Return True
\vspace{0.5mm}
\EndProcedure


\end{algorithmic}
}
\end{algorithm}

\begin{algorithm}[h]
{\fontsize{8pt}{8pt}\selectfont \caption{Prism: Tx confirmation}\label{alg:prism_con}
\begin{algorithmic}[1]

\Procedure{IsTxConfirmed}{$tx$}\label{code:fastConf}
    \State $\Pi \gets \phi$         \colorcomment{Array of set of proposer blocks}
    \For{$\ell \gets 1 \;  to \; prpTree.$maxLevel}  
        \State $votesNdepth \gets \phi$ 
        \For{$i$ in $1\;to\;m$} \label{code:getvotes_start}
        \State  $votesNdepth[i] \gets \textsc{GetVoteNDepth}(i, \ell)$ 
        \EndFor
            \If{IsPropSetConfirmed($votesNdepth$)}
        \State $\Pi[\ell] \gets$ {GetProposerSet}$(votesNdepth)$
        \Else $\;$ \textbf{break}\label{code:appendPrpList}
        \EndIf
    \EndFor
    \State \maincolorcomment{Ledger list decoding: Check if tx is confirmed in all ledgers}
    \State $prpBksSeqs \gets \Pi[1]\times \Pi[2]\times\cdots\times\Pi[\ell]$  \colorcomment{outer product} \label{code:outerproduct}
\For{$prpBks$ in $prpBksSeqs$} \label{code:ledgers_for_start}
                \State $ledger$ = \textsc{BuildLedger}($prpBks$)
    \If{$tx$ is \textbf{not confirmed} in \av{ledger}} \Return False
    \EndIf
    \EndFor
    \Return True \colorcomment{Return true if tx is confirmed in all ledgers} \label{code:fastConfirmTx}
\EndProcedure
\vspace{1mm}
\State \maincolorcomment{Return the vote of voter blocktree $i$ at level $\ell$ and  depth of the vote}
\Procedure{GetVoteNDepth}{$i, d$}
    \State $voterMC \gets vtTree[i].HeaviestChain()$
    \For{$voterBk$ in $voterMC$} \label{code:voteCountingStart}
        \For{$vote$ in $voterBk$.votes}
            \If{$d$ in  $vote.range$} 
                    \State \maincolorcomment{Depth is the difficulty of children bks of voter bk on main chain}
\State  \Return ($vote.prpBk$, $voterBk$.\text{depth})  \label{code:voteCountingEnd}
            \EndIf
        \EndFor
    \EndFor
\EndProcedure
\vspace{1mm}
\Procedure{BuildLedger}{\av{propBlocks}} \colorcomment{Input: list of prop blocks}\label{code:buildLedgerFunc}
\State \av{ledger} $\gets []$ \colorcomment{List of valid transactions}
\For{\av{prpBk} in \av{propBlocks}} 
    \State $refPrpBks \gets prpBk.$getReferredPrpBks() \label{code:epochStarts}
    \State \maincolorcomment{Get all directly and indirectly referred transaction blocks.}
    \State  $txBks \gets\;  ${GetOrderedTxBks}$(prpBk, refPrpBks)$ \label{code:txbksOrdering}
    \For{\av{txBk} in \av{txBks}}
        \State $txs \gets txBk$.getTxs() \colorcomment{Txs are ordered in \av{txBk}}
        \For{$tx$ in $txs$}
        \State \maincolorcomment{Check for double spends and duplicate txs} 
            \If{$tx$ is  \textbf{valid} w.r.t to \av{ledger}} \av{ledger}.append($tx$) \label{code:sanitizeLedger}
            \EndIf
        \EndFor
    \EndFor    
\EndFor
\State \Return \av{ledger}
\EndProcedure
\vspace{1mm}
\State \maincolorcomment{Return ordered list of confirmed transactions}
\Procedure{GetOrderedConfirmedTxs()}{} \label{code:slowConf}
   \State $\texttt{L} \gets \phi$          \colorcomment{Ordered list of leader blocks}
    \For{$ prpBk$ in $propBlocks$}
        \State $g(p) = inf_d(d: \textsc{GetLeader}(d) = p )$
    \EndFor
    \State $L \gets sort(p, key=g(p))$
    \State\Return \textsc{BuildLedger}(\texttt{L}) \label{code:slowBuildLedger}
    \EndProcedure

\end{algorithmic}
}
\end{algorithm}

\clearpage
\section{Pseudocode of \ohie}
\label{app:pseudocode_ohie}

\begin{algorithm}[h] 
{\fontsize{8pt}{8pt}\selectfont \caption{OHIE: Mining}
\label{alg:ohie_minining}
\begin{algorithmic}[1]

\State $T \leftarrow \textsc{Initial~Difficulty}$; 
\State $V_i \leftarrow$ \{(genesis block of chain $i$, the attachment for genesis block of chain $i$)\}, for $0\le i\le m-1$;
\State $M \leftarrow $ Merkle tree of the hashes of the $m$ genesis blocks;
\State ${\tt trailing}$ $\leftarrow $ hash of genesis block of chain $0$;

\State
\Procedure{OHIE}{ } 
\State \hspace*{2mm}{\bf repeat forever} \{
\State \hspace*{6mm}ReceiveState(); \label{code:mainloop1}
\State \hspace*{6mm}Mining();
\State \hspace*{6mm}SendState(); \label{code:mainloop2}
\State \hspace*{2mm}\}
\EndProcedure
\State

\Procedure{Mining}{} 
\State \hspace*{2mm}$B.{\tt transactions} \leftarrow$ get\_transactions();
\State \hspace*{2mm}$B.{\tt root}$ $\leftarrow$ root of Merkle tree $M$;
\State \hspace*{2mm}$B.{\tt trailing} \leftarrow$ ${\tt trailing}$;
\State \hspace*{2mm}$B.{\tt nonce} \leftarrow$  new\_nonce();
\State \hspace*{2mm}$\widehat{B}.{\tt hash} \leftarrow hash(B)$;
\State \hspace*{2mm}{\bf if} ($\widehat{B}.{\tt hash} < T$) \{
\State \hspace*{6mm}$i$ $\leftarrow$ last $\log_2 m$ bits of $\widehat{B}.{\tt hash}$;
\State \hspace*{6mm}$\widehat{B}.{\tt leaf} \leftarrow$ leaf $i$ of $M$;
\State \hspace*{6mm}$\widehat{B}.{\tt leaf\_proof} \leftarrow M$.MerkleProof($i$);
\State \hspace*{6mm}\red{$\widehat{B}.{\tt chain0\_parent} \leftarrow$ leaf $0$ of $M$;}
\State \hspace*{6mm}\red{$\widehat{B}.{\tt parent\_proof} \leftarrow M$.MerkleProof($0$);}
\State \hspace*{6mm}ProcessBlock($B$, $\widehat{B}$);
\State \hspace*{2mm}\}
\EndProcedure
\State

\Procedure{SendState}{} 
\State \hspace*{2mm} send $V_i$ ($0\le i\le m-1$) to other nodes; \colorcomment{In implementation, only need to send those blocks not sent before.}
\EndProcedure

\Procedure{ReceiveState}{} 
\State \hspace*{2mm} {\bf foreach} $(B, \widehat{B})\in$ received state {\bf do}
\State \hspace*{6mm} ProcessBlock($B$, $\widehat{B}$);
\State \hspace*{2mm}\maincolorcomment{a block $B_1$ should be processed before $B_2$, if $\widehat{B_2}.{\tt leaf}$, $\widehat{B_2}.{\tt chain0\_parent}$ or $\widehat{B_2}.{\tt trailing}$ points to $B_1$.}
\EndProcedure

\end{algorithmic}
}
\end{algorithm}

\begin{algorithm}[h] 
{\fontsize{8pt}{8pt}\selectfont \caption{OHIE: ProcessBlock}
\label{alg:ohie_processing}
\begin{algorithmic}[1]

\Procedure{ProcessBlock}{$B$, $\widehat{B}$} 
\State \hspace*{2mm} \maincolorcomment{do some verifications}
\State \hspace*{2mm}$i$ $\leftarrow$ last $\log_2 m$ bits of $\widehat{B}.{\tt hash}$;
\State \hspace*{2mm}\red{verify that $\widehat{B}.{\tt chain0\_parent}$ is leaf $0$ in the Merkle tree, based on $B.{\tt root}$ and $\widehat{B}.{\tt parent\_proof}$;}
\State \hspace*{2mm}\red{verify that $\widehat{B}.{\tt chain0\_parent} = \widehat{D}.{\tt hash}$ for some block $(D, \widehat{D}) \in V_0$;}
\State \hspace*{2mm}\red{$T$ $\leftarrow$ mining difficulty of $D$'s next block;}
\State \hspace*{2mm}verify that $\widehat{B}.{\tt hash} < T$;
\State \hspace*{2mm}verify that $hash(B) = \widehat{B}.{\tt hash}$;
\State \hspace*{2mm}verify that $\widehat{B}.{\tt leaf}$ is leaf $i$ in the Merkle tree, based on $B.{\tt root}$ and $\widehat{B}.{\tt leaf\_proof}$;
\State \hspace*{2mm}verify that $\widehat{B}.{\tt leaf} = \widehat{A}.{\tt hash}$ for some block $(A, \widehat{A}) \in V_i$;
\State \hspace*{2mm}verify that $B.{\tt trailing} = \widehat{C}.{\tt hash}$ for some
block $(C, \widehat{C}) \in \cup_{j=0}^{m-1}V_j$;
\State \hspace*{2mm}\red{$\widehat{B_0}$ $\leftarrow$ $\widehat{B}$'s parent block on chain $i$;}
\State \hspace*{2mm}\red{$\widehat{D_0}$ $\leftarrow$ $\widehat{B_0}.{\tt chain0\_parent}$;}
\State \hspace*{2mm}\red{verify that chaindiff($\widehat{D}) \geq$ chaindiff($\widehat{D_0}$);}
\State \hspace*{2mm}{\bf if} (any of the above 8 verifications fail) {\bf then return};
\State
\State \hspace*{2mm} \maincolorcomment{compute ${\tt rank}$ and ${\tt next\_rank}$ values}
\State \hspace*{2mm}$\widehat{B}.{\tt rank} \leftarrow \widehat{A}.{\tt next\_rank}$;
\label{code:rankstart}
\State \hspace*{2mm}$\widehat{B}.{\tt next\_rank} \leftarrow \widehat{C}.{\tt next\_rank}$;
\State \hspace*{2mm}{\bf if} ($\widehat{B}.{\tt next\_rank}\le \widehat{B}.{\tt rank}$) {\bf then} $\widehat{B}.{\tt next\_rank}\leftarrow \widehat{B}.{\tt rank}+\red{{\rm diff}(\widehat{B})}$;
\label{code:rankend}
\State
\State \hspace*{2mm} \maincolorcomment{update local data structures}
\State \hspace*{2mm}$V_i \leftarrow V_i \cup \{(B, \widehat{B})\}$;
\State \hspace*{2mm}update ${\tt trailing}$;
\State \hspace*{2mm}update Merkle tree $M$;
\State \hspace*{2mm}{\bf if} \red{(chainLength(chain 0) \% e == 0) }  {\bf then} \red{$T \gets$ \textsc{RecalculateTarget}($T$)}
\colorcomment{update mining difficulty}
\EndProcedure
\end{algorithmic}
}
\end{algorithm}

\begin{algorithm}[h] 
{\fontsize{8pt}{8pt}\selectfont \caption{OHIE: GenerateSCB}
\label{alg:ohie_scb}
\begin{algorithmic}[1]
\Procedure{OutputSCB}{}
\State \hspace*{2mm} \maincolorcomment{determine partially-confirmed blocks and ${\tt confirm\_bar}$}
\State \hspace*{2mm}{\bf for} ($i=0$; $i <m$; $i$++) \{
\State \hspace*{6mm}$W_i$ $\leftarrow$ \red{get\_heaviest\_path($V_i$);}
\State \hspace*{6mm}${\tt partial}_i$ $\leftarrow$ blocks in $W_i$ that are partially-confirmed;
\State \hspace*{6mm}($B_i$, $\widehat{B_i}$)  $\leftarrow$ the last block in ${\tt partial}_i$;
\State \hspace*{6mm}$y_i\leftarrow \widehat{B_i}.{\tt next\_rank}$;
\State \hspace*{2mm}\}
\State \hspace*{2mm}${\tt all\_partial}$ $\leftarrow$ $\cup_{i=0}^{m-1} {\tt partial_i}$;
\State \hspace*{2mm}${\tt confirm\_bar} \leftarrow \min_{i=0}^{m-1} y_i$;
\State
\State \hspace*{2mm} \maincolorcomment{determine fully-confirmed blocks and SCB}
\State \hspace*{2mm}$L \leftarrow \emptyset$;
\State \hspace*{2mm}{\bf foreach} $(B,\widehat{B})\in {\tt all\_partial}$ \{
\State \hspace*{6mm}{\bf if} ($\widehat{B}.{\tt rank} < {\tt confirm\_bar}$)
{\bf then} $L \leftarrow L \cup \{(B, \widehat{B})\}$;
\State \hspace*{2mm}\}
\State \hspace*{2mm}sort blocks in $L$ by ${\tt rank}$, tie-breaking favoring smaller chain id;
\State \hspace*{2mm}{\bf return} $L$;
\EndProcedure

\end{algorithmic}
}
\end{algorithm}

\restoregeometry 
\end{document}